\documentclass[journal]{IEEEtran}
\usepackage{graphicx}
\usepackage{epsfig,rotating}
\usepackage{subfigure}
\usepackage{tabularx}
\usepackage{algorithmic}
\usepackage{algorithm}
\usepackage{amssymb,amsmath}
\usepackage{psfrag}
\usepackage{multirow}
\usepackage{float}
\usepackage{setspace}
\usepackage{cite}     
\usepackage{url}     
\usepackage{epstopdf}
\usepackage{amsthm}
\newtheorem{mydef}{Definition}
\newtheorem{mythe}{Theorem}
\newtheorem{mylem}{Lemma}

\newtheorem{myrem}{Remark}

\begin{document}

\title{Multi-Objective Framework for Dynamic Optimization of OFDMA Cellular Systems}

\author{Prabhu~Chandhar,~\IEEEmembership{Member,~IEEE,} and Suvra~Sekhar~Das,~\IEEEmembership{Member,~IEEE}
 \thanks{Prabhu Chandhar is with the Division of Communication Systems, Link\"{o}ping University, Sweden. Suvra Sekhar Das is with the Indian Institute of Technology Kharagpur, West Bengal, India - 721302.
(email: prabhu.c@liu.se, suvra@gssst.iitkgp.ernet.in). A part of this work was presented at IEEE ICC'14, Sydney, Australia.}
} 
\maketitle

\begin{abstract}\label{abstract}
Green cellular networking has become an important research area in recent years due to environmental and economical concerns. Switching off under-utilized base stations (BSs) during off-peak traffic load conditions is a promising approach to reduce energy consumption in cellular networks. In practice, during initial cell planning, the BS locations and Radio Access Network (RAN) parameters (BS transmit power, antenna height and antenna tilt) are optimized to meet the basic system design requirements like coverage, capacity, overlap, QoS etc. As these metrics are tightly coupled with each other due to co-channel interference, switching off certain BSs may affect the system requirements. Therefore, identifying a subset of large number of BSs which are to be put into sleep mode, is a challenging dynamic optimization problem. In this work, we develop a multi-objective framework for dynamic  optimization framework for Orthogonal Frequency Division Multiple Access based cellular systems. The objective is to identify the appropriate set of active sectors and RAN parameters that maximize coverage and area spectral efficiency while minimizing overlap and area power consumption without violating the QoS requirements for a given traffic demand density. The objective functions and constraints are obtained using appropriate analytical models which capture the traffic characteristics, propagation characteristics (path-loss, shadowing, and small scale fading) as well as load condition in neighbouring cells. A low complexity evolutionary algorithm is used for identifying the global Pareto optimal solutions at a faster convergence rate. The inter-relationships between the system objectives are studied and guidelines are provided to find an appropriate network configuration that provides the best achievable trade-offs. The results show that using the proposed framework, significant amount of energy saving can be achieved and with a low computational complexity while maintaining good trade-offs among the other objectives. 
\end{abstract}


\IEEEpeerreviewmaketitle

\section{Introduction}
\IEEEPARstart{I}{n} recent years, the traffic demand in mobile networks has been growing exponentially due to the evolution of smart phones, applications such as web browsing, video streaming etc., and number of subscriptions. It has been estimated that by 2020, there would be 10 billion mobile devices thereby resulting an 11 fold increase in capacity compared with what they are experiencing today \cite{cisco2012}. To cater to the increased service requirements, a large number of base stations (BSs) are being deployed. Recent studies show that the number of BSs worldwide has doubled from 2007 to 2012, and the number of BSs as of today has reached more than 4 million \cite{GSMA,correia2010}. As a result, the energy consumption in cellular networks increases tremendously leading to an increase in the carbon footprint which leads to global warming \cite{correia2010}. A recent study show that by the end of 2012, the amount of CO$_2$ emissions from BS towers has reached 78 million tons \cite{wu2015}. Reducing energy consumption is an important concern for network operators as it leads to lower Operational Expenditure (OPEX) costs \cite{kyuho2011}.

Several approaches are being considered by the cellular operators to reduce energy consumption both at a component and at a network level \cite{correia2010}. There have been several international research projects (EARTH \cite{EARTH}, OPERANET \cite{OPERANET}, and eWin \cite{eWIN}) that are being carried out to improve energy efficiency in wireless systems. It is seen that 60-80\% of the power consumption in cellular networks is at the Radio Access Network (RAN), mainly at the BSs \cite{hasan2011,marsan2009}. At component level, using advanced design of power amplifiers, reconfigurable circuits, and downlink Discontinuous Transmission (DTX) techniques, only a small amount of energy saving (ES) can be achieved \cite{daquan2013}. However, at the network level, significant energy saving (ES) can be achieved by efficient design of network during planning and management phases \cite{correia2010,richter2009}. 

During cell dimensioning phase, the BSs are usually deployed considering peak traffic demand and future traffic growth. However, studies show that most of the time the BSs are largely underutilized due to spatial and temporal variations in traffic conditions \cite{lte_energyhuawei10,daquan2013}. It is studied that about {{$30\%$}} of time in a day, the traffic is below {{$10\%$}} of the peak \cite{kyuho2011}. Power consumption at the BS is composed of two parts: fixed and dynamic. Fixed power consumption is due to cooling, signal processing etc., and is independent of traffic load. It has been seen that the fixed part constitutes {{$50\%$}} of the total power consumption at full load condition \cite{ferling2010}. Dynamic power consumption is due to the RF transmission and varies with traffic load. During low traffic load conditions, operating under-utilized BSs leads to severe degradation of Area Energy Efficiency (AEE) [number of bits transmitted per Joule per unit area (unit: bits/Joule/m$^2$)] performance due to fixed part of power consumption. Therefore, during off-peak hours, putting under-utilized BSs to sleep mode (SLM) is considered as a promising approach for potential energy savings \cite{correia2010,daquan2013}, thereby increasing AEE performance. In SLM, the underutilized BS is switched OFF and the users associated with it are attached to neighbouring BSs. Moreover, for the purpose of activating/deactivating the BSs, distributed and centralized ES functionalities have been included in 3rd Generation Partnership Project (3GPP) standard \cite{3gpp.32.551}. 

\subsection{Consequences of Putting BSs into Sleep Mode}
Although putting BSs into SLM reduces energy consumption significantly, it has some practical implications. In current cellular systems the BSs are equipped with several (typically three) directional antennas, each radiating within a specified orientation (sector). Therefore, during cell dimensioning phase, the BS locations and RAN parameters (transmit power, antenna height, tilt angle)  are optimized in order to meet various system performance metrics such as coverage, capacity, overlap, Quality of Service (QoS) etc \cite{lieska2002,calegari1997,lieska1998,meunier2000}. Switching OFF certain number of BSs may affect these basic system requirements. For instance, it may result in coverage holes. The interference pattern and Signal to Interference-plus-Noise Ratio (SINR) experienced by the users may be deteriorated due to change in network configuration. Therefore, proper coverage estimation has to be performed before selecting BSs for SLM. 

In order to support hand-overs, it is required to ensure reception of signals from more than one BS \cite{taiPo1994} particularly for the users at the cell edge regions. This means that some amount of overlapping of the coverage region of remaining active BSs has to be ensured. At the same time an excess amount of overlap may significantly increase inter-cell interference (ICI) from the neighbouring cells especially in single frequency networks (SFNs) such as Long Term Evolution (LTE) and Worldwide Interoperability for Microwave Access (WiMAX). This increase in ICI would degrade the SINR experienced by the users. As a consequence, the capacity may reduce significantly as the capacity of a cell depends on the SINR experienced by the users in that cell. Therefore, an excess amount of overlap should be controlled in order to keep ICI within the acceptable level. Further, after switching off certain BSs, the remaining set of active cells should not be overloaded as it would lead to increased blocking (or poor grade of service) in those cells. At the same time, the remaining active BSs should have sufficient resources to serve the offered traffic demand.

Moreover, it is also important to take into account the dynamics of ICI due to switching OFF certain BSs as it has direct impact on the performance metrics. For example, SINR is affected by the radio resource occupancy in neighbouring cells due to varying co-channel interference pattern. The radio resource occupancy in the neighbouring cells is dependent upon the offered traffic load in those cells. After SLM, with the new set of active BSs, the non-uniform RAN parameter configurations and inhomogeneous traffic load lead to non-uniform interference scenarios across the network. Therefore, better estimate of system performance metrics such as coverage and capacity is required to exploit the full potential of the radio resources.

Overall, the subset of BSs to be put into SLM should be identified such that it meets the following requirements.
\begin{itemize}
\item The remaining active BS set provides sufficient coverage, capacity, and overlap requirements with minimal energy consumption. However, the above mentioned performance metrics are tightly coupled with each other due to complex interference patterns. Therefore, improving (or reducing) one metric may end up in degrading the performance of the others. Hence, while finding the appropriate BS subset for BS SLM, it is important to consider the trade-offs among these performance metrics. 
\item Traffic reorientation is required without major changes in blocking or cell overloading to ensure QoS or capacity is not affected. 
\item Along with the new set of BSs, it is required to reconfigure the RAN parameters (BS transmit power, BS antenna height, and antenna's vertical tilt angle) as well. 
\item It is also required to study the inter-relationships between different performance metrics in order to utilize the energy efficiency gains due to BS SLM in current and future cellular systems. 
\item Finding an optimal subset of BSs and RAN parameters with good trade-offs between the system performance metrics is a challenging optimization problem \cite{kyuho2011} because it involves huge search space due to a large scale of mobile networks.
\item Since the requirements of the network differs with fast as well as slow traffic variations \cite{correia2010,aliu2013}, the solution approach should provide appropriate solutions at faster time scales with minimal computational complexity. 
\end{itemize}

In the line of above discussion, in this work, we develop a multi-objective optimization framework for efficient utilization of network resources through BS SLM. The motivation for multi-objective optimization is discussed in the next section.

\subsection{Need for Multi-Objective Optimization in BS Sleep Mode}
The problem of radio network planning for 2G and 3G cellular systems considering multi-objective optimization has been addressed in \cite{lieska2002,calegari1997,lieska1998,meunier2000}. In \cite{lieska2002}, during initial coverage and capacity planning, multi-objective optimization framework is used to find the BS locations, transmission power of BSs and channels per cell. In \cite{calegari1997}, the  BS locations required to cover a given geographical region is obtained using multi-objective optimization for Universal Mobile Telecommunications Systems (UMTS) networks. In \cite{lieska1998}, the problem of maximizing coverage with minimum number of BSs is addressed. In \cite{meunier2000}, a multi-objective framework is proposed for optimizing the number of sites, traffic with minimal interference in radio network design. Since, identifying BS locations is a difficult combinatorial optimization problem, the solutions are obtained using meta-heuristic algorithms such as Genetic Algorithm (GA). In \cite{gordejuela2009}, a meta-heuristic Tabu search multi-objective optimization framework is developed for WiMAX networks to find appropriate BS locations taking into consideration the coverage, interference, and cost criteria. In \cite{lee2000,hurley2002}, a method based on Tabu search based method is used to find the BS locations during initial cell planning considering different objectives such as coverage, site cost, interference, and handover. However, the above mentioned works do not consider energy saving as a design criteria during initial cell planning.

As it is required to find new set of active BSs for different traffic load conditions, the BS SLM approach can also be seen as the \textit{dynamic network planning}. Therefore, it is required to consider multiple system objectives such as coverage, capacity, etc., as is done during initial cell planning. During initial cell planning, as detailed in the above mentioned works, the objective is to find the \textit{location of BSs} that provides maximum coverage, capacity and minimum interference for a peak traffic condition. 
Whereas, the BS SLM problem is different from the initial cell planning in the sense that in dynamic network planning, it is required to find the subset of BSs ($\mathcal{B}_{\mathrm{on}}$) and RAN parameters from an already deployed BS set ($\mathcal{B}$). In this work, we treat the problem as a dynamic optimization various system objectives according to spatial and temporal traffic variations. The related works are detailed in the next section.

\subsection{Related Works}
There have been several studies concerning ES using BS SLM. A recent survey article \cite{wu2015} lists the relevant works on BS SLM. In \cite{marsan2009}, the authors have studied the amount of possible ES through BS SLM. In \cite{tsilimantos2013}, the authors have studied ES with BS SLM in Orthogonal Frequency Devision Multiple Access (OFDMA) based cellular networks using stochastic geometry. In \cite{eunsung2013}, a distributed switching ON/OFF scheme based on the mean and variance of traffic profiles is proposed. In \cite{chiaraviglio2012}, a GA based algorithm is proposed with a single objective of minimizing the number of BSs. In \cite{han2013}, the authors have proposed a BS switching off strategy where the set of BSs are activated from the predefined network pattern according to the offered traffic load. In their work, channel outage probability and call blocking probability are used as the QoS metrics. In \cite{kyuho2011,dufková2010}, greedy-style heuristic algorithms are proposed to find the user association and the set of BSs which need to be switched ON/OFF and analyzed trade-off between energy and delay considering elastic traffic. In \cite{ghosh2012}, the possibility of ES through site-level SLM in 3GPP-LTE networks is studied using system level simulations. In \cite{das2013_es_ga}, Genetic Algorithm (GA) based approach is proposed to identify the sites for SLM. 

Most of the above literature on ES through BS SLM have focused on the single objective of minimizing energy consumption without considering consequences on other metrics such as network coverage, overlap, and QoS degradation etc. The problem is more critical in 4G and beyond cellular systems which employ single frequency reuse. In \cite{gonzalez2014}, a multi-objective optimization framework for cell switch-off is proposed to minimize energy consumption while maximizing system capacity. 

Further, the dynamics of ICI due to SLM has not been addressed in the above mentioned works. For example, in \cite{eunsung2013,han2013}, ICI has been assumed as a static Gaussian-like noise. This is not a valid assumption, because after deactivating certain number of BSs, the service area covered by an individual active BSs would be different. So, the SINR experienced by the users attached to those BSs would also be different due to change in ICI pattern. 

The work presented in this paper is an extension of our previous work \cite{prabhu2014icc}. In \cite{prabhu2014icc}, we have proposed a framework for site-level SLM without RAN parameter optimization. In this work, we present sector-level SLM along with RAN parameter optimization. Further, we used an approach called sum of weighted objectives (SWO), to solve the framed multi-objective optimization problem. In SWO, all the objectives are added together with appropriate weight values to form a single objective. The drawback of this approach is that in order to obtain the best solution one has to find the appropriate weight vector through trial and error method which may take several trails. In this work we use a different approach, where initially the set of Pareto optimal solutions are obtained and then the final solution is selected based on the preferences.

\subsection{Proposed Multi-Objective Optimization Framework for Dynamic Network Planning}
In this work, we develop a multi-objective optimization framework to identify an appropriate network configuration (set of BSs and RAN parameters) that maintains good trade-off between various system performance metrics at different traffic load conditions. The main motivation for the multi-objective optimization framework is that it provides network providers a clear idea about the different options and solutions to achieve optimal network performance. Further, the proposed framework provides flexibility in choosing the objective functions and control variables according to network conditions. 

\subsubsection{System Modeling Approach}
In this work, we model interference between different cells by statistical approach in order to capture the characteristics of propagation conditions such as shadowing and traffic conditions such as arrival rate. In practice, the network setup has to be reconfigured for every few minutes to few hours according to traffic variations. During the reconfiguration interval, one has to capture the variations in ICI which occurs due to user movements and call arrivals. The previous works considered fixed user locations for computing system performance metrics. Computing the system metrics considering fixed user locations may not be able to find the appropriate network configurations. A possible approach is to perform Monte-Carlo simulations. However, it requires high computational complexity and large number of trails to obtain the solutions. As it is required to obtain the solutions in a smaller time scales, Monte-Carlo simulation based solutions may not be a suitable approach. The statistical approach helps to accurately compute the system metrics by capturing spatial and temporal variations of the parameters with less computational burden and processing time. 

\subsection{Contributions}
The contributions of the paper are summarized as follows:
\begin{itemize}
\item Unlike previous works on BS SLM, we consider the following four important system objectives: area power consumption (APC) minimization, area spectral efficiency (ASE) maximization, coverage maximization, and overlap minimization while finding an appropriate network configuration (the sectors and their RAN parameters) in OFDMA based cellular networks such as LTE and WiMAX for a given traffic demand density. We study the inter-relationship between the different system objectives and provide guidelines to achieve best trade-offs between the conflicting objectives.

\item Unlike previous works, ICI is accurately modeled taking into account the large scale fading (shadowing), small scale fading, and load condition in neighbouring cells. The accuracy of the model is verified through simulations and it is seen that it works for wide range of standard deviation of shadowing and resource utilization (sub-carrier occupancy) in neighbouring cells. The objective functions: network coverage, overlap probability, ASE, and APC are derived using the developed ICI model.

\item  A cell load model that relates the offered traffic load to the radio resource utilization is developed in this paper. The model incorporates the traffic characteristics, and the propagation characteristics such as path loss, shadowing, and small scale fading. The resource utilization and blocking probability in an individual cell is obtained using Kaufman Roberts Algorithm (KRA) through traffic and SINR statistics which take into account the load condition in the neighbouring cells. The accuracy of the model is verified through event-driven simulations. 

\item The framed multi-objective optimization problem in this work is a complex combinatorial problem. The number of possible BS combinations exponentially increases with the number of sectors. Identifying RAN parameter along with the BS configuration adds further complexity. Genetic Algorithm (GA) based evolutionary approach is used to find solutions of the multi-objective optimization problem. The advantage of using the proposed framework in dynamically adjusting the network configuration according to varying traffic conditions with reduced complexity is discussed. 

\item The performance of the proposed framework is evaluated on the dense Urban Micro (UMi) scenario which include realistic propagation conditions, three dimensional antenna pattern etc. It is found that for a given traffic demand density, the solutions are converged to a particular region in the search space. Results demonstrate that using the proposed framework, significant amount of ES can be achieved while maintaining good trade-offs between the other system objectives at faster convergence and with lower computational complexity.
\end{itemize}

\subsection{Paper Organization}
The paper is organized as follows. Section \ref{system_model} describes the system model. Problem formulation is given in Section \ref{problem_formulation}. The details of ICI modeling, derivation of network coverage, overlap, ASE, and APC are provided in Section \ref{analysis}. Solution approach for the considered multi-objective optimization framework is given in Section \ref{solution_approach}. Results are provided in Section \ref{results}. Finally, conclusions are given in Section \ref{conclusion}.

\section{System Model}\label{system_model}
Let $\mathcal{B}$ be the set of $N_{\mathcal{B}}$ sectors located inside a geographical region $\mathbb{D}\subset \mathbb{R}^2$ with area $\mathcal{A}_{\mathbb{D}}$. The geographical region $\mathbb{D}$ is divided into $n$ small rectangular grids each with area $dA_i$, where $i$ is grid index. The users are spatially distributed within the network with some distribution $p(.)$ such that $\int_\mathbb{D}\ p(i) \ dA_i = 1$. The network uses OFDMA as the air interface with frequency reuse of unity. The available total system bandwidth {{$B$}} Hz consists of $N_{sc}$ number of sub-channels each with bandwidth $\Delta f_{sc}$ Hz. Table \ref{summary_of_notations} summarizes the list of symbols used in this work. 
     \begin{table}[htb]
     \centering
     \caption{Summary of Notations}\label{summary_of_notations}
     \begin{tabular}{|p{1.015cm}|p{6.75cm}|} \hline
     \textbf{Symbol} & \textbf{Description} \\ \hline
     $i$ & user location \\ \hline
     $j$ & sector index \\ \hline
     $N_{\mathcal{B}}$ & total number of sectors \\ \hline
     $H_{tj}$ & $j$-th sector antenna height (in meters) \\ \hline
     $\phi_{\mathrm{tilt}_j}$ & $j$-th sector antenna's tilt angle (in Deg.s) \\ \hline
     $\theta_{ij}$ & azimuth angle between $j$-th sector antenna and location $i$ (in Deg.s) \\ \hline
     $\phi_{ij}$ & elevation angle between $j$-th sector antenna and location $i$ (in Deg.s) \\ \hline
     $B$ & system bandwidth (in Hz) \\ \hline
     $N_{sc}$ & number of sub-channels \\ \hline
     $\Delta f_{sc}$  & sub-channel spacing \\ \hline
     $k$ & sub-channel index \\ \hline
     $l$ & MCS level index \\ \hline
     $N_L$ & number of MCS levels \\ \hline
     $b_l$ & number of bits transmitted using $l$-th MCS \\ \hline
     $P_{tj}$ & transmit power of $j$-th sector \\ \hline
     $w_l$ & fraction of users belongs to $l$-th MCS \\ \hline
     $v_{kj}$ & activity status of $j$-th cell on sub-channel $k$ \\ \hline
     $\beta_j$ & activity factor of $j$-th cell \\ \hline
     $h_{ij}$ & Nakagami-$m$ distributed channel gain between location $i$ and $j$-th sector antenna\\ \hline
     $\xi$ & Gaussian distributed shadowing coefficient (in dB) \\ \hline
     $\sigma_{\xi}$ & standard deviation of shadow fading (in dB) \\ \hline
     $\alpha$ & path-loss exponent \\ \hline
     $R_{\mathrm{req}}$ & user rate requirement (bits/sec)\\ \hline
     $\mathcal{B}_{\mathrm{on}}$ & set of active sectors \\ \hline
     $\lambda$ & inter-arrival time (in seconds)\\ \hline
     $\mu$ & call holding time (in seconds) \\ \hline
     $T$ & average cell spectral efficiency (in b/s/Hz) \\ \hline
     $C$ & average cell throughput (bits/sec)\\ \hline
     $P_{r,\mathrm{min}}$& received power threshold (in dBm) \\ \hline
     $\Gamma_{\mathrm{min}}$ & minimum required SINR (in dB)\\ \hline
     $P_{b}(l)$ & blocking probability of $l$-th class \\ \hline
     $P_{bj}$ & average blocking probability in $j$-th cell \\ \hline
     $P_{b_{\mathrm{max}}}$ & maximum allowable blocking probability \\ \hline
     ${f_{\mathrm{APC}}}$ & area power consumption (in W/m$^2$) \\ \hline
     ${f_{\mathrm{ASE}}}$ & area spectral efficiency (in b/s/Hz/m$^2$) \\ \hline
     ${f_{\mathrm{COV}}}$ & network coverage (in \%) \\ \hline
     ${f_{\mathrm{OL}}}$ & coverage overlap (in \%)\\ \hline
     \end{tabular}
     \end{table}

\subsubsection{Channel Model}\label{channel_model}
The signal power received at location $i$ from sector $j$ is modeled as 
\begin{equation}\label{rx_power}
P_{r,ij} = P_{tj} \ G_A(\theta_{ij},\phi_{ij}) \ PL(d_{ij},{\alpha},f_c) \ |h_{ij}|^2 \ \chi_{ij},
\end{equation}
where 
\begin{itemize}
\item {{$P_{tj}$}} is transmit power of $j$-th sector
\item Antenna gain \cite{m2135} 
\begin{equation}
G_A(\theta_{ij},\phi_{ij}) = -\operatorname{\mathrm{min}}[-(A_{\theta_{ij}}+A_{\phi_{ij}}),A_{m}], 
\end{equation}
where \begin{equation}
A_{\theta_{ij}}=-\operatorname{\mathrm{min}}\left[\left(\frac{\theta_{ij}}{\theta_{3dBj}}\right)^2,A_{m}\right]\nonumber
\end{equation}
and 
\begin{equation}
A_{\phi_{ij}}=-\operatorname{\mathrm{min}}\left[\left(\frac{\phi_{ij}-\phi_{\mathrm{tilt}j}}{\phi_{3dBj}}\right)^2,A_{m}\right]\nonumber 
\end{equation}
are the horizontal and vertical antenna patterns, respectively \cite{m2135}, $\theta_{ij}$ and $\phi_{ij}$ are the azimuth and elevation angle between $j$-th sector antenna and location $i$, respectively, $\phi_{\mathrm{tilt}j}$ is the vertical tilt angle of $j$-th sector antenna, $\theta_{3dBj}$ and $\phi_{3dBj}$ are the 3 dB beamwidth of horizontal and vertical antenna patterns of $j$-th sector antenna, and {{$A_{m}$}} is maximum attenuation \cite{m2135}. The elevation angle is calculated as 
\begin{equation}
\phi_{ij}=\operatorname{arctan}\bigg(\frac{H_{tj}-H_{ti}}{d_{ij}}\bigg), 
\end{equation}
where $H_{tj}$ and $H_{ti}$ are the heights of $j$-th sector and user at location $i$, respectively. 

\item Nakagami-$m$ distributed random variable $|h_{ij}|$ represents the envelope of small scale fading gain between location {{$i$}} and $j$-th sector, so the power of fast fading $|h_{ij}|^2$ follows Gamma distribution. 

\item Log-normally distributed random variable $\chi_{ij}$ represents the shadowing component between location $i$ and $j$-th sector i.e. 
\begin{equation}
\chi_{ij}=\exp(\eta \xi_{ij}), 
\end{equation}
where $\xi_{ij}$ is a Gaussian random variable (in dB) with zero mean and variance $\sigma_{\xi_{ij}}^2$ and $\eta = \frac{\ln(10)}{10}$. 

\item Path-loss component $PL(d_{ij},{\alpha},f_c)$ is denoted as a function of the distance between $j$-th sector and location $i$ ($d_{ij}$), path-loss exponent $\alpha$, and carrier frequency $f_c$. 

\item The set of transmit powers, sector antenna's vertical tilt angles and antenna heights are denoted by \begin{equation}
\boldsymbol{\mathcal{P}}=[P_{t1},P_{t2},...,P_{tN_{\mathcal{B}}}], \nonumber 
\end{equation}
\begin{equation}
\boldsymbol{\phi}=[\phi_{\mathrm{tilt}1},\phi_{\mathrm{tilt}2},...,\phi_{\mathrm{tilt}N_{\mathcal{B}}}], \nonumber 
\end{equation}
and 
\begin{equation}
\boldsymbol{\mathcal{H}}=[H_{t1},H_{t2},...,H_{t{N_{\mathcal{B}}}}],\nonumber 
\end{equation}
respectively.
\end{itemize}

\subsubsection{Traffic Model}
In this work we consider streaming traffic and blocking probability as the QoS metric. Let $\tau_a$ be independently distributed exponential random variable represents the inter-arrival time of the streaming calls per unit area with mean $1/\lambda$. Let $\tau_s$ be independently distributed exponential random variable represents the duration of the call with mean $1/\mu(i)$ at location $i$. Then the traffic demand density of the call at location $i$ is $\rho(i)=\lambda(i)/\mu(i)$ (Erlang/m$^2$). Let ${\mathbb{D}_j}$ be the set of locations covered by sector $j$. Then, the offered traffic by the users in cell $j$ to the network in Erlang is $\rho_{j} = \int_{\mathbb{D}_j}\rho(i) \ dA_i$. The framework can be easily extended to mixed traffic scenarios to analyze energy saving in terms of additional QoS parameters such as throughput and delay as well. 

\subsubsection{Traffic Load and Activity Status of Interfering Sectors} \label{traffic_load_assumption}
As mentioned in the introduction section, the previous works considered that the interferers are always active on a sub-channel and ICI is modeled as a Gaussian-like noise \cite{eunsung2013,han2013}. However, in practice, the number of occupied sub-channels in a cell depends on the offered traffic load and the SINR statistics in that cell. Therefore, in this work, first we obtain SINR distribution by using the accurate estimate of resource utilization in neighbouring cells. Next, we use the SINR distribution to obtain the important system performance metrics such as coverage, overlap, ASE, and APC. Let $\beta_j$ be the ratio of number of occupied sub-channels to the total number of sub-channels $N_{sc}$ in $j$-th sector for a given traffic demand density $\rho$ (Erlang/m$^2$). Assuming that the BS uniformly choose the sub-channels for it users, the $j$-th interfering sector's activity status on sub-channel $k$ can be modeled as a binary random variable $v_{kj}$ with the first order probability mass function \cite{prabhu2014_twc}
\begin{equation}
Pr[v_{kj}=1]= \beta_j 
\end{equation}
and 
\begin{equation}
Pr[v_{kj}=0]= 1-\beta_j.
\end{equation}
Calculation of $\beta_j$ for a given traffic demand density and SINR distribution is discussed in Section \ref{sec:blocking_probability}. 

\section{Problem formulation}\label{problem_formulation}
Given the traffic demand density {{$\rho$}} (Erlang/m$^2$), our objective is to find the set of active sectors ($\mathcal{B}_{\mathrm{on}}$), sector antenna's vertical tilt angles ($\boldsymbol{\phi}$), sector transmit powers ($\boldsymbol{\mathcal{P}}$), sector antenna heights ($\boldsymbol{\mathcal{H}}$) that jointly maximizes the network coverage (${f_{\mathrm{COV}}}$), area spectral efficiency (${f_{\mathrm{ASE}}}$), and minimizes area energy consumption (${f_{\mathrm{APC}}}$) and the overlap (${f_{\mathrm{OL}}}$) while satisfying target blocking probability requirements. Let 
\begin{equation}\textbf{x}=[x_1,x_2,x_3,x_4,x_5]^T=[\mathcal{B}_{\mathrm{on}},\boldsymbol{\beta},\boldsymbol{\mathcal{P}},\boldsymbol{\phi},\boldsymbol{\mathcal{H}}]^T,\nonumber \end{equation}
\begin{equation}g_1(\textbf{x})={-f_{\mathrm{APC}}}(\mathcal{B}_{\mathrm{on}},\boldsymbol{\beta},\boldsymbol{\mathcal{P}},\boldsymbol{\phi},\boldsymbol{\mathcal{H}}), \nonumber
\end{equation}
\begin{equation}
g_2(\textbf{x})={f_{\mathrm{ASE}}}(\mathcal{B}_{\mathrm{on}},\boldsymbol{\beta},\boldsymbol{\mathcal{P}},\boldsymbol{\phi},\boldsymbol{\mathcal{H}}), \nonumber
\end{equation}
\begin{equation}
g_3(\textbf{x})={f_{\mathrm{COV}}}(\mathcal{B}_{\mathrm{on}},\boldsymbol{\beta},\boldsymbol{\mathcal{P}},\boldsymbol{\phi},\boldsymbol{\mathcal{H}}),\nonumber
\end{equation} and 
\begin{equation}
g_4(\textbf{x})={-f_{\mathrm{OL}}}(\mathcal{B}_{\mathrm{on}},\boldsymbol{\beta},\boldsymbol{\mathcal{P}},\boldsymbol{\phi},
\boldsymbol{\mathcal{H}}).\nonumber
\end{equation}
The multi-objective optimization problem is formulated as 
\begin{equation}\label{objective_function}
{\underset{\textbf{x}} {\mathrm{max}} \ \ \textbf{g}(\textbf{x})=[g_1(\textbf{x}),g_2(\textbf{x}),g_3(\textbf{x}),g_4(\textbf{x})]^T}
\end{equation}

\begin{align}
\normalsize{\text{Subject to:}} \ \ & C1:P_{bj}\leq P_{b,\mathrm{max}},\ \forall j \in \mathcal{B}_{\mathrm{on}};\ \nonumber \\ &C2: g_3(\textbf{x})\ge {f_{\mathrm{COV}}}_{\mathrm{min}};\nonumber \\& C3:\boldsymbol{\mathcal{P}}_{\mathrm{min}}\leq x_3 \leq \boldsymbol{\mathcal{P}}_{\mathrm{max}}; \ \nonumber \\ & C4:\boldsymbol{\phi}_{\mathrm{min}}\leq x_4  \leq \boldsymbol{\phi}_{\mathrm{max}}; \ \nonumber \\ & C5:\boldsymbol{\mathcal{H}}_{\mathrm{min}}\leq x_5  \leq \boldsymbol{\mathcal{H}}_{\mathrm{max}}. \nonumber
\end{align}
\normalsize
Here, constraint $C1$ represents the blocking probability requirement i.e. with the given  set of active sectors ($\mathcal{B}_{\mathrm{on}}$), the blocking probability in every cell $j$ should not exceed the maximum allowable blocking probability $P_{b,\mathrm{max}}$. Constraint $C2$ represent the minimum coverage requirement. Constraints $C3$, $C4$, and $C5$ represent the range of transmit power, tilt, and height, respectively at individual sector. $\boldsymbol{\mathcal{P}}_{\mathrm{min}}$, $\boldsymbol{\phi}_{\mathrm{min}}$, and $\boldsymbol{\mathcal{H}}_{\mathrm{min}}$ are the set of minimum value of transmit power, tilt angle, and height available at individual sector, respectively. $\boldsymbol{\mathcal{P}}_{\mathrm{max}}$, $\boldsymbol{\phi}_{\mathrm{max}}$, and $\boldsymbol{\mathcal{H}}_{\mathrm{max}}$ are the set of maximum value of transmit power, tilt angle, and height available at individual sector, respectively.

In the next section, first we derive the above mentioned objective functions and constraints. Then we study the inter-relationships between the objective functions and different variables defined in the optimization problem in \eqref{objective_function}.

\section{Analysis of Network coverage, Overlap, ASE, and APC}\label{analysis}
In this section, we derive the interference statistics, blocking probability, resource utilization and the objective functions: network coverage $f_{\mathrm{COV}}$, overlap probability ${f_{\mathrm{OL}}}$, ASE ${f_{\mathrm{ASE}}}$, and APC ${f_{\mathrm{APC}}}$, as a function of decision variables: active set of sectors $\mathcal{B}_{\mathrm{on}}$, transmit powers $\boldsymbol{\mathcal{P}}$, sector antenna's vertical tilt angles $\boldsymbol{\phi}$,  sector antenna heights $\boldsymbol{\mathcal{H}}$, and resource occupancy $\boldsymbol{\beta}$. However, for notational simplicity, we omit these parameters in the subsequent sections. 

We also validate the models presented in this section through Monte-Carlo simulations. The results are generated for the dense Urban Micro-cell (UMi) network scenario with an inter-site-distance of 200 m as shown in Figure \ref{Layout_UMi_UF_TWC}. Wrap around model is used to incorporate equal ICI at the edge regions. The user locations are assumed to be uniformly distributed within the network i.e. $p(i) = \frac{1}{\mathcal{A}_\mathbb{D}}, \ \forall i$. The total system bandwidth is 10 MHz composed of 600 useful sub-channels each with bandwidth $15$ kHz. For the traffic model, we assume streaming call requests with a rate requirement ($R_{\mathrm{req}}$) of $128$ kbps. The maximum allowable blocking probability in a cell is assumed to be $P_{b,\mathrm{max}}= .02$. Further, homogeneous traffic distribution is assumed (i.e. $\rho(i) = \rho, \ \forall i$). The non-line of sight (NLoS) path-loss model for UMi scenario is used as given in \cite{m2135}. The horizontal and vertical antenna pattern parameters are used as per \cite{m2135}. The additional system parameters are given in Table \ref{system_parameters}. Minimum required received power ($P_{r,\mathrm{min}}$) and minimum required SINR ($\Gamma_{\mathrm{min}}$) thresholds are assumed to be $-102$ dBm ($10$ MHz) and $-10$ dB, respectively. The SINR thresholds for $15$ different MCS levels for SISO Rayleigh fading scenario is used with the target BLER of $.1$ \cite{mehlfuhrer2009}. The corresponding number of sub-channels required to make a call is given in Table \ref{SINR_class}.

\begin{table}[thb]
\caption{ADDITIONAL SYSTEM PARAMETERS}
\begin{tabular}{|p{5.25cm}|p{2.75cm}|}\hline
\label{system_parameters}
\textbf{Parameter} &  \textbf{Value}\\\hline
Carrier Frequency & 2.5 GHz\\\hline
Bandwidth ($B$) & 10 MHz\\\hline
Sub-carrier spacing ($\Delta f_{sc}$) & 15 KHz\\\hline
Number of MCS levels ($N_L$) & 11\\\hline
3 dB vertical beamwidth ($\phi_{3dB}$) & 15 Deg. \\\hline
3 dB horizontal beamwidth ($\theta_{3dB}$) & 70 Deg.\\\hline
Antenna gain (boresight) & 17 dBi\\\hline
Sector antenna height ($H_t$) & 20 m\\\hline
Vertical tilt angle ($\phi_{\mathrm{tilt}}$) & 12 Deg. \\\hline
Thermal noise level ($N_0$) & -104 dBm/10 MHz \\\hline
Minimum required SINR threshold ($\Gamma_{\mathrm{min}}$) & -10 dB \\\hline
Target blocking probability ($P_{b,\mathrm{max}}$) & 2\%\\\hline
Std. of shadow fading ($\sigma$) & 6 dB\\ \hline
\end{tabular}
\end{table}

     \begin{table*}[htb]
     \centering
     \caption{MCS levels and corresponding SINR Thresholds}\tiny{
     \begin{tabular}{|p{4cm}|p{.35cm}|p{.36cm}|p{.36cm}|p{.36cm}|p{.36cm}|p{.36cm}|p{.36cm}|p{.36cm}|p{.36cm}|p{.36cm}|p{.36cm}|p{.36cm}|p{.36cm}|p{.36cm}|p{.36cm}|p{.36cm}|p{.36cm}|}
     \hline \textbf{MCS level index ($l$)} & 1  &2&3&4&5&6&7&8&9&10&11&12&13&14&15\\
     \hline \textbf{Modulation } & QPSK  &QPSK&QPSK&QPSK&QPSK&QPSK&16-QAM&16-QAM&16-QAM&64-QAM&64-QAM&64-QAM&64-QAM&64-QAM&64-QAM\\
     \hline \textbf{Coding rate} & .076  &.12&.19&.3&.44&.59&.37&.48&.6&.45&.55&.65&.75&.85&.93\\
     \hline \textbf{SINR Threshold (dB) ($\gamma_l$)} & -7.5 &-5&-3&-1&1&3.5&5&7&9&11&13.5&15&16&17.5&19\\
     \hline \textbf{Number of sub-channels required ($n_{sc}(l)$)}& 56  &36&22&14&10&7&6&4&4&3&3&2&2&2&2\\
     \hline
     \end{tabular}}\label{SINR_class}
     \end{table*}

\subsection{SINR Model}
We model SINR as a function of shadow parameter ($\chi$), small scale fading parameter ($h$) and activity status of interfering cells ($v$). Following the channel model given in \eqref{channel_model}, the received SINR for a user at location $i$ when attached to sector $j$ can be expressed as

\small
\begin{align}\label{sinr_load}
&\gamma_{k,ij}(\chi,h,v)=\nonumber \\& \frac{P_{tkj} \ G_{ij}(\theta_{ij},\phi_{ij}) \  PL(d_{ij},{\alpha},f_c) \ \chi_{ij}\ |h_{kij}|^2 }{\sum\limits_{g\neq j, g\in \mathcal{B}_{\mathrm{on}}}\! \! \! \! {v_{kg} P_{tkg} G_{ig}(\theta_{ig},\phi_{ig})  PL(d_{ig},{\alpha},f_c) \chi_{ig} |h_{kig}|^2 } + \Delta f_{sc} N_0},
\end{align}\normalsize
where $\mathcal{B}_{\mathrm{on}}\subseteq \mathcal{B}$ is set of active sectors, $P_{tkj}$ and $P_{tkg}$ denote the transmit power of desired sector $j$ and interfering sector $g$, respectively, on sub-channel $k$ and $N_0$ is noise power spectral density in W/Hz. $v_{kg}$ denotes the activity status of $g$-th interfering sector which is related to resource occupancy in $g$-th sector. 

For further analysis, we rewrite the SINR in \eqref{sinr_load} as
\begin{equation}\label{sinr_load_1}
\gamma_{k,ij}(\chi,h,v)=\frac{\hat{\chi}_{ij}\ |h_{kij}|^2 }{\sum\limits_{g\neq j, g\in \mathcal{B}_{\mathrm{on}}}\! \! \! \! {v_{kg} \ \hat{\chi}_{ig} \ |h_{kig}|^2 } + P_{N}},
\end{equation}
where 
\begin{equation}
\hat{\chi}_{ij} = \exp(\ln(P_{tkj} \ G_{ij}(\theta_{ig},\phi_{ig}) \  PL(d_{ij},{\alpha},f_c))+\eta\xi_{ij})\nonumber 
\end{equation} and $P_N=\Delta f_{sc} N_0$.

\subsection{Received Power Statistics}
The desired signal power i.e. the numerator in \eqref{sinr_load_1} is the product of log-normal RV {{$\hat{\chi}_{ij}$}} with mean 
\begin{equation}
\mu_{\hat{\chi}_{ij}}=\ln(P_{tkj} \ G_{ij}(\theta_{ij},\phi_{ij}) \  PL(d_{ij},{\alpha},f_c))\nonumber \end{equation}
and variance \begin{equation}
\sigma_{\hat{\chi}_{ij}}^2=\eta^2\sigma_{\xi_{ij}}^2,\end{equation} and Gamma RV {{$|h_{kij}|^2$}}. The combined PDF of log-normal-Gamma RV {{$P_{r,ij}=\hat{\chi}_{ij} |h_{kij}|^2$}} is expressed as \cite[Sec: 4.2.1]{stuber2001} \vspace{-.5cm}

\begin{align}
\!\! p_{P_{r,ij}}(x)\! =\!& \int_0^\infty p_{{|h_{kij}|^2}/\hat{\chi}_{ij}}(x/y) \ p_{\hat{\chi}_{ij}}(y) \ dy \nonumber \\\! =\!&  \int_0^\infty\!\! \bigg(\frac{m}{y}\bigg)^m \frac{x^{m-1}e^{ -\frac{mx}{y}}}{\Gamma(m)} \frac{e^{-\frac{(\ln y-\mu_{\hat{\chi}_{ij}})^2}{2\sigma_{\hat{\chi}_{ij}}^2}}}{\sqrt{2\pi} y \sigma_{\hat{\chi}_{ij}}} dy.
\end{align}\normalsize

The above PDF can be approximated as a single log-normal random variable \cite[Sec: 4.2.1]{stuber2001} as
\begin{equation}\label{pdf_ln}
p_{P_{r,ij}}(P_r)=\frac{1}{P_r \sigma_{P_{r,ij}} \sqrt{2\pi}} \exp{\left[\frac{-(\ln P_r-\mu_{P_{r,ij}})^2}{2\sigma_{P_{r,ij}}^2}\right]},
\end{equation}
In \eqref{pdf_ln}, 
\begin{equation}
\mu_{P_{r,ij}} =(\psi(m)-\ln(m))+\mu_{\hat{\chi}_{ij}}\nonumber \end{equation}
 and 
\begin{equation}
\sigma_{P_{r,ij}}^2 = \zeta(2,m)+\sigma_{\hat{\chi}_{ij}}^2,\nonumber \end{equation}
 where $\psi(.)$ is the Euler psi function and {{$\zeta(.,.)$}} is Riemann's zeta function.

\subsection{Interference Power Statistics : An MGF Based Approximation Approach}
Since the interference pattern changes with the set of active sectors, user traffic, and resource occupancy in an individual sector, the network performance metrics: coverage, overlap, area spectral efficiency, and area energy efficiency also changes. Therefore, it is important to capture the dynamics of various parameters into the interference model in order to understand the behaviors of these performance metrics. In this work we derive the distribution of interference power taking into account the shadowing, small scale fading, and resource occupancy in the neighbouring sectors. Since the distribution of the sum of interference powers is unknown, it is approximated as a single log-normal RV using various methods (F-W method, moment generating function (MGF) method \cite{mehta2007}, etc.). However, all these approximations considered full resource occupancy in the neighbouring cells (i.e. interferers are always active). A modified F-W method is proposed in \cite{fischione2007} for approximating the sum of log-normal processes weighted by binary processes. However, the modified F-W method works for standard deviation less than 4 dB and high values of activity factor (i.e. $\beta > .4$). In this work, by extending the MGF method proposed in \cite{mehta2007}, we approximate the total interference power at location $i$  i.e.
\begin{equation}
{\sum\limits_{g\neq j, g\in \mathcal{B}_{\mathrm{on}}} {v_{kg} P_{tkg} G_{ig}(\theta_{ig},\phi_{ig})  PL(d_{ig},{\alpha},f_c) \chi_{ig} |h_{kig}|^2 }} \nonumber
\end{equation}
as a single log-normal RV $P_{I_i}=10^{X/10}$, where $X$ is Gaussian RV with mean $\mu_X$ and variance $\sigma_{X}^2$.

The PDF of $P_{I_i}$ is 
\begin{equation}\label{intf_dist}
{p_{P_{I_i}}(I)=\frac{1}{I \sigma_{P_{r,ij}} \sqrt{2\pi}} \exp{\left[\frac{-(\ln I-\mu_{P_{I_i}})^2}{2\sigma_{P_{I_i}}^2}\right]}}.
\end{equation}
Detailed procedure for obtaining the above PDF using MGF method is given in Appendix A. 

\subsection{Network Coverage}

\begin{mydef} \textit{(Coverage probability of a location associated with sector $j$)}
The user at location $i$ is said to be under the coverage of sector-$j$, if the received power from the sector-$j$ i.e. $P_{r,ij}$ is above the threshold value $P_{r,\mathrm{min}}$ and the SIR  $\gamma_{ij}$ is above the threshold value $\Gamma_{\mathrm{min}}$ \cite{3gpp.32.522}. 
\end{mydef}

According to Definition 1, the probability that the user at location $i$ is under the coverage of sector-$j$ is obtained using $p_{P_{r,ij}}(P_r)$ and $p_{P_{I_i}}(I)$ as 

\begin{align}\label{cov_prob_location}
Pr(\gamma_{ij}& \ge \Gamma_{\mathrm{min}},P_{r,ij}\ge P_{r,\mathrm{min}})\nonumber\\
&=Pr\bigg(\frac{P_{r,ij}}{\Gamma_{\mathrm{min}}} \ge P_{I_i} ,P_{r,ij}\ge P_{r,\mathrm{min}}\bigg)\nonumber \\&=\!\int_{P_{r,\mathrm{min}}}^{\infty}\int_0^{Pr/\Gamma_{\mathrm{min}}} p_{P_{I_i}}(I) \ dI\  p_{P_{r,ij}}(P_r) \ dP_r.
\end{align}

The above integral can be approximated as  
\begin{align}\label{coverage_final_exp}
\hat{Pr}(\gamma_{ij} &\ge \Gamma_{\mathrm{min}},P_{r,ij}\ge P_{r,\mathrm{min}})\nonumber\\=&\exp{\bigg(-\frac{\big(\ln(P_{r,\mathrm{min}})-\mu_{P_{r,ij}}\big)^2}{2\sigma_{P_{r,ij}}^2}\bigg)}\nonumber \\
&-\frac{\sigma_{P_{I_i}}\exp{\bigg(E-\bigg(\frac{\ln(\frac{P_{r,\mathrm{min}}}{\Gamma_{\mathrm{min}}})-\mu_c}{\sqrt{2}{\sigma_c}}\bigg)^2\bigg)}}{\sqrt{\sigma_{P_{I_i}}^2+\sigma_{P_{r,ij}}^2}},
\end{align}
where 
\begin{equation}
E=\frac{-(\mu_{P_{I_i}}-\mu_{P_{r,ij}})^2-(\mu_{P_{I_i}}+\ln \Gamma_{\mathrm{min}})^2+\mu_{P_{I_i}}^2\big(1+\frac{2\mu_{P_{r,ij}}\ln \Gamma_{\mathrm{min}}}{\sigma_{P_{I_i}}^2}\big)}{\sigma_{P_{I_i}}^2+\sigma_{P_{r,ij}}^2}, \nonumber 
\end{equation}
\begin{equation}
\mu_c=\bigg(\frac{(\mu_{P_{r,ij}}-\ln \Gamma_{\mathrm{min}})\sigma_{P_{I_i}}^2+\mu_{P_{I_i}}\sigma_{P_{r,ij}}^2}{\sigma_{P_{I_i}}^2+\sigma_{P_{r,ij}}^2}\bigg), \nonumber 
\end{equation}
and 
\begin{equation}
\sigma_c={\frac{\sigma_{P_{I_i}}\sigma_{P_{r,ij}}}{\sqrt{\sigma_{P_{I_i}}^2+\sigma_{P_{r,ij}}^2}}}. \nonumber 
\end{equation} 
The derivation of \eqref{coverage_final_exp} is given in Appendix B.

\begin{myrem}
\textit{Bounds on coverage probability at any location $i$}: Since the coverage is defined based on both the SIR and received signal strength, at any location $i$, for a fixed active sector set $\boldsymbol{\mathrm{B}}_{\mathrm{on}}$, coverage would be maximum (i.e. best SIR condition) when the resource utilization in neighbouring sectors is low and it would be minimum when the resource utilization in neighbouring sectors is high. Further, let us assume that there are $N_{\mathcal{B}_{\mathrm{on}}}(\rho)$ sectors required to satisfy the traffic demand density of $\rho$. In this particular case, there are 
\begin{equation}\binom{N_{\mathcal{B}}}{N_{\mathcal{B}_{\mathrm{on}}}(\rho) }= \frac{N_{\mathcal{B}}!}{N_{\mathcal{B}_{\mathrm{on}}}(\rho) ! (N_{\mathcal{B}} - N_{\mathcal{B}_{\mathrm{on}}}(\rho) )!}\nonumber \end{equation} possible sector configurations. Over all possible sector configurations, assuming equal RAN parameter configurations, the sector set with the minimum resource utilization will provide the maximum coverage.
\end{myrem}

Further, the fraction of total area covered by the sector $j$ is obtained by averaging over the locations i.e. 
\begin{align}\label{coverage_individual_BS}
P_{Cov,j}(&\Gamma_{\mathrm{min}},P_{r,\mathrm{min}}) \nonumber \\
&= \sum_{i \in \mathbb{D}}\!  \hat{Pr}(\gamma_{ij} \ge \Gamma_{\mathrm{min}},P_{r,ij}\ge P_{r,\mathrm{min}}) p(i) \ dA_i. 
\end{align}

\begin{figure}[htb]
\centering
\includegraphics[scale=.58]{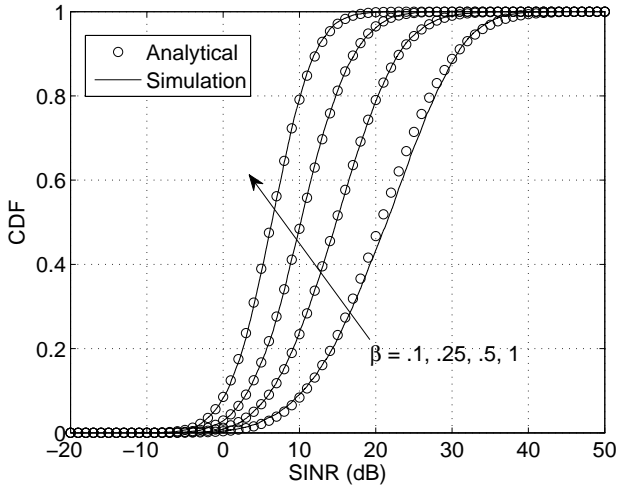}
\caption{CDF of downlink SINR}
\label{CDF_SINR_TVT}
\end{figure}

Fig. \ref{CDF_SINR_TVT} plots Eqn. \eqref{coverage_individual_BS} for path-loss exponent $\alpha = 4$ and $\sigma = 6$ dB. It can be seen that the analytical results are closely matching with simulation results for all ranges of load factor ($0<\beta_g\le1$). It is also observed that the approximation is good for wide range of standard deviation (up to 12 dB). 

\begin{mydef} (\textit{Coverage probability of a location associated with a set of active sectors }) The location $i$ is said to be under coverage, if the received signal power and the SINR from at least any one sector is greater than the minimum required received power $P_{r,\mathrm{min}}$ and the minimum required SIR $\Gamma_{\mathrm{min}}$, respectively.
\end{mydef}

According to Definition 2 the probability that the location $i$ is covered by at least one sector can be obtained as
\begin{align}
P_{Cov,i}(&\Gamma_{\mathrm{min}},P_{r,\mathrm{min}})\nonumber \\&  = 1 -  \prod_{j \in \mathcal{B}_{\mathrm{on}}}  \big[1-\hat{Pr}(\gamma_{ij} \ge \Gamma_{\mathrm{min}}P_{r,ij}\ge P_{r,\mathrm{min}})\big]. 
\end{align} \normalsize

Finally, the coverage probability for the entire region can be obtained by averaging over all locations i.e. 
\begin{align}\label{cost_func_cov}
{f_{\mathrm{COV}}}(\mathcal{B}_{\mathrm{on}},\boldsymbol{\beta},\boldsymbol{\mathcal{P}},\boldsymbol{\phi},\boldsymbol{\mathcal{H}}) = \sum_{i \in \mathbb{D}} P_{Cov,i}(\Gamma_{\mathrm{min}},P_{r,\mathrm{min}}) \ p(i) \ dA_i.
\end{align}
\begin{figure*}[ht]
\hrule
\footnotesize
\begin{equation}\label{overlap_ind_cell}
 \hspace{-.51cm} P_{OL,i}^{(b)} = \sum_{x_1=1}^{N-b+1} (1-Q_{x_1})\Bigg[ \sum_{x_2=x_1+1}^{N-b+2} (1-Q_{x_2})\Bigg[......\bigg[\sum_{x_{N-2}=x_{N-3}+1}^{N-1} (1-Q_{x_{N-2}})\bigg[\sum_{x_{N-1}=x_{N-2}+1}^{N} (1-Q_{x_{N-1}}) \sum_{x_{N+1}\neq x_1,x_2,x_3,...,x_N}^{N}Q_{x_{N+1}} \bigg]\bigg] ...... \Bigg]\Bigg].
\end{equation}\normalsize
For example, when $N_{\mathcal{B}}=4$ and $b=2$, 

\begin{align*}
P_{OL,i}^{(2)} = &(1-Q_{1})\Bigg[ \sum_{x_2=2}^{2} (1-Q_{x_2})\Bigg[\sum_{x_{3}=x_{2}+1}^{3} (1-Q_{x_{3}})\bigg[\sum_{x_4=x_3+1}^{4} (1-Q_{x_{4}})\sum_{x_{5}\neq x_1,x_2,x_3,x_4}^{4}Q_{x_{5}} \bigg]\Bigg]\Bigg]. \nonumber \\
= &(1-Q_1)(1-Q_2)Q_3 Q_4 + (1-Q_1)(1-Q_3)Q_2 Q_4 +(1-Q_1)(1-Q_4)Q_2 Q_3 \nonumber \\ 
& +(1-Q_2)(1-Q_3)Q_1 Q_4 +(1-Q_2)(1-Q_4)Q_1 Q_3 +(1-Q_3)(1-Q_4)Q_1 Q_2. \nonumber
\end{align*}\normalsize 
\hrule
\end{figure*}

\subsection{Overlap probability}
As mentioned in the introduction section it is important to cover a location by more than one sector in order to support hand-overs and reduce call drop probability. 
\begin{mydef}  \textit{(Probability that a location is under the coverage of $b$ active sectors)} The location $i$ is said to be under the coverage of $b$ sectors, if the received signal power from any of $b$ of active sector set is greater than the minimum required received power $P_{r,\mathrm{min}}$ and is less than that for remaining sectors. 
\end{mydef}

According to Definition 3 the overlap probability can be represented in terms of coverage and outage probabilities. The probability that the location $i$ is covered by $b$ sectors is obtained as given in \eqref{overlap_ind_cell}. In \eqref{overlap_ind_cell}, $Q_{x_1}$ represents the outage probability with respect to $x_1$-th sector i.e. \begin{equation}
Q_{x_1} = Q\bigg[\frac{\mu_{P_{r_{ix_1}}}-\ln(P_{r,\mathrm{min}})}{\sigma_{ix_1}}\bigg].\nonumber \end{equation}
Similarly, the probability that the location is covered by at least $b$ number of sectors cab be obtained as 
\begin{equation}P_{OL,i}^{(>b)} = \sum_{b' = b}^{N} P_{OL,i}^{(b')}.\nonumber
 \end{equation}
The fraction of total area covered by at least $b$ sectors is obtained by
\begin{equation}\label{cost_func_ol}
f_{\mathrm{OL}}(\mathcal{B}_{\mathrm{on}},\boldsymbol{\beta},\boldsymbol{\mathcal{P}},\boldsymbol{\phi},\boldsymbol{\mathcal{H}})= \sum_{i \in \mathbb{D}}P_{OL,i}^{(>b)} \ p(i) \ dA_i.
\end{equation}

\begin{myrem} The call dropping probability is minimum when the location $i$ is covered by more number of sectors. On the other hand, increased number of sectors would result in reduction in coverage because of increased ICI.
\end{myrem}

\subsection{Area Spectral Efficiency}
Let $N_L$ be the number of MCS levels and the corresponding SINR thresholds be $\Gamma_1$, $\Gamma_{2},..., \Gamma_{N_L-1}$, $\Gamma_{N_L}$. In the long-term, the fraction of users belonging to a particular range of SINR thresholds with respect to $j$-th sector can be obtained as 
\begin{align}\label{w_j_eqn}
w_j&(l)=  \sum_{i \in \mathbb{D}} \bigg[\hat{Pr}(\gamma_{ij} \ge \Gamma_{l},P_{r,ij}\ge P_{r,\mathrm{min}}) \nonumber \\ &\hspace{1cm} -\hat{Pr}(\gamma_{ij} \ge \Gamma_{l+1},P_{r,ij}\ge P_{r,\mathrm{min}})\bigg] \ p(i) \ dA_i.
\end{align}
\normalsize

Let $b_l$ be the number of bits transmitted per Hz by using $l$-th MCS i.e. \begin{equation}b_l = \log_2\bigg(1+\frac{\Gamma_l}{G}\bigg), \nonumber 
\end{equation} where $G$ is Shannon gap. Then the average spectral efficiency can be written as 
\begin{equation}\overline{T}_j = \sum_{l=1}^{N_{L}} w_j(l) . \ b_l\nonumber 
\end{equation} and the average throughput (bits/sec) achieved by the $j$-th cell is obtained by 
\begin{equation}
C_{j} =  B  . \overline{T}_j = B. \sum_{l=1}^{N_{L}} w_j(l) . \ b_l.
\end{equation}

Finally, the ASE (b/s/Hz/m$^2$) is obtained by sum of spectral efficiencies achieved by all the active sectors divided by total area i.e.
\begin{equation}\label{cost_func_ase}
{f_{\mathrm{ASE}}}(\mathcal{B}_{\mathrm{on}},\boldsymbol{\beta},\boldsymbol{\mathcal{P}},\boldsymbol{\phi},\boldsymbol{\mathcal{H}}) = \frac{1}{B\mathcal{A}_{\mathbb{D}}}  \sum_{j\in \mathcal{B}_{\mathrm{on}}} C_{j}. \ 
\end{equation}

\begin{myrem}
The higher the overlap between the sectors lesser the call drop probability whereas lesser the spectral efficiency achieved by a cell as the SIR degrades with increased number of interfering components. However, whenever the resource utilization in the interfering cells is minimum, then the spectral efficiency achieved by the cell would be high.
\end{myrem}

\subsection{Blocking Probability and Resource Utilization} \label{sec:blocking_probability}
Since the shadowing and ICI are the important factors that significantly affect the coverage and QoS performance, in this work we obtain the blocking probability and the resource utilization from the SINR statistics of an individual sector. The blocking probability is defined as follows.

\begin{mydef} Blocking probability in cell $j$, i.e. $P_{bj}$ is defined as the ratio of the number of blocked calls to the total number of call arrivals in cell $j$ in the long run. 
\end{mydef}

Since the BS assigns a certain number of sub-channels according to the feedback provided by the UEs on the channel conditions, we consider each MCS level as a class. For example, the UEs located near the BS may require only a few sub-channels compared to the UEs located far away from the BS. In adaptive transmission, the BS assigns a particular MCS level if the received SINR lies between a particular range of SIR thresholds. When the $l$-th MCS level is chosen for transmission, the bandwidth required to support a call with rate requirement $R_{\mathrm{req}}$ is 
\begin{equation}
 B_{\mathrm{req}}(l) = \frac{R_{\mathrm{req}}}{b_l}, \ \Gamma_l<\gamma<\Gamma_{l+1}.
\end{equation}
Then the number of sub-channels required by $l$-th class can be written as 
\begin{equation}
n_{sc}(l) = \Big\lceil \frac{B_{\mathrm{req}}(l)}{\Delta f_{sc}}\Big\rceil\ , \ l = 1,...,N_L. \nonumber 
\end{equation}

Let $n_{uj}(l)$ be the number of class-$l$ calls in $j$-th cell. The total number of sub-channels utilized by the calls is 
\begin{equation}
\boldsymbol{n_{uj}}.\boldsymbol{n_{sc}}= \sum_{l=1}^{L}n_{uj}(l)n_{sc}(l), \nonumber 
\end{equation} 
where \begin{equation}
\boldsymbol{n_{uj}}=(n_{uj}(1),...,n_{uj}(l),...,n_{uj}(L))\end{equation}
 and \begin{equation}
\boldsymbol{n_{sc}}=(n_{sc}(1),...,n_{sc}(l),...,n_{sc}(L)).\end{equation}
 A new call is admitted in cell $j$ when there is a required amount of resources available for the new coming calls. That is, the number of sub-channels utilized by the ongoing calls should be less than the total number of available sub-channels $N_{sc}$. Since the calls arrive from different locations inside the network at different times , the dynamic variation of the call arrival process can be treated as a Markov process \cite{karray2010,ross1995}. The Markov process is defined as the state space 
\begin{equation}
\boldsymbol{\mathcal{S}}:=\bigg\{\boldsymbol{n_{uj}} \in \mathcal{I}^L :  \boldsymbol{n_{uj}}.\boldsymbol{n_{sc}}  \leq N_{sc}\bigg \},\nonumber 
\end{equation}
 where $\mathcal{I}^L$ is the set of non-negative integers. The blocking probability experienced by the streaming calls is as follows:

\begin{mythe} \textit{For a given traffic demand density $\rho$, the average blocking probability in $j$-th cell is given by}
\begin{equation} \label{average_bp_cell}
P_{bj} = \frac{1}{\sum_{c=0}^{N_{sc}} g_j(c)}\sum_{l=1}^{N_L}w_j(l)\sum_{c=N_{sc}-n_{sc}(l)+1}^{N_{sc}} g_j(c),
\end{equation}
where $g(c)$ represents the probability that there are $c$ number of sub-channels are occupied i.e. 
\begin{equation}\frac{1}{c} \sum_{l=1}^{N_L}\rho_j(l) . n_{sc}(l). g_j(c-n_{sc}(l)), \ c=0,...,N_{sc}. \nonumber 
\end{equation}
\end{mythe}

\begin{proof}
The proof is given in Appendix C.
\end{proof}

\begin{myrem}
The blocking probability in cell $j$ is primarily affected by the fraction of users belonging to lower class i.e. users with poor SIR condition especially at the edge regions. This is because the number of sub-channels required to make a call is more for those users with poor SIR which results in increased resource consumption i.e. the higher the number of low class users higher the blocking probability. One way to maintain the blocking probability requirements is to keep the overlap in a controlled manner.
\end{myrem}

 The maximum traffic supported by the $j$-th cell is obtained when the blocking probability is equal to the maximum allowable blocking probability $P_{b,\mathrm{max}}$ i.e. 
 \begin{equation}\label{max_traffic_cell}
 \rho_{j}^{\mathrm{max}} = \big\{\rho_{j}(\rho) \big|P_{bj}(\rho) = P_{b,\mathrm{max}} \big\}.
 \end{equation}
 
Next, we state the resource utilization from the occupancy probabilities.
\begin{mydef} \textit{(Average Resource Utilization:)} The resource utilization in cell $j$ is defined as the ratio of the number of occupied sub-channels due to offered traffic load $\rho_j$ to the total number of sub-channels $N_{sc}$.
\end{mydef}

\paragraph*{\textbf{Corollary 1}} 
\textit{The fraction of bandwidth utilized in cell $j$ can be obtained from the occupancy probabilities} \eqref{occupancy_probability_1} i.e. 

\begin{align} \label{occupancy_probability}
\beta_{j} &= \frac{1}{N_{sc}}\frac{\sum_{c=0}^{N_{sc}}c. g_j(c)}{\sum_{c=0}^{N_{sc}} g_j(c)} \nonumber \\ & =\frac{1}{N_{sc}}\frac{\sum_{c=0}^{N_{sc}} \sum_{l=1}^{N_L} \rho_{j} w_j(l). n_{sc}(l). g_j(c-n_{sc}(l))}{\sum_{c=0}^{N_{sc}}\frac{1}{c} \sum_{l=1}^{N_L} \rho_{j} w_j(l). n_{sc}(l). g_j(c-n_{sc}(l))}.
\end{align}

\begin{figure}[htb]
\centering 
\includegraphics[scale=.55]{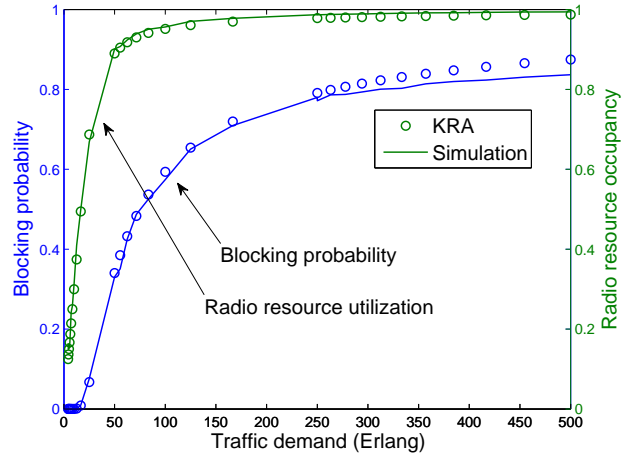}
\caption{{Blocking probability and corresponding resource occupancy}}
\label{BP_RU_Synopsis}
\end{figure}

Fig. \ref{BP_RU_Synopsis} shows the blocking probability and corresponding resource utilization in a single cell scenario (cell radius: 400 m, transmit power: $P_t$=43 dBm, bandwidth: $B$=10 MHz, and $R_{\mathrm{req}}=128 $) kbps) for varying traffic demand density of streaming users. It can be observed that the maximum resource occupancy is achieved at nearly 25 Erlang with very high blocking probability. Whereas, for 2\% blocking probability requirement, the corresponding resource utilization is only 60\%.

Next, we explain the properties of resource utilization vector for a given set of active sectors.
\subsection{Average Resource Utilization Vector}
From and \eqref{average_bp_cell} and \eqref{occupancy_probability}, the feasible load vector that satisfies the blocking probability requirements (i.e. $P_{bj}\leq P_{b,\mathrm{max}}$) can be written as

\begin{align} \label{load_vector}
 \mathcal{F}\big(\boldsymbol{\beta}(\mathcal{B}_{\mathrm{on}},&\boldsymbol{\beta},\boldsymbol{\mathcal{P}},\boldsymbol{\phi},\boldsymbol{\mathcal{H}})\big) = \nonumber \\
\Bigg\{ \boldsymbol{\beta} \Bigg| \beta_j = &\frac{1}{N_{sc}}\frac{\sum_{c=0}^{N_{sc}} \sum_{l=1}^{N_L} \rho_{j} w_j(l). n_{sc}(l). g_j(c-n_{sc}(l))}{\sum_{c=0}^{N_{sc}}\frac{1}{c} \sum_{l=1}^{N_L} \rho_{j} w_j(l). n_{sc}(l). g_j(c-n_{sc}(l))},\nonumber \\ &  0 \leq \beta_j \leq 1-\epsilon, \ P_{bj}\leq P_{b,\mathrm{max}},  \ \forall j\in{\mathcal{B_{\mathrm{on}}}} 
\Bigg\}, 
\end{align}
where $\epsilon>0$ is an arbitrarily small positive number. The load vector is solution of the system
\begin{equation}\label{load_vector_solution}
\boldsymbol{\beta^*}(\mathcal{B}_{\mathrm{on}},\boldsymbol{\beta},\boldsymbol{\mathcal{P}},\boldsymbol{\phi},\boldsymbol{\mathcal{H}}) = \mathcal{F}\big(\boldsymbol{\beta}(\mathcal{B}_{\mathrm{on}},\boldsymbol{\beta},\boldsymbol{\mathcal{P}},\boldsymbol{\phi},\boldsymbol{\mathcal{H}})\big).
\end{equation}

\begin{myrem}
 The fraction of time-frequency resources utilized in $j$-th cell depends on the following factors: distribution of SINR ($Pr(\gamma_{ij} \ge \Gamma_{l},P_{r,ij}\ge P_{r,\mathrm{min}}),\ l = 1,2,...,N_L$) and traffic generated by the users ($\rho_j(l), \ l = 1,2,...,N_L$, resource utilization in neighbouring cells ($\beta_{j^{'}},j^{'}\neq j$). Note that the values of $w_j(l)$ are obtained from the distribution of SINR which takes into account shadowing and small scale fading, and activity status of interferers. Since the SINR in \eqref{sinr_load} is a function of activity status of interfering sectors, the resource utilization in cell $j$ is coupled with resource utilization in neighbouring cells through \eqref{cost_func_cov}, \eqref{w_j_eqn} and \eqref{occupancy_probability}.
\end{myrem}

\begin{mydef} A function $\boldsymbol{f}:\mathbb{R}_+^M \rightarrow \mathbb{R}_{++}$ is said to be standard interference function if it satisfies the following conditions \cite[Definition 1]{cavalcante2014}:
\begin{itemize}
\item \textit{Monotonicity: $\alpha \boldsymbol{f}(\boldsymbol{x}) > \boldsymbol{f}(\alpha \boldsymbol{x}), \ \forall \boldsymbol{x} \in \mathbb{R}_{+}^{M}, \ \forall \alpha > 1$}.
 \item \textit{Scalability: $\boldsymbol{f}(\boldsymbol{x_1})\ge \boldsymbol{f}(\boldsymbol{x_2})$ if $\boldsymbol{x_1} \ge \boldsymbol{x_2}$}.\\
\end{itemize}
\end{mydef}

\begin{mydef} Concave functions $\boldsymbol{f}:\mathbb{R}_{+}^{M} \rightarrow \mathbb{R}_{++}$ are standard interference functions \cite[Preposition 1]{cavalcante2014}.\\
\end{mydef}

To show the uniqueness of the solution to \eqref{load_vector_solution} we use the following lemma.
\begin{mylem}
\textit{For any cell $j$, the average resource utilization $\beta_j$ is strictly concave for $\beta_{j^{'}}, j^{'}\neq j, j^{'} \in \mathcal{B}_{\mathrm{on}}$}.
\end{mylem}

\begin{proof}
The proof is given in Appendix D.
\end{proof}

From lemma 1, we obtain the following theorems on the existence of the unique solution of the system \eqref{load_vector_solution}. \\
\begin{mythe} 
\textit{The resource utilization in vector $\beta_j$ is standard interference function.}
\end{mythe}
\begin{proof} 
 From Definition 2 and Lemma 2, it can be concluded that the resource utilization vector is standard interference function. 
\end{proof}

\begin{mythe} 
(\textit{Existence of unique fixed point}) \textit{
\begin{itemize}
 \item The standard interference mapping $\beta$ has a fixed point and the fixed point is unique.
 \item For an arbitrary vector $\boldsymbol{\beta}^{0} \in \mathbb{R}_{+}^{|\mathcal{B}_{\mathrm{on}}|}$, the sequence $\{\boldsymbol{\beta}^{k}\}_{k\in \mathbb{N}}$ converges to the fixed point $\boldsymbol{\beta}^{*}\in \mathrm{Fix}(f)$.
\end{itemize}}
\end{mythe}

\begin{proof} 
 Since $\beta$ is standard interference function, according to \cite[Fact 3]{cavalcante2014} the solution of \eqref{load_vector_solution} has a fixed point in $[0,1)^{|\mathcal{B}_{\mathrm{on}}|}$. The solution can be iteratively obtained using fixed point iteration method \cite{siomina2012}.
\end{proof}

\subsection{Area Power Consumption}
Power consumption at the BS comprises two parts: fixed power consumption and dynamic power consumption. Fixed power consumption is due to signal processing at the base band section and the RF section. Dynamic part of the power consumption is due to Power Amplifier (PA) section which varies with the fraction of time-frequency resources occupied (i.e. $\beta$). The power consumption at sector $j$ is modeled as \cite{auer2011}
\begin{equation}\label{power_consumption_eqn}
P_{Cj} = \frac{N_{\mathrm{TRX}}. (\beta_j . \frac{P_{tmax,j}}{\eta_{\mathrm{eff}}}+P_{RF}+P_{BB})}{(1-\sigma_{DC})(1-\sigma_{MS})(1-\sigma_{cool})},
\end{equation}
where $N_{\mathrm{TRX}}$ denotes the number of TRX chains of the sector $j$, $P_{tmax,j}$ denotes the maximum RF output power of sector $j$ at peak load, $\eta_{\mathrm{eff}}$ denotes the PA efficiency. $\sigma_{DC}$, $\sigma_{MS}$ and $\sigma_{cool}$ are the loss factors due to DC-DC power supply, main supply and cooling, respectively.

The APC (W/m$^2$) by all the active BSs ($\mathcal{B}_{\mathrm{on}}$) is obtained by 
\begin{equation}\label{cost_func_apc}
{f_{\mathrm{APC}}}(\mathcal{B}_{\mathrm{on}},\boldsymbol{\beta},\boldsymbol{\mathcal{P}},\boldsymbol{\phi},\boldsymbol{\mathcal{H}}) = \frac{1}{\mathcal{A}_{\mathbb{D}}} \sum_{j\in \mathcal{B}_{\mathrm{on}}} P_{Cj}.
\end{equation}

The objective functions: network coverage (Eqn. \eqref{cost_func_cov}), overlap (Eqn. \eqref{cost_func_ol}), ASE (Eqn. \eqref{cost_func_ase}), and APC (Eqn. \eqref{cost_func_apc}), and the blocking probability constraint (Eqn. \eqref{average_bp_cell}) are coupled with each other through ICI. The detailed analysis is presented in the next section. 

\subsection{Effect of traffic demand density on cost functions} \label{effect_objective}

\begin{table}[thb]
\centering
\caption{POWER CONSUMPTION PARAMETERS (PER SECTOR)}
\begin{tabular}{|p{4.5cm} | p{3cm}|}\hline 
\label{power_consumption_parameters}
\textbf{Parameter}  & \textbf{Value} \\
\hline $N_{\mathrm{TRX}}$ & 1 \\
\hline \textbf{PA power consumption:} &  \\
 $P_{\mathrm{max}}$ & 20.0 W (43 dBm) \\
Back-off & 8 dB \\
PA Efficiency $\eta_{\mathrm{eff}}$ & 31.1 \%\\ 
\textbf{Total PA} $P_{PA} = \frac{P_{\mathrm{max}}}{\eta_{\mathrm{eff}}}$ & \textbf{64 W}\\
\hline \textbf{RF power consumption:} & \\
$P_{TX}, P_{RX}$ & 6.8 W, 6.1 W\\
\textbf{Total RF} $P_{RF} = P_{TX}+P_{RX}$ & \textbf{13 W} \\
\hline \textbf{Baseband power consumption:} &  \\
$P_{BB}$ & \textbf{29.5 W} \\ 
\hline \textbf{Loss factors:} $\sigma_{DC}, \sigma_{MS}, \sigma_{cool}$  & 7.5\%, \ 9.0\%,\ 10.0\% \\
\hline \textbf{Total $P_{in}$} & \textbf{140.5 W} \\
\hline
\end{tabular}
\end{table}

\begin{figure*}[thb]
\centering
\hspace{-.45cm}
\subfigure[]{\includegraphics[scale=.36]{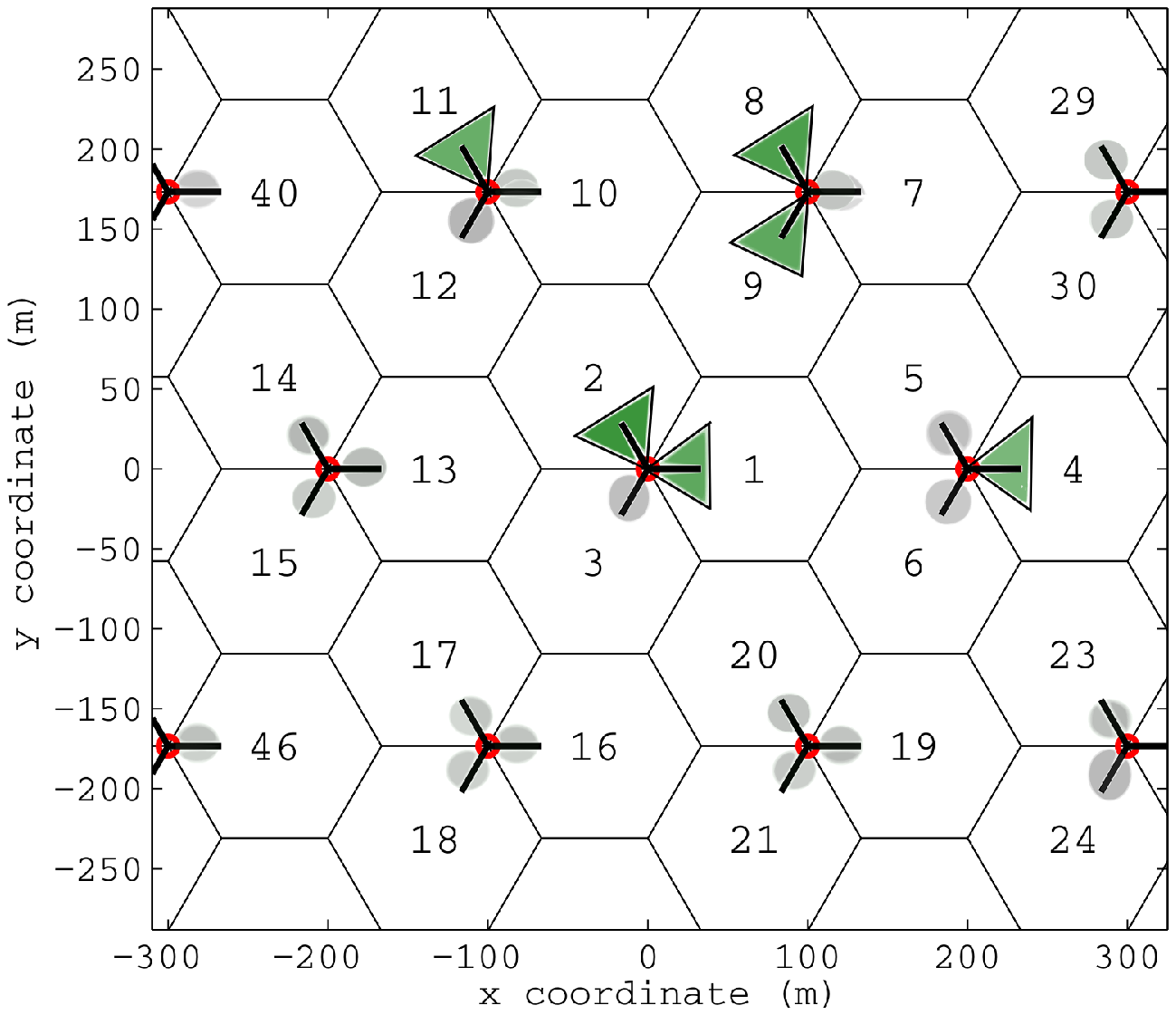}\label{Layout_UMi_UF_TWC}}
\subfigure[]{\includegraphics[scale=.38]{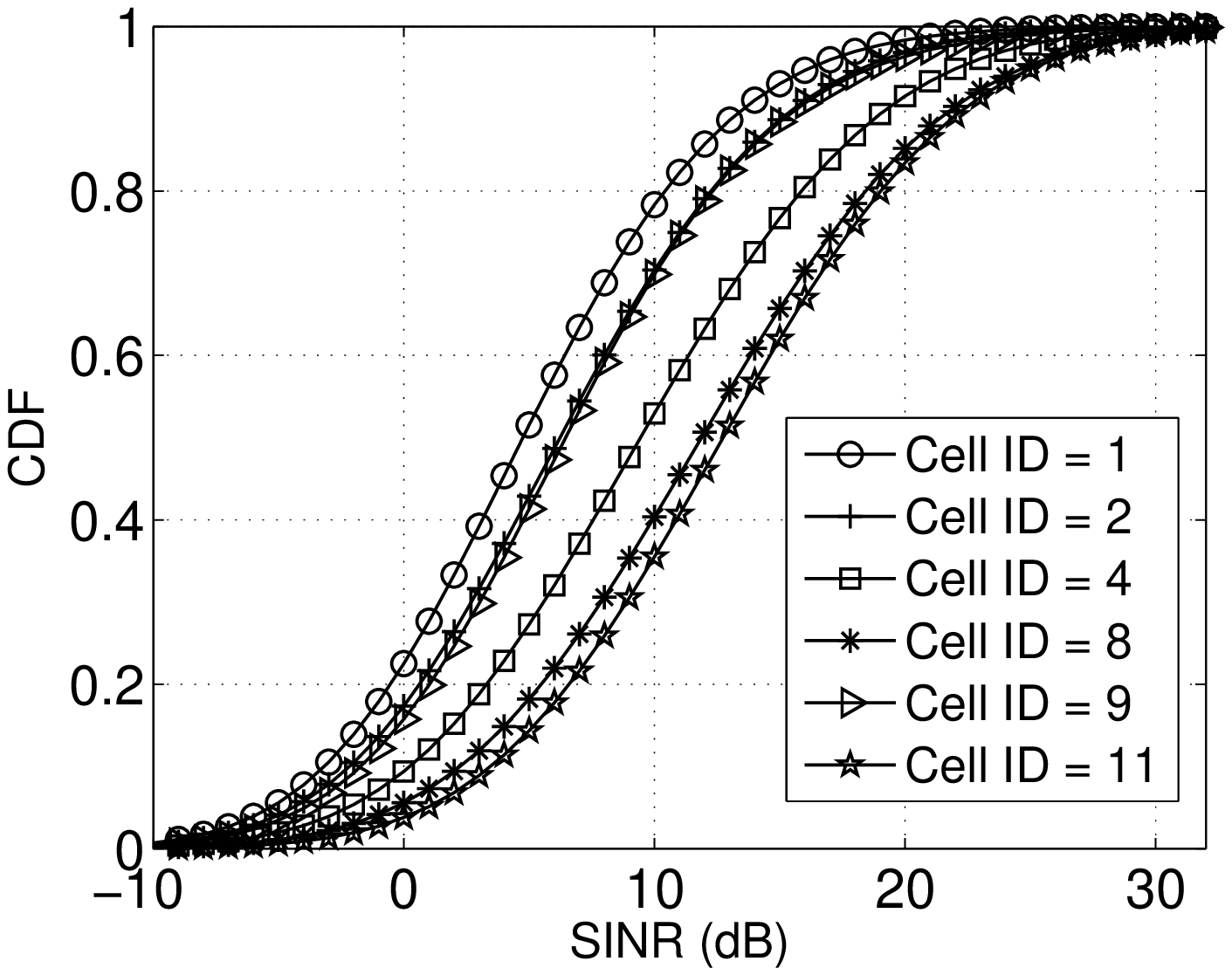} \label{SINR_Distribution_TWC}} \hspace{-.5cm}
\subfigure[]{\includegraphics[scale=.38]{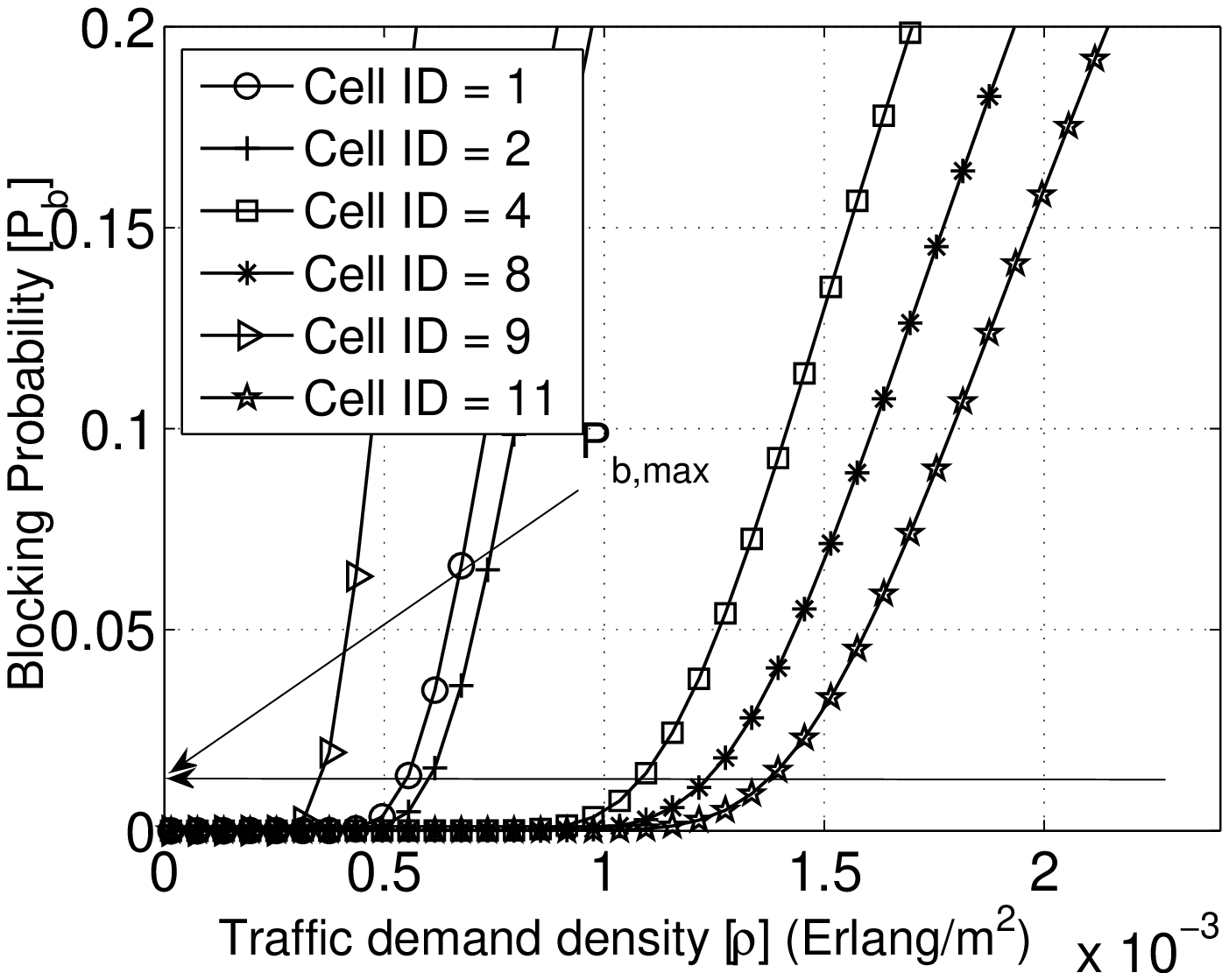} \label{BP_TWC}}\\ \hspace{-.95cm}
\subfigure[]{\includegraphics[scale=.38]{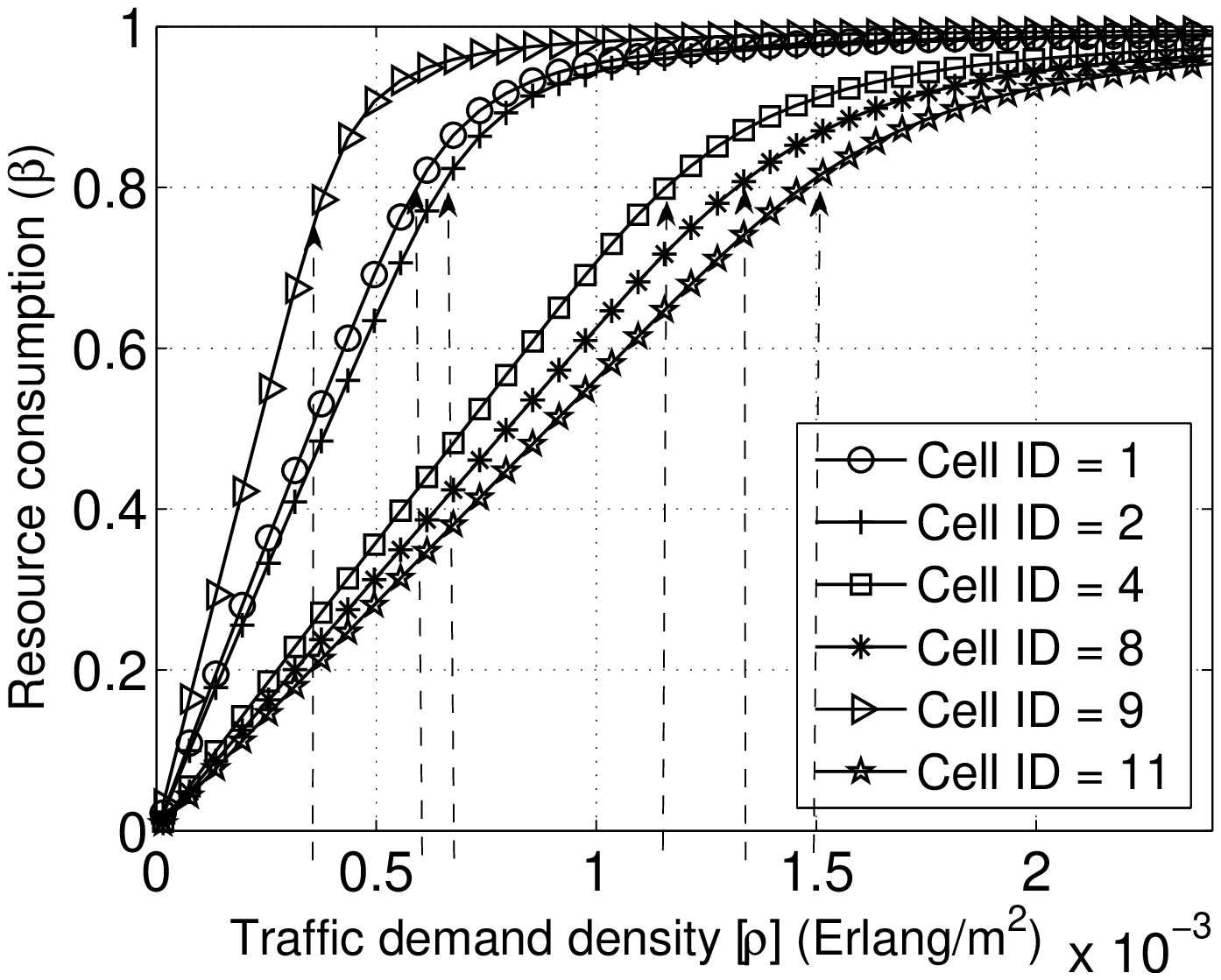} \label{RU_TWC}} \vspace{-.15cm} \hspace{-.5cm}
\subfigure[]{\includegraphics[scale=.38]{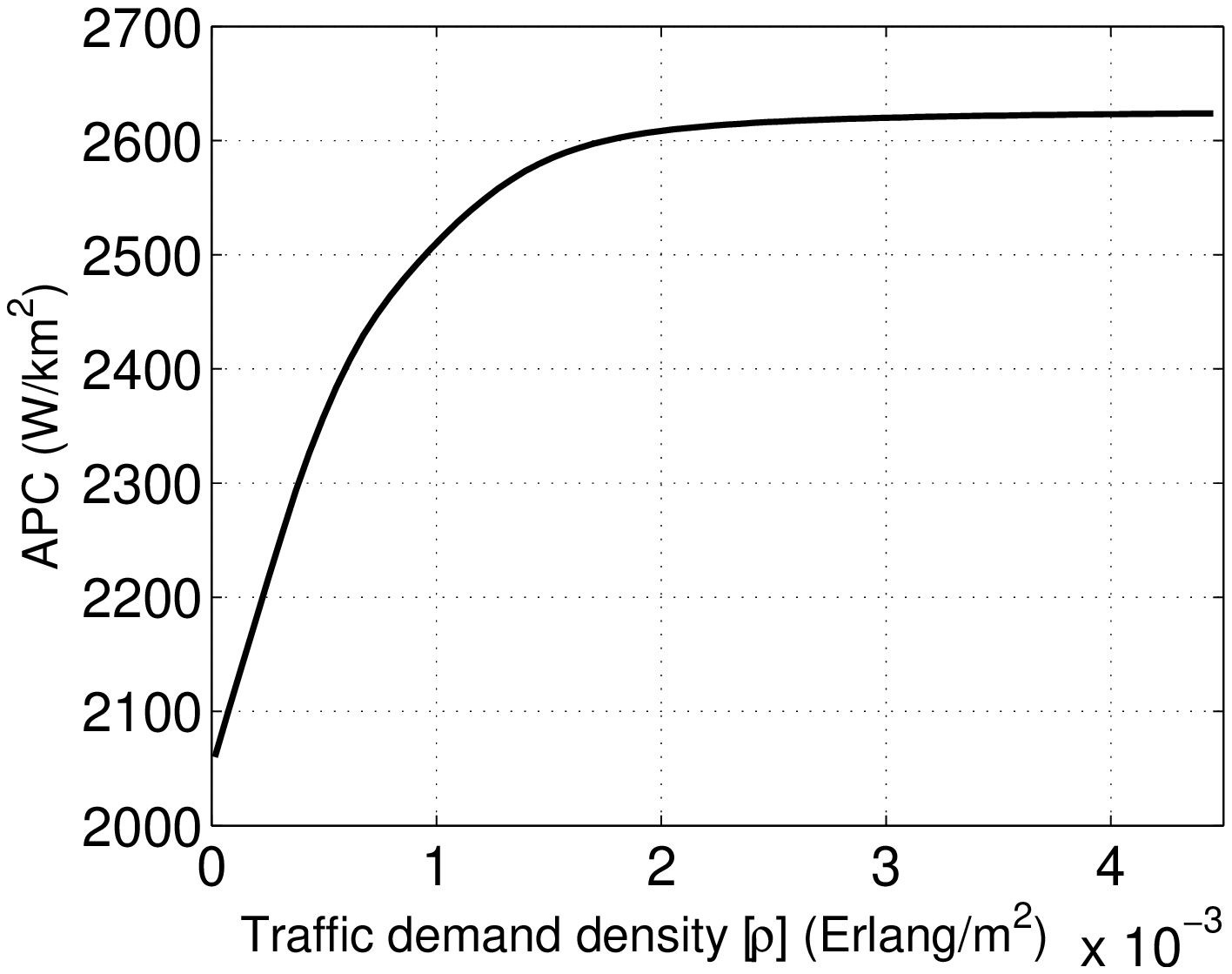}} \label{APC_TWC}\hspace{-.35cm}
\subfigure[]{\includegraphics[scale=.38]{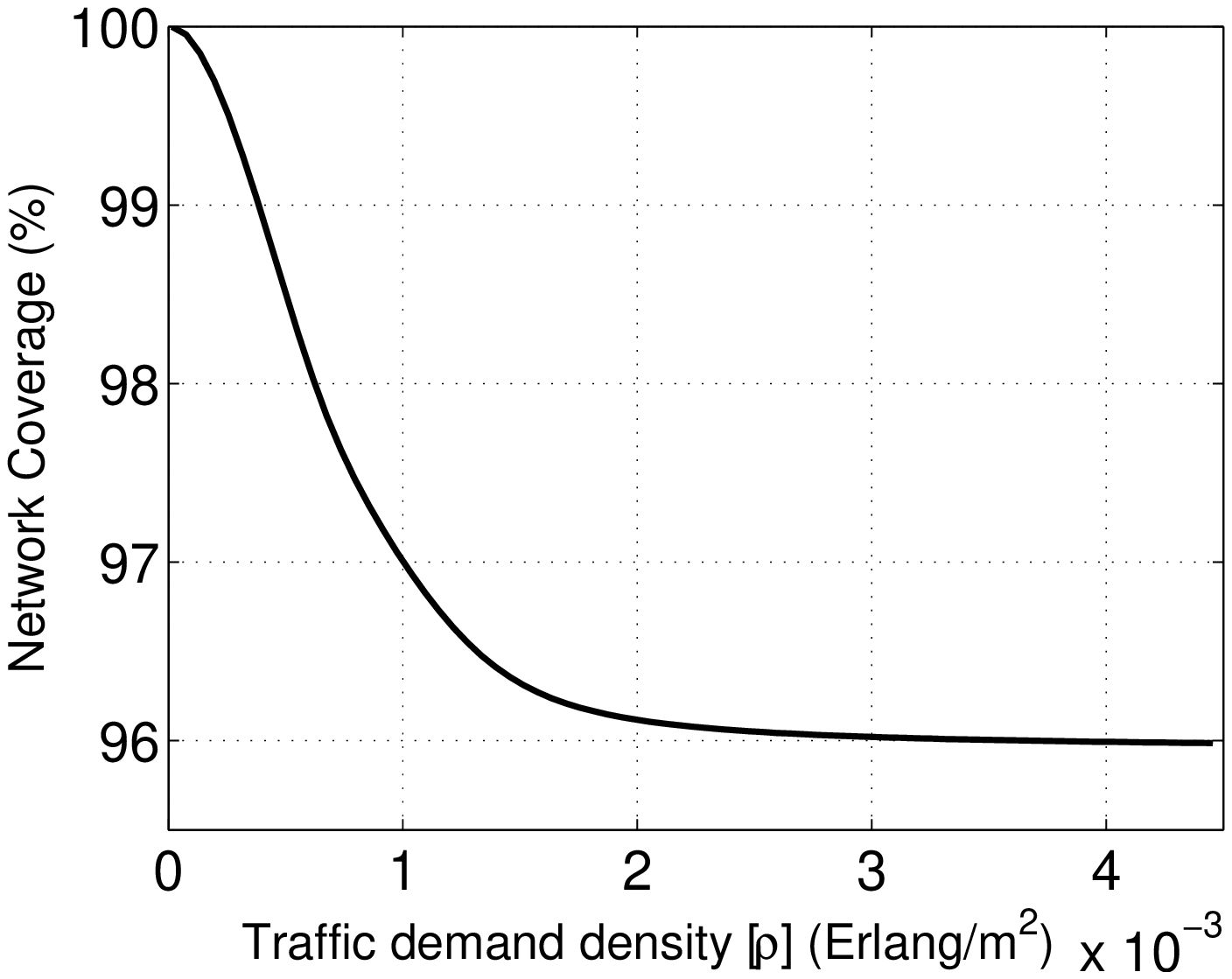} \label{NC_RxPR_TWC}}\hspace{-.55cm}
\caption{\scriptsize{Effect of traffic variation on network performance with a set of active sectors $\mathcal{B}_{\mathrm{on}}: \{1, 2, 4, 8, 9, 11\}$  (a) Urban Microcell network layout used for evaluation  (b) SINR distribution of individual cell (c) Blocking probability in an individual cell (d) Resource utilization in an individual cell (e) APC (f) Network coverage}}
\label{impact_traffic_demand}
\end{figure*}

In this section we show the effect of varying the traffic demand density $\rho$ on the blocking probability performance and the network level cost functions: ASE (${f_{\mathrm{ASE}}}$), APC (${f_{\mathrm{APC}}}$), \% coverage (${f_{\mathrm{COV}}}$), and \% overlap (${f_{\mathrm{OL}}}$). For analysis, a sample solution i.e. the sectors $\mathcal{B}_{\mathrm{on}}=\{1, 2, 4, 8, 9, 11\}$ as shown in Figure \ref{Layout_UMi_UF_TWC} are considered to be active ($\bigtriangledown$ and $\bigcirc$ represent active and inactive sectors, respectively).  The power consumption parameters shown in Table \ref{power_consumption_parameters} are used as given in \cite{auer2011}. The number of TRX chains is assumed to be one. The percentage of area covered by an individual sector is  $P_{Cov,1}= 20.21\%$, $P_{Cov,2}= 19.95\%$, $P_{Cov,4}= 11.66\%$, $P_{Cov,8}= 9.84\%$, $P_{Cov,9}= 30.05\%$ and $P_{Cov,11}=8.29\%$. Figure \ref{SINR_Distribution_TWC} shows the CDF of SINR of an individual cell. It can be seen that the SINR distribution varies from one cell to the other due to an uneven distribution of interference in the network. It can be seen that the cells $1, 2$ and $9$ experience poor SINR distribution due to strong co-channel interferers compared to cells $4, 8$ and $11$.

Figure \ref{BP_TWC} shows the average blocking probability $P_{bj}$ (from \eqref{average_bp_cell}) experienced by an individual cell for varying traffic demand density. It can be observed that the blocking probability in cells $1, 2$ and $9$ is high, even for low traffic density due to poor SINR distribution in those cells. Although the cells $2$ and $9$ experience similar SINR distribution, the blocking probability performance of cell $9$ is significantly worse than cell $2$. This is because as the sector $9$ covers a larger portion of the geographical area, the cumulative traffic demand experienced by the $9$-th cell is significantly higher than that of other cells which results in a high blocking probability. The resource utilization in an individual cell $\beta_j$ (obtained from \eqref{occupancy_probability}) for varying traffic demand density is shown in 
Figure \ref{RU_TWC}. It can be seen that the resource consumption at sectors $1, 2$ and $9$ is significantly high, even at low traffic demand densities due to poor SINR and larger coverage regions.

Figure 3(e) shows the APC for increasing traffic demand density. As the traffic demand increases, the resource consumption at individual sector is seen to increase resulting in increased dynamic part of power consumption and APC according to \eqref{power_consumption_eqn}. With further increase in traffic demand density, the APC is seen to saturate at a certain maximum value due to limitation of available resources at the sectors. Figure 3(f) shows the \% network coverage for varying traffic demand density. As the traffic demand increases the activity of sectors on the sub-channels ($v_{kg}$ in \eqref{sinr_load}) increases which translates into an increased ICI to neighbouring cells thereby decreasing the SINR performance. Therefore, the overall coverage decreases with increasing traffic demand density. It can be observed that the area covered by $3$ and $4$ sectors is significantly less than the area covered by $2$ sectors for the active sector set $\mathcal{B}_{\mathrm{on}}=\{1,2,4,8,9,11\}$.

It can be observed from Figure 2(b) that with $P_{b,\mathrm{max}}= 2\%$ blocking probability requirement, the maximum traffic demand density supported by the active sector set $\mathcal{B}_{\mathrm{on}}=\{1,2,4,8,9,11\}$ is $\rho_{\mathcal{B}_{\mathrm{on}}}^{\mathrm{max}}=.36\times 10^{-3}$ Erlang/m$^2$. Here, 
\begin{equation}
\rho_{\mathcal{B}_{\mathrm{on}}}^{\mathrm{max}} = \underset{j\in \{1,2,4,8,9,11\}}{\operatorname{argmin}}{\big\{P_{bj}^{-1}(P_{b,\mathrm{max}})\big\}}.\nonumber
\end{equation}
The corresponding traffic demand supported \big(i.e. $\rho_{j}^{\mathrm{max}} = \big\{\rho_{j}(\rho) \big|P_{bj}(\rho) = P_{b,\mathrm{max}} \big\}$\big) by the cells $1,2,4,8,9$, and {{$11$}} is $6,11.1,9.9,5.1,12,$ and $3.9$ Erlang, respectively. From Figure 2(c) it can be observed that the corresponding resource utilization in an individual cell is $.85,.8,.75,.82,.75$, and $.68$, respectively. These values are used in the calculation of power consumption at individual sector \eqref{power_consumption_eqn}. The APC at traffic demand density $\rho=.36\times10^{-3}$ Erlang/m$^2$ is $2300$ W/km$^2$.

Overall, from the results it can be concluded that with increasing traffic demand density the APC increases while the coverage and overlap performance decreases. Further, a particular set of active sectors supports up to a certain traffic demand density due to blocking probability constraint. Our objective is to select a active sector set which minimizes APC and overlap and maximizes coverage and ASE while satisfying the blocking probability requirements. For a given traffic demand density $\rho$, if a particular active sector set satisfies the blocking probability requirements it will be included in the potential solution set $Q_{\mathrm{Pot}}$. For example, the active sector set $\mathcal{B}_{\mathrm{on}}=\{1,2,4,8,9,11\}$ will be considered as a potential solution for all traffic demand densities less than $\rho_{\mathcal{B}_{\mathrm{on}}}=.36\times 10^{-3}$ Erlang/m$^2$.

\section{Traffic dynamics and Complexity}
For a given traffic demand density, finding the optimal set of sector and RAN parameters among a large number of combinations is a complex combinatorial problem. Due to its large scalability, dynamic optimization of mobile networks should be carried out according to the varying traffic load conditions, in a self organized manner, without any manual intervention. Traffic variations usually take place on an hourly basis. So, the required reconfiguration of the network needs to be done in small (from several minutes to few hours) as wells as large time scales (from few days to few months). The optimization method should be able to adapt these traffic fluctuations and provide the appropriate solutions with minimal computational complexity whenever it is necessary. The solutions can be obtained by optimizing all the objectives either jointly or individually according to the system requirements, complexity afford-ability, and need of the operator. Typically dynamic optimization of cellular networks involve the following situations:

\textit{a) Joint optimization of active BS set and RAN parameters considering all four objectives:} Joint optimization is required in case of initial cell planning as well as for re-planning due to introduction of addition of sites to serve an increased population in a given geographical region. Further, during instances such as: change of parameters (minimum received power threshold $P_{r,\mathrm{min}}$, minimum SINR threshold $\Gamma_{\mathrm{min}}$), addition of a new site, change in coverage, overlap or blocking probability requirements etc., it may be required to entirely reconfigure the network. In these situations complexity may not be a major issue as it is required to find the solutions in the large time scales. In case of joint optimization of active sector set and RAN parameters the search space is calculated as follows.

Let \begin{equation}
\mathcal{T}_{\mathcal{B}}=\{\mathcal{B}_{\mathrm{on}}^1,\mathcal{B}_{\mathrm{on}}^2,...,\mathcal{B}_{\mathrm{on}}^{(N_{\mathrm{SP}}^{\mathcal{T}_{\mathcal{B}}})}\}\end{equation}
be the search space which contains all possible active sector configurations with size \begin{equation}
N_{\mathrm{SP}}^{\mathcal{T}_{\mathcal{B}}} = |\mathcal{T}_\mathcal{B}| = 2^{N_{\mathcal{B}}}-1.\end{equation}
 For example, with six sectors the search space is 
\begin{equation}
\mathcal{T}_{\mathcal{B}}=\{000001,000010,000011...,100000\}\nonumber
\end{equation}
with length $2^{6}-1=63$. Let \begin{equation}
\phi_{\mathrm{tilt}j} = \{ \phi_{\mathrm{tilt}j_{\mathrm{min}}},...,\phi_{\mathrm{tilt}j_{\mathrm{max}}} \},\end{equation}
\begin{equation}
P_{tj} = \{P_{tj_{\mathrm{min}}},...,P_{tj_{\mathrm{max}}}\},\end{equation}
and \begin{equation}
H_{tj} =\{H_{tj_{\mathrm{min}}},...,H_{tj_{\mathrm{max}}} \}\end{equation}
be the set of values of tilt angles, transmit power, and heights, respectively, available at sector $j$. Let \begin{equation}
K_{P_t}^j = |P_{tj}|,\end{equation}
\begin{equation}
K_{\phi_{\mathrm{tilt}}}^j=|\phi_{\mathrm{tilt}j}|,\end{equation}
 and \begin{equation}
K_{H_t}^j=|H_{tj}|\end{equation}
 be the cardinality of $\phi_{\mathrm{tilt}j}$, $P_{tj}$, and $H_{tj}$, respectively. One particular combination of a set of RAN parameters is denoted as 
\begin{equation}
\mathcal{R} = \{X,Y,Z|X \in \phi_{\mathrm{tilt}j},Y \in  P_{tj}, Z \in H_{tj} \}, \ \mathcal{R} \in \mathcal{T}_\mathcal{R}(\mathcal{B}_{\mathrm{on}}). \nonumber
 \end{equation}
The search space for the RAN parameter of a particular set of active sectors $\mathcal{B}_{\mathrm{on}}$ is 
\begin{equation}\mathcal{T}_\mathcal{R}(\mathcal{B}_{\mathrm{on}}) = \{\mathcal{R}^1,\mathcal{R}^2,...,\mathcal{R}^{(N_{\mathrm{SP}}^{\mathcal{T}_\mathcal{R}}(\mathcal{B}_{\mathrm{on}}))} \}. \nonumber 
\end{equation}
Here the possible number of combinations of RAN parameters in the search space $\mathcal{T}_\mathcal{R}(\mathcal{B}_{\mathrm{on}})$ is equal to the product of number of elements in individual RAN parameter set i.e. 
\begin{equation}
N_{\mathrm{SP}}^{\mathcal{T}_\mathcal{R}}(\mathcal{B}_{\mathrm{on}}) = \prod_{j\in \mathcal{B}_{\mathrm{on}} } K_{P_t}^j K_{\phi_{\mathrm{tilt}}}^j K_{H_t}^j.\nonumber 
\end{equation}
Note that the antenna tilt angle can be adjusted either by mechanical or electrical tilt. In electrical tilt, the antenna pattern is adjusted without changing the physical angle of antenna. Whereas, mechanical tilt changes the physical angle of the antenna. 

Finally, the total search space length becomes
\begin{equation}
N_{\mathrm{SP}}^{\mathcal{T}_{\mathcal{B}},\mathcal{T}_\mathcal{R}} = \sum_{c=1}^{2^{N_{\mathcal{B}}}-1} \prod_{j\in \mathcal{B}_{\mathrm{on}}^{c} } K_{P_t}^j K_{\phi_{\mathrm{tilt}}}^j K_{H_t}^j.
\end{equation}

\textit{b) Individual optimization of active BS set and RAN parameters considering only a subset of objectives:} If the complexity is a major concern in smaller time scales then individual optimization is the best approach for obtaining the solutions. In this case search space is very much much less than the previous case. For example, it may be possible that the same number of BSs with different set of RAN parameters may be able to support the new traffic demand which is slightly higher than the previous demand. It is also possible that the BS set which provides maximum energy saving may not be able to provide sufficient coverage. In that case the coverage performance can be improved through RAN parameter optimization.

\subsection{Complexity of the Problem}
In this section we provide the details of complexity of the problem in \eqref{objective_function}.
\begin{mythe}
 \textit{The multi-objective optimization problem in \eqref{objective_function} is NP-hard.}
\end{mythe}
\begin{proof}
Consider a simplified problem of finding minimum number of BSs required to maintain network coverage $f_{{\mathrm{COV}}_{\mathrm{min}}}$. The problem is similar to the minimum disk cover problem where the set of points in the region $\mathbb{D}$ to be covered with a subset of disks $\mathbb{D}_j$ with minimal cardinality. It is known that the  problem is NP-hard \cite{chenyi2014}. The minimum disk cover problem is a subset of our problem in \eqref{objective_function}. Therefore, the considered multi-objective optimization problem is NP-hard. 
\end{proof}

In Section \ref{solution_approach}, we discuss the details of solving the multi-objective optimization problem defined in \eqref{objective_function} with reduced complexity by utilizing the traffic fluctuations.

\subsection{Solution Approach} \label{solution_approach}

\subsubsection{Pareto Optimal Region}
As it is seen in the previous section, the four objectives in \eqref{objective_function} are conflicting as they are coupled with each other due to the complex relationships among the variables. So, there can be no single solution that maximizes all the objectives simultaneously. In general, there is no global optimum to the multi-objective optimization problem in \eqref{objective_function}. Let 
\begin{equation}
\mathcal{G}=\big\{g(\textbf{x}),\forall \textbf{x}\big\}
\end{equation}
 be the objective set which contains all possible combinations of objective values. Therefore, our aim is to obtain the solution that provides the optimum trade-off among the conflicting objectives. Since the objective functions are conflicting, there will be multiple solutions forming a \textit{Pareto optimal set}. 

\begin{mydef}
A solution $\textbf{x}^*$ is said to be non-dominated (Pareto optimal) if, (i) There is no other solution dominating the objectives other than $\textbf{x}^*$. In other words, $g_q(\textbf{x}^*)<g_q(\textbf{w}), \ \forall q$ does not exist. (ii) The solution $\textbf{x}^*$ is strictly better than $\textbf{w}$ i.e. $g_q(\textbf{x}^*)>g_q(\textbf{w})$ for at least one objective $q\in \{1,2,3,4\}$. 
\end{mydef}

It can be said that a solution is Pareto optimal, if none of the objective functions can be improved in value without degrading some of the other objective values \cite{abraham2005}. Each solution in the Pareto optimal set has certain trade-offs between the objectives. As there will be multiple solutions, from the Pareto solution set, network operator has an opportunity to select appropriate BS and RAN parameter configuration according to their needs. 

There are two ways to solve the framed multi-objective optimization problem. In one approach, called sum of weighted objectives (SWO), all the objectives are added together with appropriate weight values to form a single objective as is done in our previous work \cite{prabhu2014icc}. This method is referred to as \textit{a priori method} as it is required to assign preferences to the weight vector beforehand. However, in order to obtain the best solution one has to find the appropriate weight vector through trial and error method which may take several trails. It is also possible that one may not be able to obtain the appropriate weight vector within the required time frame. therefore, in this work we use a different approach, called \textit{a posteriori method}. In this approach, first the set of Pareto optimal solutions are obtained and then the final solution is selected based on the preferences.

\begin{figure}[htb]
\centering
\includegraphics[trim={.1cm 4.5cm .1cm 1cm},clip,scale=.45]{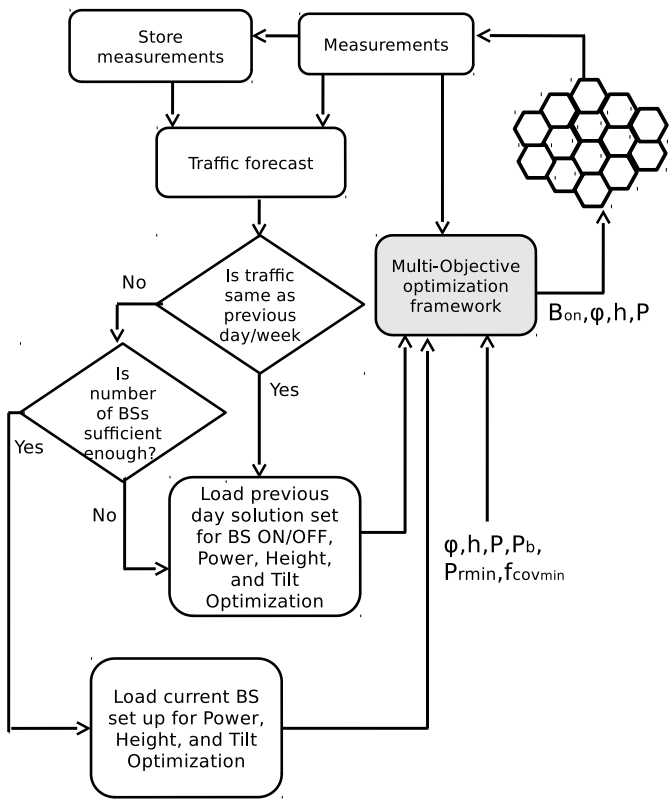}
\caption{System block diagram for dynamic optimization of the network}
\label{block_diagram}
\end{figure}

\subsection{Genetic Algorithm based meta-heuristic approach for finding Pareto optimal solutions}
Classical optimization techniques are difficult to apply for the considered combinatorial optimization problem due to non-linear relationships between the variables. Further, obtaining the global Pareto solution set is a challenging task. There are a number of search optimization approaches such as simulated annealing, Tabu search, ant colony optimization et. for finding the approximate Pareto optimal set. The main disadvantage of these approaches is that they often get stuck at local solutions and do not provide a global Pareto set. Whereas, Genetic Algorithm (GA) is well known for obtaining the global Pareto optimal solution of the multi-objective optimization problems with much lower complexity \cite{abraham2005}. This is because GA processes a group of solutions in the search space unlike other methods which process a single solution at a time. Complexity analysis of GA for the current problem is discussed in Section \ref{GA_complexity}. In addition, GA does not require complex mathematical functions, rather it requires numerical values also called fitness functions. The speed of GA can also be improved by parallel implementation of GA \cite{deb2001multi}.

In context to the present problem being discussed, the most important advantage of GA is its suitability in applications to dynamic environments.
There are two different approaches that can be used for dynamically adapting GA: \textit{search based} and \textit{memory based} approaches \cite{lam2005,mori2000}. The details of these approaches are detailed in the next section. The dynamics of GA can be implemented both in \textit{online} mode and \textit{offline} mode.

\begin{algorithm*}[htb]
 \caption{Algorithm for finding Pareto optimal solutions}
 \label{algorithm_ES}
 \begin{algorithmic}[1]
 \footnotesize{
 \STATE $\mathcal{Q} \leftarrow \mathcal{T}$ \COMMENT{Binary string assignment}
 \STATE $N_{\mathrm{Init}}$: Initial population size; $p_{co}$: Crossover probability; $p_{mu}$: Mutation probability; $M$: Number of objectives; $s$: Front number
 \STATE $\mathcal{Q} = \{\mathcal{Q}_1,\mathcal{Q}_2,...,\mathcal{Q}_{N_{\mathrm{SP}}} \}$ 
 \STATE $\mathcal{Q}_{\mathrm{Pareto}} \leftarrow \emptyset$
 \STATE $\mathcal{Q}_{\mathrm{Init}} \leftarrow$ Randomly select $N_{\mathrm{Init}}$ chromosomes from $\mathcal{Q}$
 \FOR{$x=1$ \TO $x\le N_{\mathrm{Iter}}$}
 \STATE $\mathcal{Q}_{\mathrm{Pot}} \leftarrow \emptyset$
 \FOR[For all the solutions in the initial population set]{$n=1$ \TO $n \le N_{\mathrm{Init}}$}
 \STATE \underline{Estimate Load Vector: Algorithm 2}
 \STATE $g_1(n)\leftarrow - {f_{\mathrm{APC}}}$; $g_2(n)\leftarrow {f_{\mathrm{ASE}}}$; $g_3(n) \leftarrow {f_{\mathrm{COV}}}$;$g_4(n)\leftarrow - {f_{\mathrm{OL}}};$
 \COMMENT{Compute fitness values: $g_q(n),\forall q$}
 \IF[Constraint check]{$P_{bj} \leq P_{b,\mathrm{max}}, \ \forall j$} 
 \STATE $\mathcal{Q}_{\mathrm{Pot}} \leftarrow \mathcal{Q}_{\mathrm{Pot}} \cup \mathcal{Q}_{\mathrm{Init}}^n$ \COMMENT{Update potential solution set}
 \ENDIF
 \ENDFOR
  \STATE $s \leftarrow 1; n_s \leftarrow 0$ \COMMENT{Initialize front number}
  \WHILE{$\mathcal{Q}_{\mathrm{Pot}} \neq \emptyset$}
  \STATE $\mathcal{Q}_{\mathrm{Dom}} \leftarrow \emptyset;\mathcal{Q}_{\mathrm{NDom}} \leftarrow \emptyset;$  \COMMENT{Update population size}
  \FOR[For all the solutions in the potential population]{$n=1$ \TO $n \leq |\mathcal{Q}_{\mathrm{Pot}}|$}
  \FOR{$m=1$ \TO $m \leq |\mathcal{Q}_{\mathrm{Pot}}|$} 
  \IF{$n\neq m$}
  \IF{$\big(f_q(n)>f_q(m), q \in \{1,2,3,4\}\big) \ \& \ \big(f_q(n)\nleq f_q(m), \ \forall q \big)$} 
  \STATE $\mathcal{Q}_{\mathrm{NDom}}^{x,s} \leftarrow \mathcal{Q}_{\mathrm{NDom}}^{x,s} \cup  \mathcal{Q}_{\mathrm{Pot}}^n$ \COMMENT{Update  $s$-th non-dominated solution set}
  \STATE $n_s \leftarrow n_{s}+1$ \COMMENT{Increment number of elements in $s$-th front}
  \STATE $\mathcal{Q}_{\mathrm{Pot}} \leftarrow \mathcal{Q}_{\mathrm{Pot}}\setminus \mathcal{Q}_{\mathrm{Pot}}^n$
  \ELSE
  \STATE $\mathcal{Q}_{\mathrm{Dom}}^x \leftarrow \mathcal{Q}_{\mathrm{Dom}}^x \cup \mathcal{Q}_{\mathrm{Pot}}^n$ \COMMENT{Update dominated solution set}
  \STATE $\mathcal{Q}_{\mathrm{Pot}} \leftarrow \mathcal{Q}_{\mathrm{Pot}}\setminus \mathcal{Q}_{\mathrm{Pot}}^n$
  \ENDIF
  \ENDIF
  \ENDFOR
  \ENDFOR
  \STATE $\mathcal{Q}_{\mathrm{Pot}} \leftarrow \mathcal{Q}_{\mathrm{Dom}}$; \ $df_u^s \leftarrow f_\mathrm{dum}, \ u=\{1,2,...,n_s\}$ \COMMENT{Assign dummy fitness value} 
  \STATE \textbf{\underline{SHARING}:} 
  \FOR[For all the solutions in the $s$-th front]{$u=1:n_s$} 
  \FOR{$w=1:n_s$}
  \STATE $\sigma_{\mathrm{sh}} \leftarrow \frac{.5}{\sqrt[4]{n_s}}$; \ $f_q^{\mathrm{max}} \leftarrow \mathrm{max}(f_q(n))$; \ $f_q^{\mathrm{min}}=\mathrm{min}(f_q(n))$, \ $n=1,2,...,|\mathcal{Q}_{\mathrm{Pot}}|$
  \STATE $\mathrm{d_{uw}} \leftarrow \sqrt{\sum_{q=1}^{4}\bigg(\frac{f_q(u)-f_q(w)}{f_q^{\mathrm{max}}-f_q^{\mathrm{min}}}\bigg)^2}$ ; \COMMENT{Calculation of distance between $u$ and $w$}
  \IF{$\mathrm{d_{uw}}\leq \sigma_{\mathrm{sh}}$}
  \STATE $\kappa(\mathrm{d_{uw}}) \leftarrow 1-\big(\frac{\mathrm{d_{uw}}}{\sigma_{\mathrm{sh}}}\big)^2$;  \COMMENT{Calculation of sharing function $\kappa(\mathrm{d_{uw}})$}
  \ELSE
  \STATE $\kappa(\mathrm{d_{uw}}) \leftarrow 0$ 
  \ENDIF
  \ENDFOR
  \STATE $m_u^s \leftarrow \sum_{u=1}^{n_s}\kappa(\mathrm{d_{uw}})$ \COMMENT{Calculation of niche count $m_u^s$}
  \STATE $df_u^{s} \leftarrow \frac{df_u^s}{m_u}$; \ \ $df_w^{s} \leftarrow \underset{u}{\operatorname{min}}\{df_u^s\}$; \ \ $df_u^{s+1} \leftarrow df_w^{s}-\epsilon_{s}$
   \COMMENT{Calculation of dummy fitness values for $s$-th front}
  \ENDFOR
  \STATE $s \leftarrow s+1$; \COMMENT{Update front number}
  \ENDWHILE
  \STATE $\mathcal{Q}_{\mathrm{Pareto}} \leftarrow \mathcal{Q}_{\mathrm{Pareto}} \cup \mathcal{Q}_{\mathrm{NDom}}^x$ \COMMENT{Update Pareto solution set}
 \STATE $\mathcal{Q}_{\mathrm{Pot}} \leftarrow$ \textit{reproduce}($\mathcal{Q}_{\mathrm{Pot}}$) ;\ $\mathcal{Q}_{\mathrm{Pot}} \leftarrow$ \textit{crossover}($\mathcal{Q}_{\mathrm{Pot}},p_{co}$); \ $\mathcal{Q}_{\mathrm{Pot}} \leftarrow$ \textit{mutate}($\mathcal{Q}_{\mathrm{Pot}},p_{mu}$)  \COMMENT{Perform selection, crossover and mutation operations}
 \STATE $\mathcal{Q}_{\mathrm{Init}} \leftarrow \mathcal{Q}_{\mathrm{Pot}}$ \COMMENT{Update the initial population for next iteration}
 \ENDFOR
 \STATE $\mathcal{B}_{\mathrm{on}} \leftarrow \mathcal{Q}_{\mathrm{Pareto}}$} 
 \end{algorithmic} 
 \end{algorithm*}
 
\subsection{Dynamic Network Optimization Architecture}
Fig. \ref{block_diagram} illustrates the block diagram of the proposed system architecture where the network is controlled in a centralized manner. All the eNBs are connected to the central controller via. At first, the central controller creates a data base by running the optimization algorithm in \textit{offline} mode. The algorithm for finding Pareto optimal solution set is detailed in Section \ref{algorithm_pareto_details}. The data base contains the set of Pareto optimal solutions for different traffic demand density $\rho$. Let $\rho_{\mathcal{B}}^{\mathrm{max}}$ be the peak traffic load that the network can support and 
\begin{equation}
\hat{\rho}=\frac{\rho}{\rho_{\mathcal{B}}^{\mathrm{max}}} \nonumber
\end{equation}
 be the normalized load. The number of sectors required for any traffic load $\hat{\rho}$ should be such that 
\begin{equation}
N_{\mathcal{B}_{\mathrm{min}}}\leq N_{\mathcal{B}_{\mathrm{on}}}(\hat{\rho}) \leq N_{\mathcal{B}}, \nonumber
\end{equation}
 where $N_{\mathcal{B}_{\mathrm{min}}}$ is the minimum number of sectors required to 
satisfy the minimum coverage requirements irrespective of the traffic conditions in the network. Consider that the normalized traffic load, $\hat{\rho}$ is quantized into $N_q$ discrete levels, 
\begin{equation}
\hat{\rho} \in\{ \hat{\rho}_1,\hat{\rho}_2,...,\hat{\rho}_{N_q}\}.\nonumber
\end{equation}
Let $\bar{N}_{\mathcal{B}_{\mathrm{on}}}(\hat{\rho}_q)$ be the number of sectors required to support the traffic load $\hat{\rho}_q$.

The number of solutions in the sub-population the set of active sectors with $\bar{N}_{\mathcal{B}_{\mathrm{on}}}(\hat{\rho}_q)$ sectors is 
\begin{equation}
N_{\mathrm{SP}}^{\bar{N}_{\mathcal{B}_{\mathrm{on}}}(\hat{\rho}_q)} = \binom{N_{\mathcal{B}}}{\bar{N}_{\mathcal{B}_{\mathrm{on}}}(\hat{\rho}_q) }= \frac{N_{\mathcal{B}}!}{\bar{N}_{\mathcal{B}_{\mathrm{on}}}(\hat{\rho}_q) ! (N_{\mathcal{B}} - \bar{N}_{\mathcal{B}_{\mathrm{on}}}(\hat{\rho}_q) )!}, \nonumber
\end{equation}
$q = 1,2,3,...,N_q$. The sub-population that contains the solution set that is required for supporting the traffic load $\hat{\rho}_q$ is 
\begin{align}
\mathcal{T}_{\mathcal{B}_{\hat{\rho}_q}}&= \nonumber \\
 \Bigg\{&\mathcal{B}_{\mathrm{on}}^n \in \mathcal{T}_\mathcal{B}\Bigg|\nonumber \\&  \bar{N}_{\mathcal{B}_{\mathrm{on}}}(\hat{\rho}_q)-\sigma_{N_{\mathcal{B}_{\mathrm{on}}}(\hat{\rho}_q)} \le |\mathcal{B}_{\mathrm{on}}^n| \le \bar{N}_{\mathcal{B}_{\mathrm{on}}}(\hat{\rho}_q)+\sigma_{N_{\mathcal{B}_{\mathrm{on}}}(\hat{\rho}_q)}, \nonumber \\
& n =1,2,...,2^{N_{\mathcal{B}}-1}\Bigg\}.\nonumber
\end{align}

By assuming $\bar{N}_{\mathcal{B}_{\mathrm{on}}}(\hat{\rho}_q) = \hat{\rho}_q . N_{\mathcal{B}}$, for $N_{\mathcal{B}}= 50$ and $25$, the number of combinations $N_{\mathrm{SP}}^{\bar{N}_{\mathcal{B}_{\mathrm{on}}}(\hat{\rho}_q)}$ is shown in Figure \ref{numberofCombinations_TWC}. It can be seen that the number of combinations is large during medium load conditions and small. However, when the traffic load is $20\%$ (or $60\%$), the size of sub-population $|\mathcal{T}_{\mathcal{B}_{\hat{\rho}_q}}|$ is approximately $300$ times less than the total search space length $N_{\mathrm{SP}}^{\mathcal{T}_{\mathcal{B}}}$. This significant reduction in search space greatly helps to reduce the computational complexity and leads to fast convergence.

For continuous adaptation, the central controller frequently collects measurement reports (instantaneous traffic, call drops, etc.) from all the active BSs and estimate the traffic using the traffic forecast algorithm. If the traffic load at a particular time of the day is the same as that of previous day, then the solutions of previous day is used as a candidate solution set i.e. memory based approach \cite{lam2005,mori2000}. However, if the traffic condition is unpredictable and different from previous day traffic, then a search based \textit{random immigrants} method \cite{Yaochu2005} is used. In this method, for each generation the worst individuals are replaced by the randomly generated individuals in order to increase the diversity in the population. Incorporating dynamic features of GA can be used to cope up with the short term and long term variations. In such circumstances, random immigrants method can be used.

\begin{figure}[htb]
\centering
\includegraphics[scale=.55]{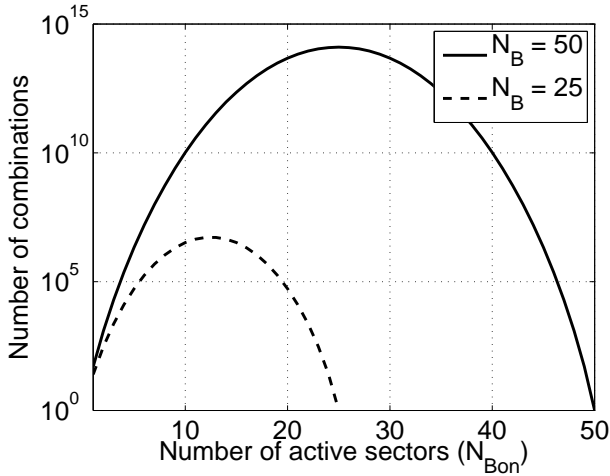}
\caption{Number of combinations in the search space vs Number of active sectors}
\label{numberofCombinations_TWC}
\end{figure}

\subsection{Algorithm to find Pareto optimal solutions}\label{algorithm_pareto_details}
The details of algorithm to find the Pareto optimal solutions given in Algorithm \ref{algorithm_ES} are as follows. The working procedure of GA is motivated by the biological evolutionary principles such as genetics and natural selection. In GA, each individual (solution) in the search space is represented as a binary string called \textit{chromosome}. Let $N_{\mathrm{SP}}$ be the search space length. First, all solutions in the search space are indexed to a binary bit string with length $\log_2 (N_{\mathrm{SP}})$ bits. For example, when the number of sectors $N_{\mathcal{B}}$ is equal to 3, the total search space, $\mathcal{T}_\mathcal{B} =\big\{ \{1 \},\{2 \},\{3 \},\{1,2 \},\{1,3\},...,\{1,2,3\} \big\}$ is coded as $\mathcal{Q} =\big\{[001],[010],[100],[011],[101],...,$ $[111]\big\}$. After that a sub population $\mathcal{Q}_{\mathrm{Init}}\in \mathcal{Q}$ whose length is much less than the search space ($N_{\mathrm{Init}}<<N_{\mathrm{SP}}$) is randomly chosen from the search space. Then fitness values of randomly selected population of chromosomes are calculated.  Then fitness values for all the solutions in the initial population are evaluated at the solutions of \eqref{load_vector_solution}. For a given set of active sectors and the corresponding RAN parameters, the load vector is estimated using Theorem 2. The details of the procedure is given in Algorithm \ref{algorithm_load_estimation}.

\begin{algorithm}[!htpb]
 \caption{Algorithm for finding load vector}
 \label{algorithm_load_estimation}
 \begin{algorithmic}[1]
 \STATE $k \leftarrow 1; \boldsymbol{\beta}(k) \leftarrow \boldsymbol{\beta}^{0}$; \COMMENT{Initialize}
  \WHILE{$\boldsymbol{\beta}(k)-\boldsymbol{\beta}(k-1) > \boldsymbol{\epsilon}$}
  \STATE $\boldsymbol{\beta}(k+1) = r^{(k)}\boldsymbol{\beta}(k) +(1-r^{(k)}) \boldsymbol{\beta}(k)$
  \STATE $k \leftarrow k+1$; \COMMENT{Update iteration number}
  \ENDWHILE
 \end{algorithmic} 
 \end{algorithm}

Next, constraint check is performed for all the chromosomes in the initial sub population ($\mathcal{Q}_{\mathrm{Init}}$) and only the solutions which satisfy the blocking probability requirements for a given traffic demand are selected as potential population ($\mathcal{Q}_{\mathrm{Pot}}$) for the next generation or iteration. The solutions in the potential population are ranked using non dominated sorting (NDS). 

\textbf{Step 14 to 45:} In NDS, ranking of the solutions is done based on the non-domination level. Based on the non-domination level Pareto fronts are formed \cite{srinivas1994}. Each solution in the set $\mathcal{Q}_{\mathrm{Pot}}$ is compared with all other solutions using the conditions for non-domination for all four objectives. A solution $u^*$ is marked as non-dominated if,
\begin{enumerate}
\item There is no other solution dominating the objectives other than {{$u^*$}}. In other words, $f_q(u^*)<f_q(w), \ \forall q$ does not exist.
\item The solution $u^*$ is strictly better than $v$ i.e. $f_q(u^*)>f_q(w)$ for at least one objective $q\in \{1,2,3,4\}$. 
\end{enumerate}
The solutions which are satisfying the conditions for non-domination are marked as non-dominated solutions of the first ($s=1$) non-dominated front i.e. $\mathcal{Q}_{\mathrm{NDom}}^{1}$. Let $n_s$ be the number of solutions in $s$-th non-dominated front. Then a dummy fitness value \begin{equation}
df_u^s=f_\mathrm{dum}\end{equation}
is assigned to all $n_s$ solutions in the first non-dominated front. After assigning a fitness value, sharing is performed to maintain diversity in the population. After sharing, the solutions in the first non-dominated front are temporarily ignored. Then the above procedure is repeated to find the second non-dominated levels. Second front is assigned a dummy fitness value lesser than front one and sharing is applied again. This procedure is repeated till all the members of the population assigned a shared fitness value. 

\textbf{Step 31 to 43:} Sharing is applied to maintain diversity in the search space. Sharing value is calculated between individuals ($\mathrm{u,v}$) in each front using the formula in line 33. Where, $\mathrm{d}_{\mathrm{uv}}$ is the distance between two individuals and $\sigma_{\mathrm{sh}}$ is the size of the niche. The parameters $f_q^{\mathrm{min}}$ and $f_q^{\mathrm{max}}$ are the minimum and maximum fitness values, respectively of the objective $q$. Niche count $m_u^s$ is used to spread individual along the Pareto front based on the sharing value \cite{srinivas1994}. The fitness value of the solution $u$ in $s$-th front is modified as 
\begin{equation}
df_u^{s} = \frac{df_u^s}{m_u^s}. \nonumber 
\end{equation}
After sharing, the worst fitness value ($df_w^{s}$) in the $s$-th non-dominated front is used as a dummy fitness value of $(s+1)$-th non-dominated front i.e. 
\begin{equation}
df_u^{s+1} = df_w^s-\epsilon_s,\nonumber 
\end{equation}
 where $\epsilon_s$ is a small positive number. Then the non-dominated solutions in the $s$-th front ($\mathcal{Q}_{\mathrm{NDom}}^s$) are temporarily ignored and the remaining solutions are further processed for finding next ($s+1$-th) non-dominated front. This procedure is continued till all members of the population are assigned a fitness value.

\textbf{Step 47:} After assigning dummy fitness value to each solution in the population, the following operations are performed: reproduction (or selection), crossover, and mutation \cite{deb2001multi}. Stochastic proportionate selection method \cite{deb2001multi} is used for reproducing the best individuals. Let $f_{\mathrm{avg}}$ be the average fitness of all the individuals, then the individual with fitness value $f_u$ gets an expected number of copies, $\frac{f_u}{f_{\mathrm{avg}}}$. Since the solutions in the first non-dominated Pareto front have better fitness values, the reproduction probability is more for the solutions in the first non-dominated front than the solutions in the remaining fronts. In crossover, solutions are chosen in pairs based on crossover probability $p_{co}$. After crossover, mutation is performed with probability, $p_{mu}$ to keep diversity in the solutions. 

Then stochastic proportionate selection method \cite{deb2001multi} is used for reproducing the best individuals. Only the best solutions are stored for further processing. The dominated or worst solutions are discarded. New set of chromosomes called \textit{offspring} are generated by performing: \textit{selection} (or \textit{reproduction}), \textit{crossover} and \textit{mutation}. The genetic operators are bit wise operations used to form new better chromosomes (The detailed discussion on GA operators can be found in \cite{mitchell1996}). Then the current population of chromosomes are replaced by new set of chromosomes. Each iteration of the above process is called \textit{generation}. This process is repeated for multiple generations. As the number of generations increases the above mentioned GA operators guides the search towards optimal solution set. The entire set of generations is called \textit{run}. At the end of each generation, the Pareto optimal solutions are compared with the solutions of the previous generation. The solutions which are non-dominated by any other solutions in the next generation are stored to form the final Pareto optimal set $\mathcal{Q}_{\mathrm{Pareto}}$. The final Pareto optimal set is selected as a potential solution based on the operator's requirements. The details of the selection of the final solution is detailed in Section \ref{results_Pareto_optimal}.

\subsection{Complexity of the Algorithm}\label{GA_complexity}
The total complexity of the Algorithm \ref{algorithm_ES} for finding active sector set at normalized traffic load $\hat{\rho}_s$ can be expressed as
\begin{align}\label{computational_complexity} 
O\Big( G(\hat{\rho}_s).\Big[M.N_{\mathrm{Init}}. \! \big[&O(\mathrm{Fitness})+p_{co}. O(Cr)+p_{mu}.O(Mu)\big] \nonumber \\ & + O(NSGA)  + O(SH) + O(SEL)\Big] \Big). 
\end{align}
Here $N_{\mathrm{Init}}$ is the initial population length, $G(\hat{\rho}_s)$ is the number of iterations required for convergence when traffic load is at $\hat{\rho}_s$, $M$ is the number of objectives, $O(\mathrm{Fitness})$ is complexity of fitness evaluation, $O(NSGA)$ is complexity of NSGA. In NSGA, each solution is compared with every other solution for $M$ different cost functions and is repeated for $N_{\mathrm{Front}}$ number of fronts. Hence the total complexity of NSGA is $O(MN_{\mathrm{Init}}^2 N_{\mathrm{Front}})$. In sharing, since each solution is compared with other solutions the complexity of sharing is $O(SH) = O(N_{\mathrm{Init}}^2)$ \cite{deb2002}. $O(SEL)$ is complexity of selection operation i.e. $O(N_{\mathrm{Init}})$ and $O(Cr)$ and $O(Mu)$ are complexity of crossover and mutation operator, respectively. It can be observed from \eqref{computational_complexity} that the complexity mainly depends on the number of individuals in the initial population $N_{\mathrm{Init}}$ and number of iterations required for convergence $G(\hat{\rho}_s)$. The complexity will be high during medium load conditions as the number of iterations required for convergence is higher than that of low and high load conditions. However, it is much less than the complexity of exhaustive search method. Note that the complexity of exhaustive search to find the Pareto optimal set involves the computation of $N_{\mathrm{SP}}$ fitness functions and sorting which is practically infeasible. The convergence of Algorithm \ref{algorithm_ES} can be made further faster by adjusting crossover and mutation probabilities in each iteration \cite{mitchell1996}. 

\subsection{Practical Implementation of the proposed framework}

The centralized self organizing network (SON) functionalities such as load balancing, hand-over parameter optimization, interference control, capacity and coverage optimization, etc. are being considered in LTE networks \cite{3gpp.32.500}. Further, centralized ES functionalities have been added in 3GPP standard \cite{3gpp.32.551}. Therefore, the sectors and the corresponding RAN parameters can be identified for SLM by exploiting the features of centralized ES and SON functionalities.

\section{Results and Discussion}\label{results}

\begin{figure*}[htb]
\centering
\psfragscanon
\psfrag{y}{\hspace{-.6cm}\scriptsize{$({f_{\mathrm{APC}}},{f_{\mathrm{ASE}}},{f_{\mathrm{COV}}},{f_{\mathrm{OL}}})=$}}
\psfragscanoff
\subfigure[]{\includegraphics[scale=.42]{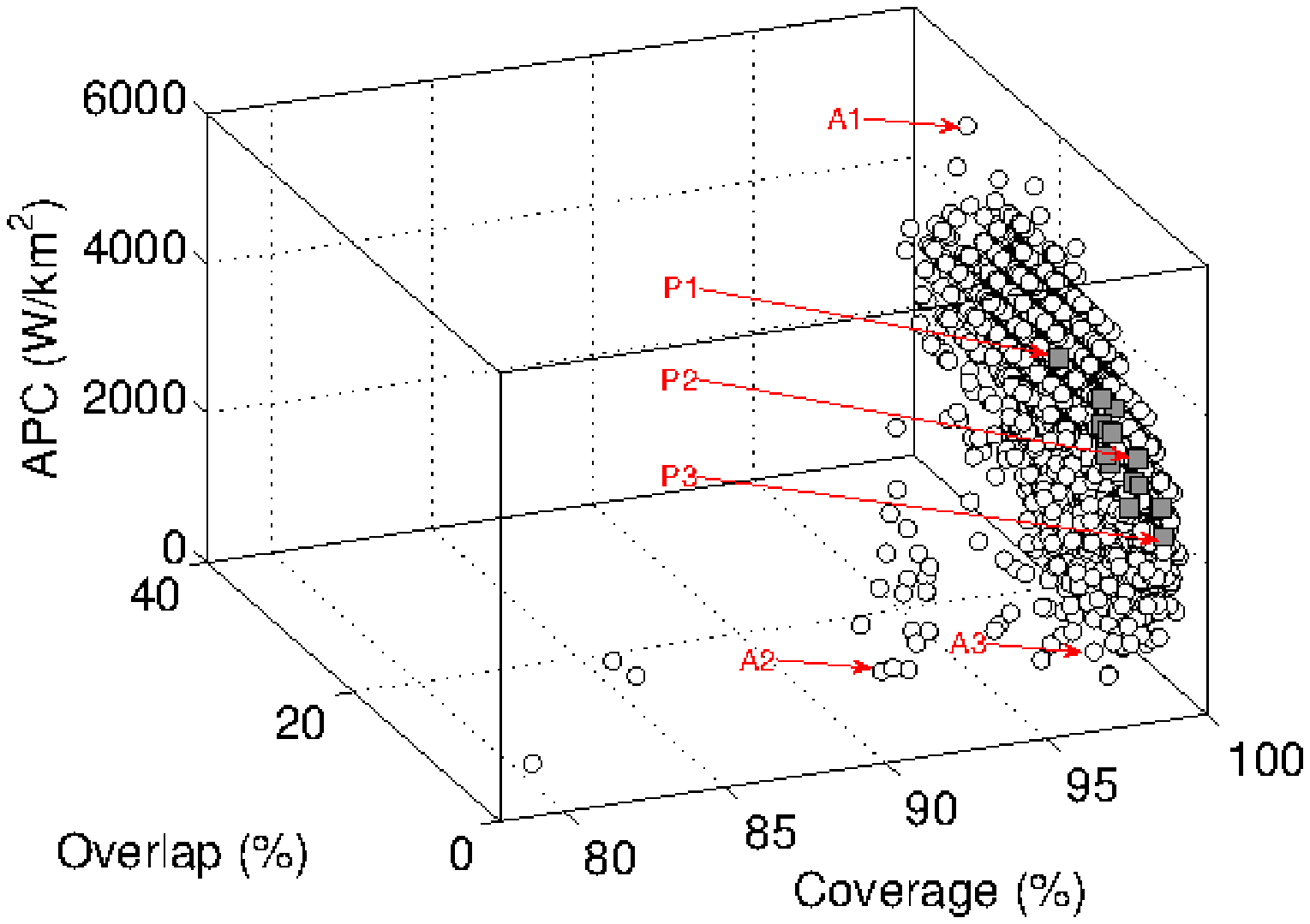}\label{pareto_front}}
\hfil
\hspace{-.75cm}
\subfigure[]{\includegraphics[scale=.42]{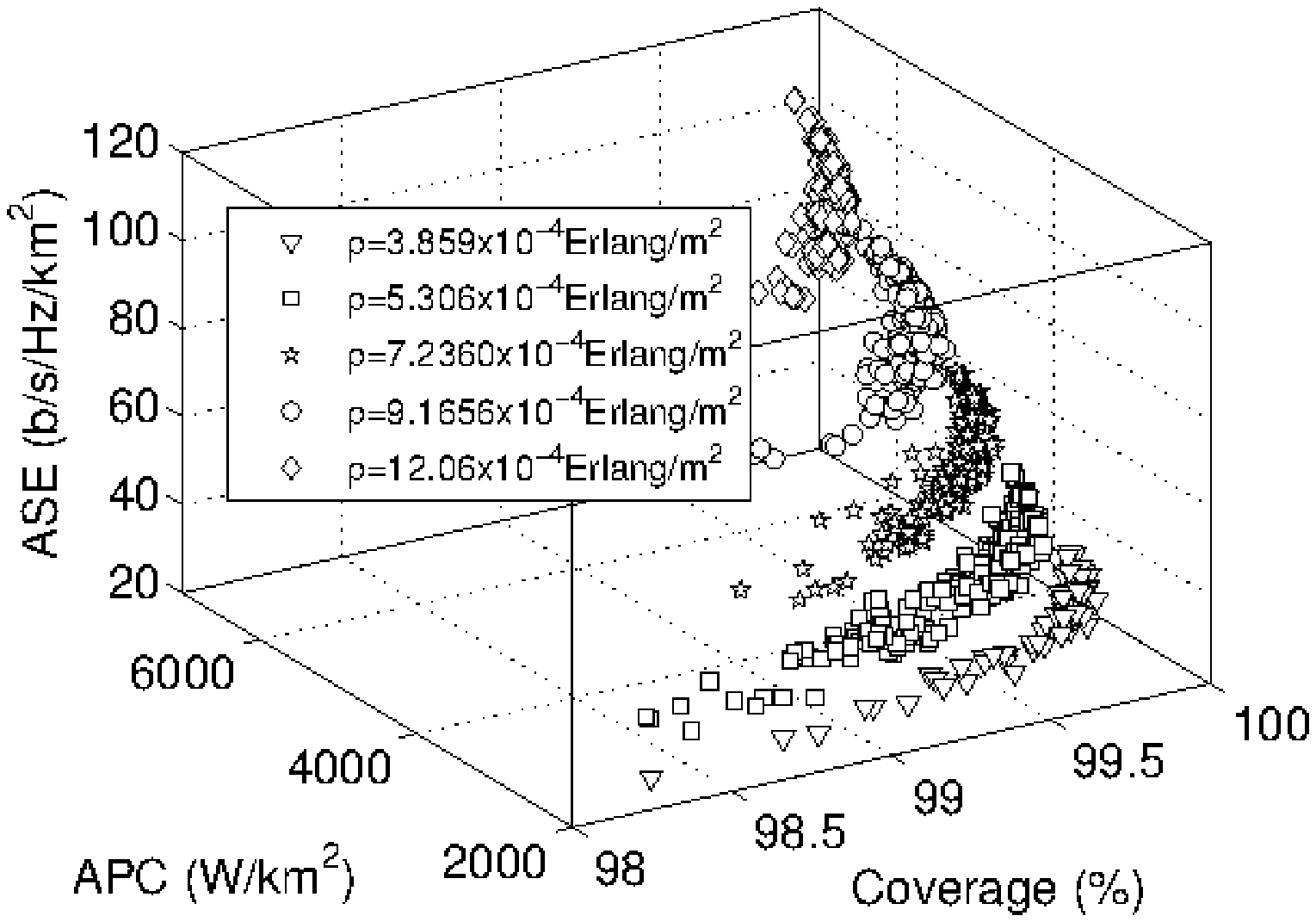}\label{APC_ASE_COV}}
\caption{Pareto optimal solutions (a) Values of cost functions for traffic demand density $\rho=3.376\times10^{-4}$ Erlang/m$^2$\newline
\textbf{A1}: $\mathcal{B}_{\mathrm{on}}=\{1,2,3,5,6,8,12,13,17,18,20,21\}$, ${f_{\mathrm{APC}}}=5038$ (W/km$^2$), ${f_{\mathrm{COV}}}=99.99\%$, ${f_{\mathrm{OL}}}=32.70\%$, ${f_{\mathrm{ASE}}}=35.60$ b/s/Hz/km$^2$ \newline
\textbf{A2}: $\mathcal{B}_{\mathrm{on}}=\{6,9,10,16,17,18\}$, \hspace{1.875cm} ${f_{\mathrm{APC}}}=1067$ (W/km$^2$), ${f_{\mathrm{COV}}}=90.26\%$, ${f_{\mathrm{OL}}}=01.67\%$, ${f_{\mathrm{ASE}}}=27.73$ b/s/Hz/km$^2$ \newline
\textbf{A3}: $\mathcal{B}_{\mathrm{on}}=\{1,5,6,9,10,11,14,19,20\}$, \hspace{.875cm} ${f_{\mathrm{APC}}}=1405$ (W/km$^2$), ${f_{\mathrm{COV}}}=97.22\%$, ${f_{\mathrm{OL}}}=33.28\%$, ${f_{\mathrm{ASE}}}=24.02$ b/s/Hz/km$^2$ \newline
\textbf{P1}: $\mathcal{B}_{\mathrm{on}}^*=\{1,2,5,6,7,8,9,10,11,17,20\}$, \hspace{.485cm} ${f_{\mathrm{APC}}}=3060$ (W/km$^2$), ${f_{\mathrm{COV}}}=99.86\%$, ${f_{\mathrm{OL}}}=19.68\%$, ${f_{\mathrm{ASE}}}=36.16$ b/s/Hz/km$^2$ \newline
\textbf{P2}: $\mathcal{B}_{\mathrm{on}}^*=\{4,5,6,8,10,17,18\}$,   \hspace{1.72cm} ${f_{\mathrm{APC}}}=2730$ (W/km$^2$),\ ${f_{\mathrm{COV}}}=99.76\%$, ${f_{\mathrm{OL}}}=08.25\%$, ${f_{\mathrm{ASE}}}=37.45$ b/s/Hz/km$^2$ \newline
\textbf{P3}: $\mathcal{B}_{\mathrm{on}}^*=\{5,6,8,9,10,11,17,18\}$, \hspace{1.3cm} ${f_{\mathrm{APC}}}=2070$ (W/km$^2$), ${f_{\mathrm{COV}}}=99.55\%$, ${f_{\mathrm{OL}}}=03.83\%$, ${f_{\mathrm{ASE}}}=37.63$ b/s/Hz/km$^2$ \newline
(b) Cost functions of the Pareto optimal solutions for different traffic demand density}
\end{figure*}

\begin{figure*}[htb]
\begin{center}
\hspace{-.94cm}
\subfigure[]{\includegraphics[scale=.37]{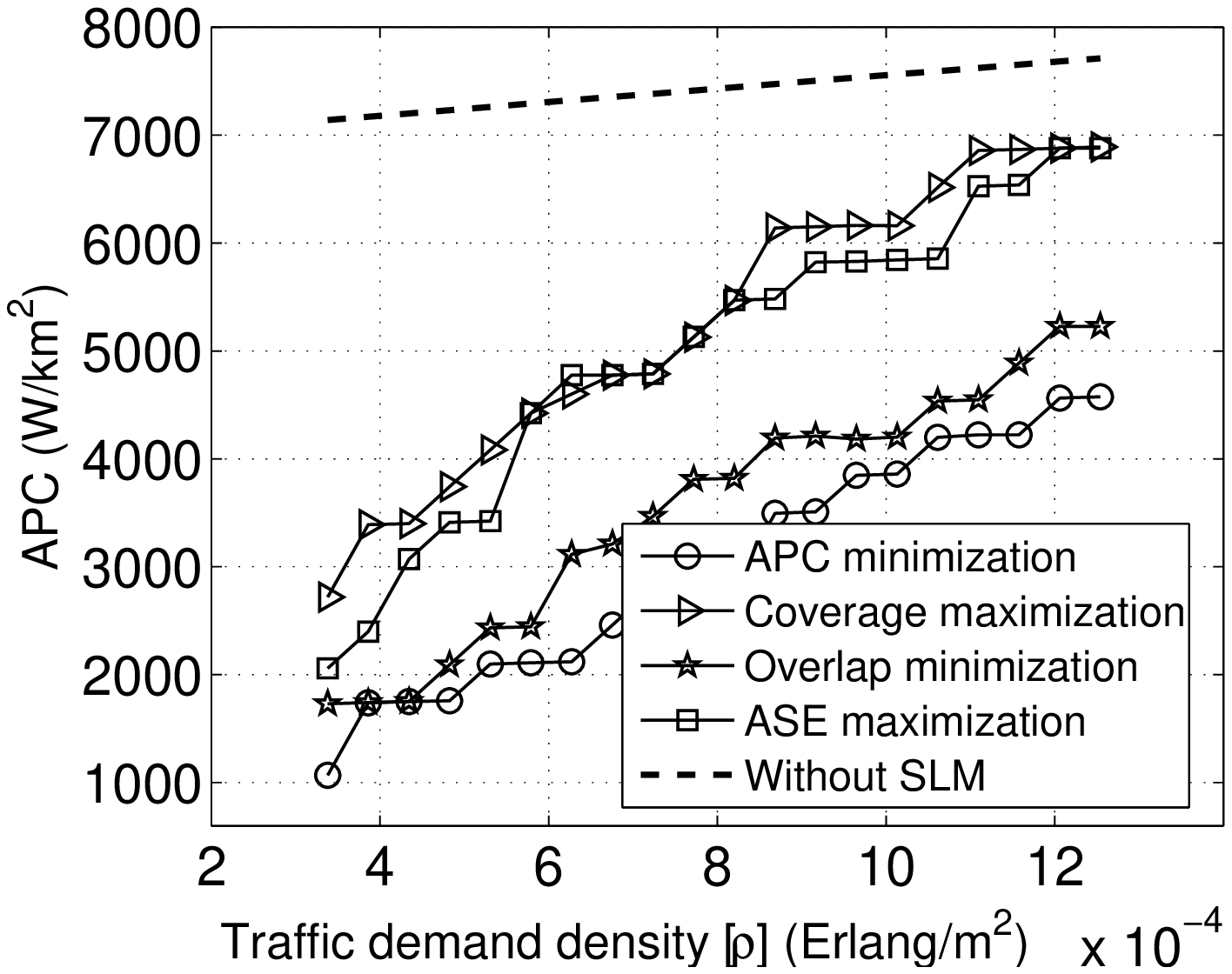}\label{APC_Final_TWC}}
\hspace{-.25cm}
\subfigure[]{\includegraphics[scale=.37]{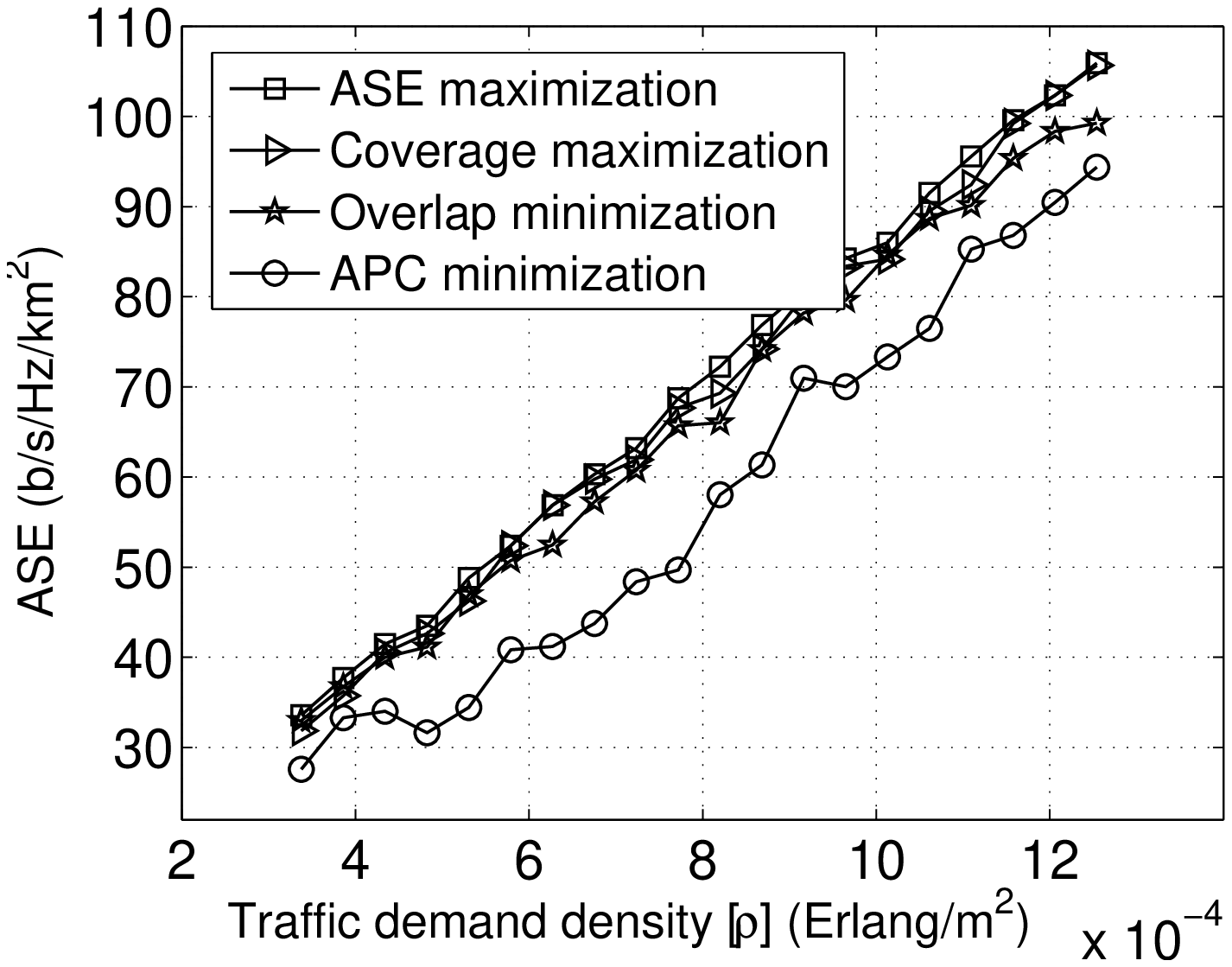}\label{ASE_vs_Erlang_TWC}}
\vspace{-.75cm}
\subfigure[]{\includegraphics[scale=.37]{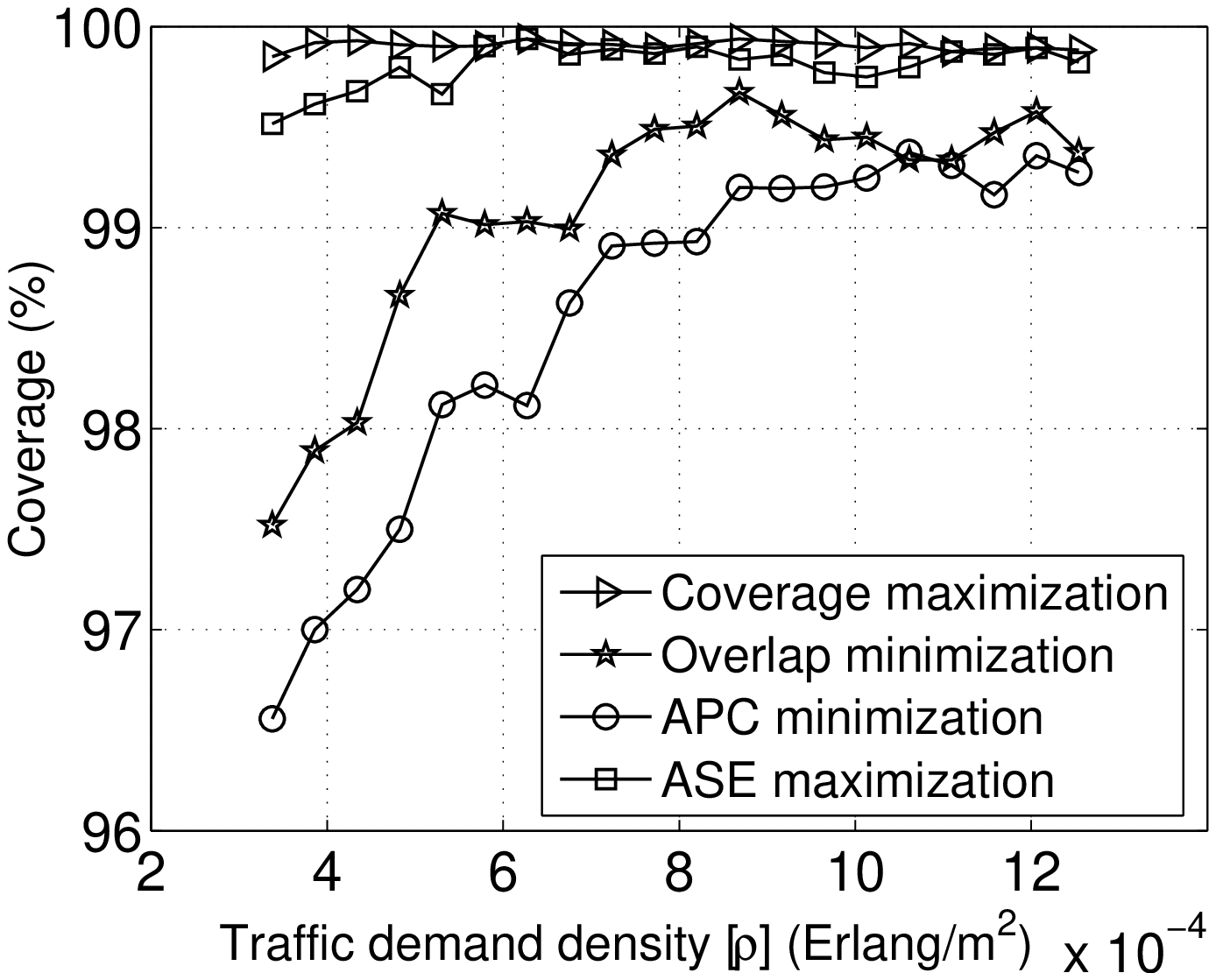}\label{Cov_Final_TWC}}\\ \vspace{.5cm}
\subfigure[]{\includegraphics[scale=.37]{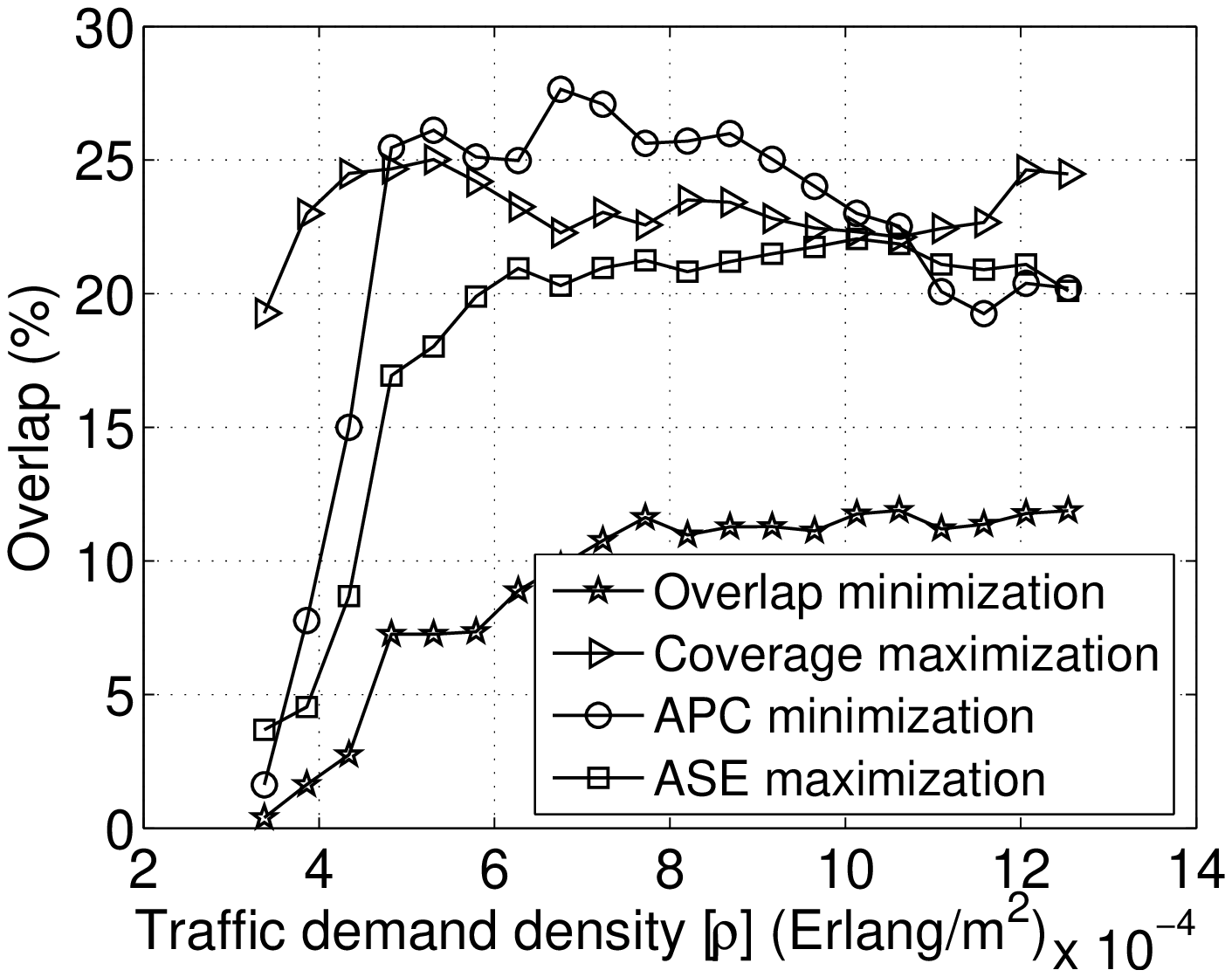}\label{Overlap_Final_TWC}} 
\subfigure[]{\includegraphics[scale=.37]{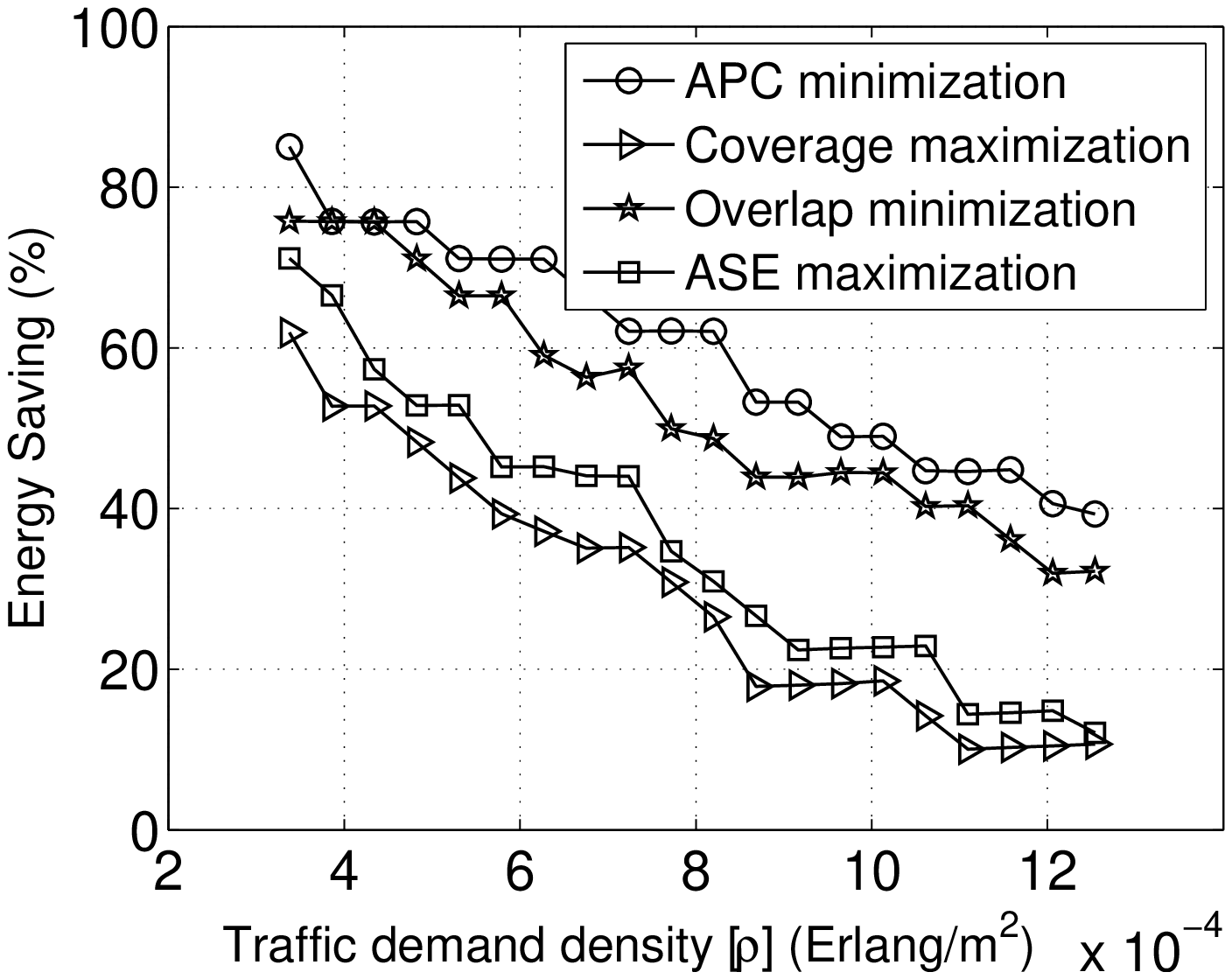}\label{ES_Final_TWC}}
\caption{Impact of selection of the final solution on the cost functions and ES for varying traffic demand density (a) APC (b) ASE (c) Coverage (d) Overlap vs  (e) ES}
\label{Cost_functions_different_objectives}
\end{center}
\end{figure*}

\subsection{Pareto optimal solutions obtained using NSGA}\label{results_Pareto_optimal}
In this section we show the Pareto optimal set of active sectors with fixed RAN parameter configurations (as given in Table I(a)) for different traffic demand density. For GA, the initial population size $N_{\mathrm{Init}}$ is taken as $100$. The crossover probability $p_{co}$ and mutation probability $p_{mu}$ are taken as $.7$ and $.01$, respectively. Figure \ref{pareto_front} shows the values of cost functions and Pareto optimal solutions $\mathcal{B}_{\mathrm{on}}^*$ obtained using the procedure as described in Section \ref{algorithm_pareto_details} for traffic demand density $\rho=3.376\times 10^{-4}$ Erlang/m$^2$. The circles represent the solutions of the entire search space and the solid squares represent the front one of the Pareto solution set obtained at the end of iteration number $4000$. It can be seen from the figure that the Pareto optimal solutions are concentrated towards the maximum fitness value of coverage, minimum fitness value of overlap and minimum fitness value of APC. 

The solutions A1, A2, and A3 and P1, P2, and P3 are the three sample solutions taken from the entire search space and from the Pareto optimal set, respectively. Though the solutions A2 and A3 provide minimum APC, their coverage and ASE performance is poorer than that of remaining solutions due to which they are not captured in the Pareto optimal solution set. The solution P3 provides minimum APC and overlap while the solution P1 provides slightly higher APC and overlap. Since each of the solutions in the Pareto set has certain trade-offs between the cost functions it gives the opportunity to choose a desired solution according to operator's need. It can also be seen that the Pareto solution set is spread across the solution space due to the sharing operation used.

\begin{figure*}[thb]
\subfigure[]{\includegraphics[scale=.24]{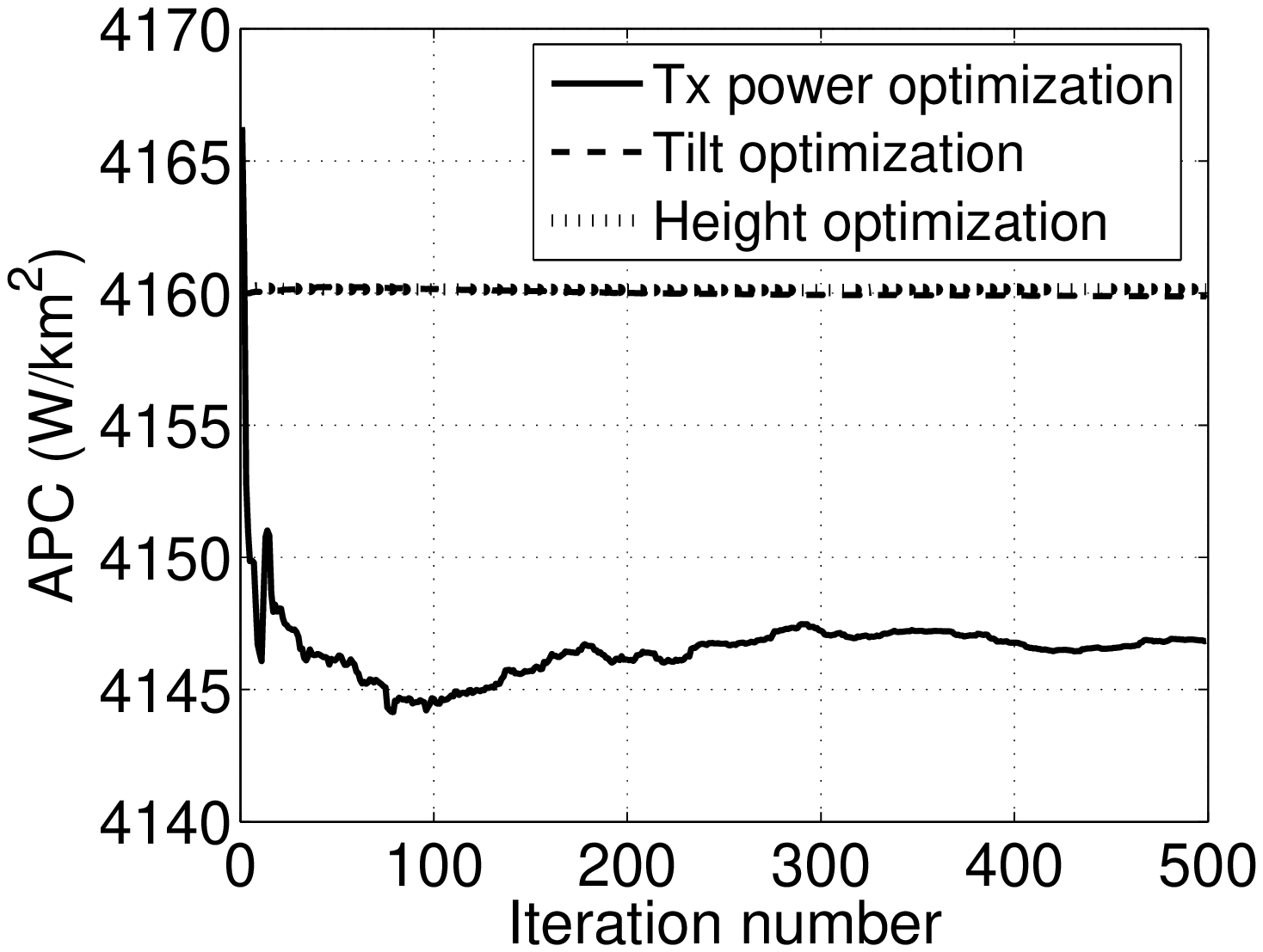}\label{Iteration_no_vs_APC_RAN}}\hspace{-.385cm}
\subfigure[]{\includegraphics[scale=.24]{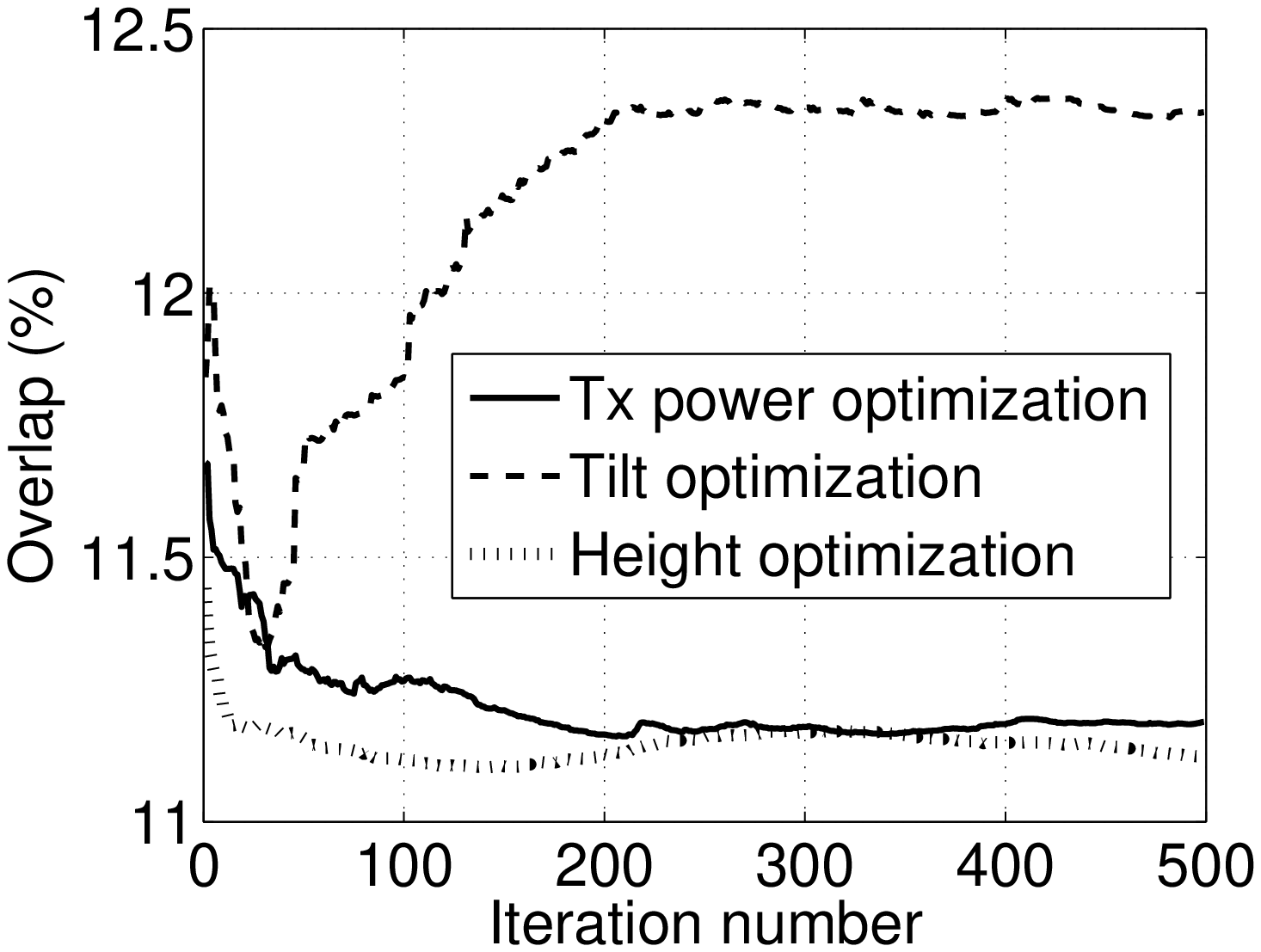}\label{RAN_OL_TWC}} \hspace{-.385cm}
\subfigure[]{\includegraphics[scale=.24]{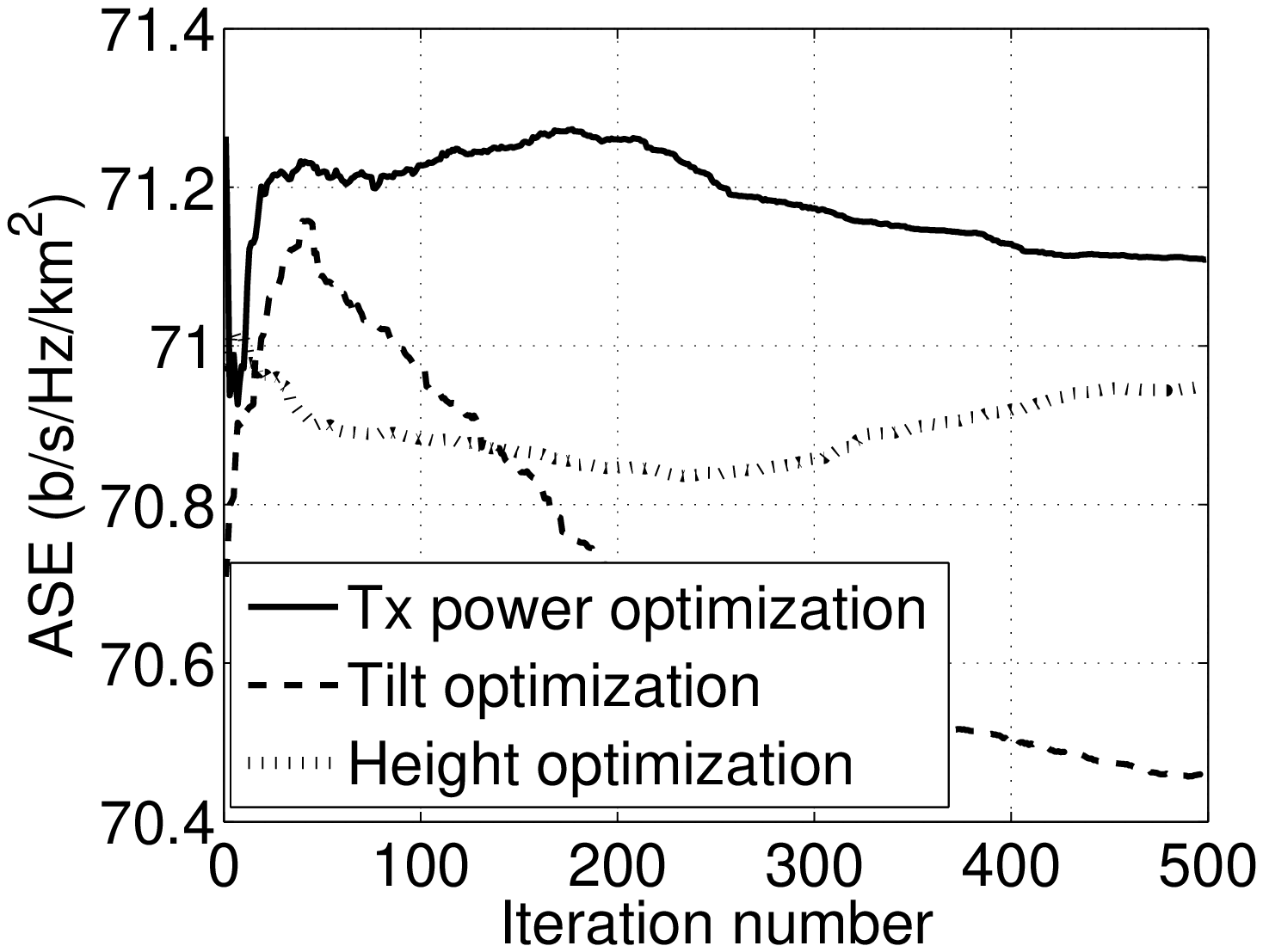}\label{RAN_ASE_TWC}}\hspace{-.385cm} 
\subfigure[]{\includegraphics[scale=.24]{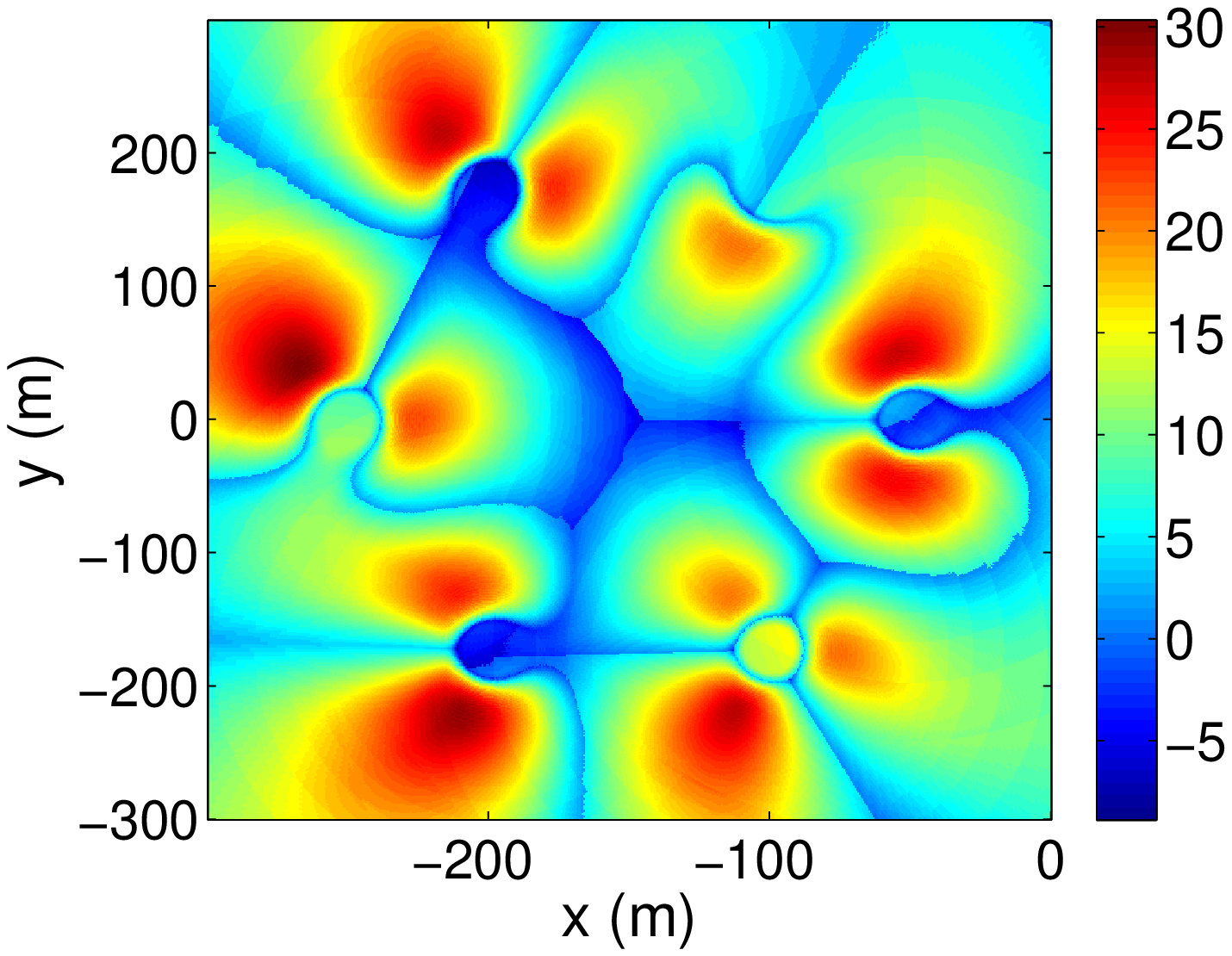}\label{3d_Pattern_mean_12Tx}} \hspace{-.385cm}
\subfigure[]{\includegraphics[scale=.24]{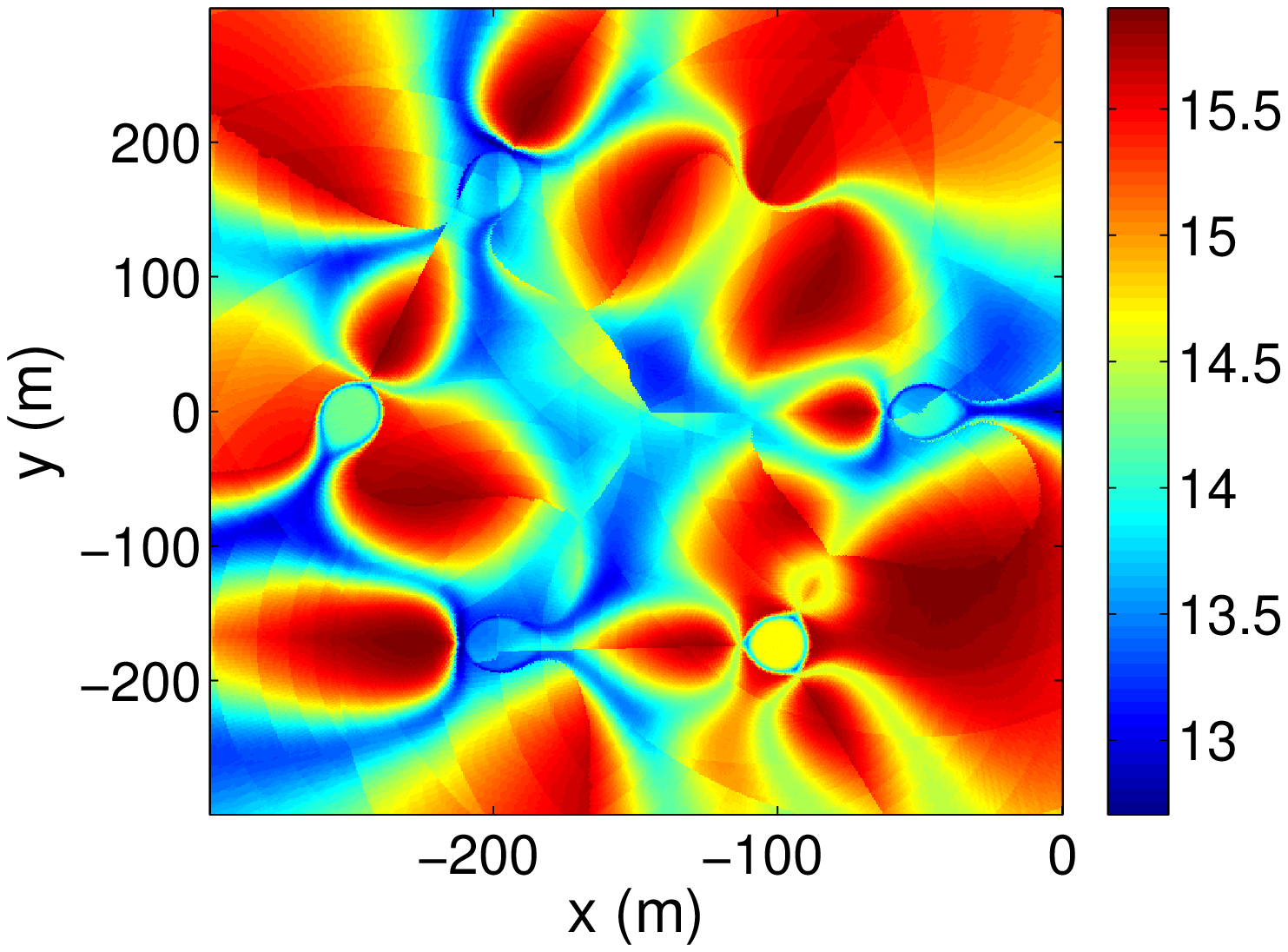}\label{3d_Pattern_variance_12Tx}}
\caption{Optimization of RAN parameters for active sector set: $\mathrm{B}_{\mathrm{on}}=\{5,6,9,10,11,13,14,17,18,19,20,21\}$ for $\rho = 8.683\times10^{-4}$ Erlang/m$^2$ (a) Average APC vs Iteration number (b) Average overlap vs Iteration number (c) Average ASE vs Iteration number (d) Mean of SINR (in dB) with $\boldsymbol{\mathcal{P}} ^*=[43,43,43,41,41,40,42,40,42,40,40,42]$ (dBm), $\mathrm{H}_{tj} = 20$ m, and $\mathrm{\Phi}_{\mathrm{tilt}j}=12$ Deg., $\forall j \in \mathrm{B}_{\mathrm{on}}$ (e) Variance of SINR (in dB$^2$)}
\label{Iteration_RAN}
\end{figure*}

It is seen that the maximum traffic demand density supported by the network when all cells are active is $\rho_{\mathcal{B}}^{\mathrm{max}}=12.5$ Erlang/m$^2$. Figure \ref{APC_ASE_COV} shows the cost functions of the Pareto solution set for different traffic demand density, $\rho$(Erlang/m$^2$) $=3.859\times10^{-4},5.306\times10^{-4},7.2360\times10^{-4}$, $9.1656\times10^{-4},11.09\times10^{-4}$, and $12.06\times10^{-4}$ (i.e. normalized traffic load ($\hat{\rho}=\rho/\rho_{\mathcal{B}}^{\mathrm{max}}$) is $26\%,41\%,56\%,70\%,85\%$ and $92\%$, respectively). It can be seen that as the traffic demand density increases the coverage, APC and ASE of the Pareto solution set also increases as expected. Also the minimum value of the coverage in the Pareto solution set increases with increasing traffic demand due to increased number of active sectors in the corresponding Pareto solution set. Further, it can be seen that the solutions for a particular traffic demand focuses on to a particular region of the search space.


Figure \ref{Cost_functions_different_objectives} shows the impact of selection of solutions with the goal of achieving an individual objective on cost functions and ES. Figure \ref{APC_Final_TWC} shows the APC when the solution is selected with an individual objective for varying traffic demand density. The corresponding ASE, \% overlap, \% coverage, and ES are shown in Figure \ref{ASE_vs_Erlang_TWC}, \ref{Cov_Final_TWC}, \ref{Overlap_Final_TWC}, and \ref{ES_Final_TWC}, respectively. It can be observed from Figure \ref{ES_Final_TWC} that the selection of a solution with the objective of minimizing APC provides maximum ES (i.e. $71\%$) compared to other cases. However, it is achieved at the cost of significant reduction in ASE, and coverage performance. For instance, if the solution is selected with the objective of minimizing APC, then for a traffic demand density of $\rho=6.2\times10^{-4}$ Erlang/m$^2$, the ES is $25\%$ more compared to the case where the solution is selected with the objective of maximizing ASE. However, it comes at the cost of a loss of $27$ b/s/Hz/km$^2$ of ASE and a loss of $1.8\%$ network coverage. The \% overlap increases from $21\%$ to $25\%$ when the solution is selected with the objective of minimizing APC instead of the case where the solution is selected with the objective of maximizing ASE. Further, it can be seen that if the solution is selected with the objective of minimizing overlap, the possible ES is close to the case where the solution is selected with the objective of minimizing APC. The ASE also nearly equal to the case where the solution is selected with the objective of maximizing ASE. Therefore, it can be concluded that the best approach is to select the solution with the objective of minimizing \% overlap. However, there may several possible cases where different trade-offs could be observed if the solution is selected with fixed coverage and overlap 
requirements.

\subsection{Results of RAN Parameter Optimization}
In this section, we show the results of RAN parameter optimization for the active sector set $\mathrm{B}_{\mathrm{on}} = \{5,6,9,10,11,13,14,$$17,18,19,20,21\}$ for the traffic demand density of $\rho = 8.683\times10^{-4}$ Erlang/m$^2$. The available range of transmit power, tilt angle, and height of an individual sector are taken as three units above and three units below the fixed RAN parameter values as given in Table I(a) with step size of 1 Deg., 2 m, and 1 W, respectively (i.e. $P_{tj}\!\! \in\{40,41,42,43,45,46,47\}$, $\phi_{\mathrm{tilt}j} \in \{9,10,11,12,13,14,15\}$, $H_{tj} \in \{17,18,19,20,21,22,23\}$, $\forall j \in \mathrm{B}_{\mathrm{on}}$). In order to see the impact of an individual RAN parameter, we individually optimized the RAN parameters keeping others parameters at a fixed value.

Figures \ref{Iteration_no_vs_APC_RAN}-\ref{RAN_ASE_TWC} show the mean values of the APC, overlap and ASE, respectively of the solutions in the Pareto solution set. It can be seen that the impact of sector transmit power optimization on APC is significant than tilt and height optimization. Further, transmit power optimization decreases the overlap and increases the ASE. Tilt optimization slightly increases the overlap which results in decrease in the ASE performance. However, height optimization decreases the average overlap similar to transmit power optimization. Though the number of solutions in the Pareto solution set increases with iteration number (Figure is not shown here), the average APC converges quickly at around $100$ iterations. Reduction in the APC is not significant as the active sector set itself an optimized solution. Figures \ref{3d_Pattern_mean_12Tx} and \ref{3d_Pattern_variance_12Tx} shows the mean and variance of the SINR at different locations in the network for a sample 
solution when only the transmit power is optimized. It can be observed that the variance is low near the BSs and is high at the edge regions between the sectors. This is because the variance of SINR depends on the strength of the received signal power from an individual sector. As the regions near the BSs receive the strongest signal strength compared to neighbouring sectors which does not increase the variance of SINR. The edge regions receive almost equal power from the neighbouring sectors which results in increased variance. This increased variance leads to increased handover in those regions.

\begin{figure*}[htb]
\begin{center}
\hspace{-.25cm}
\subfigure[]{\includegraphics[scale=.55]{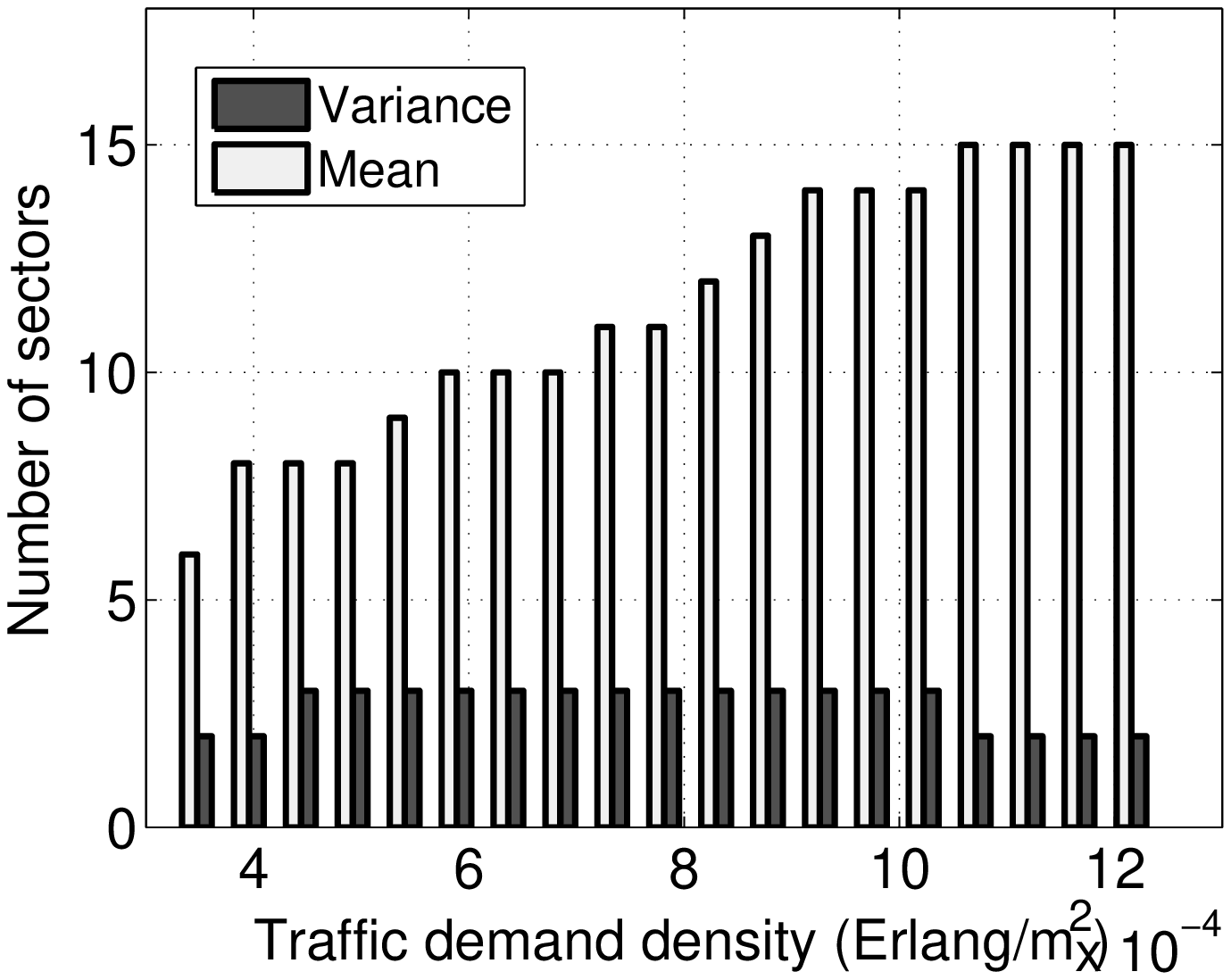}\label{CDF_NBS_Erlang_TWC}} \vspace{-.25cm} 
\subfigure[]{\includegraphics[scale=.55]{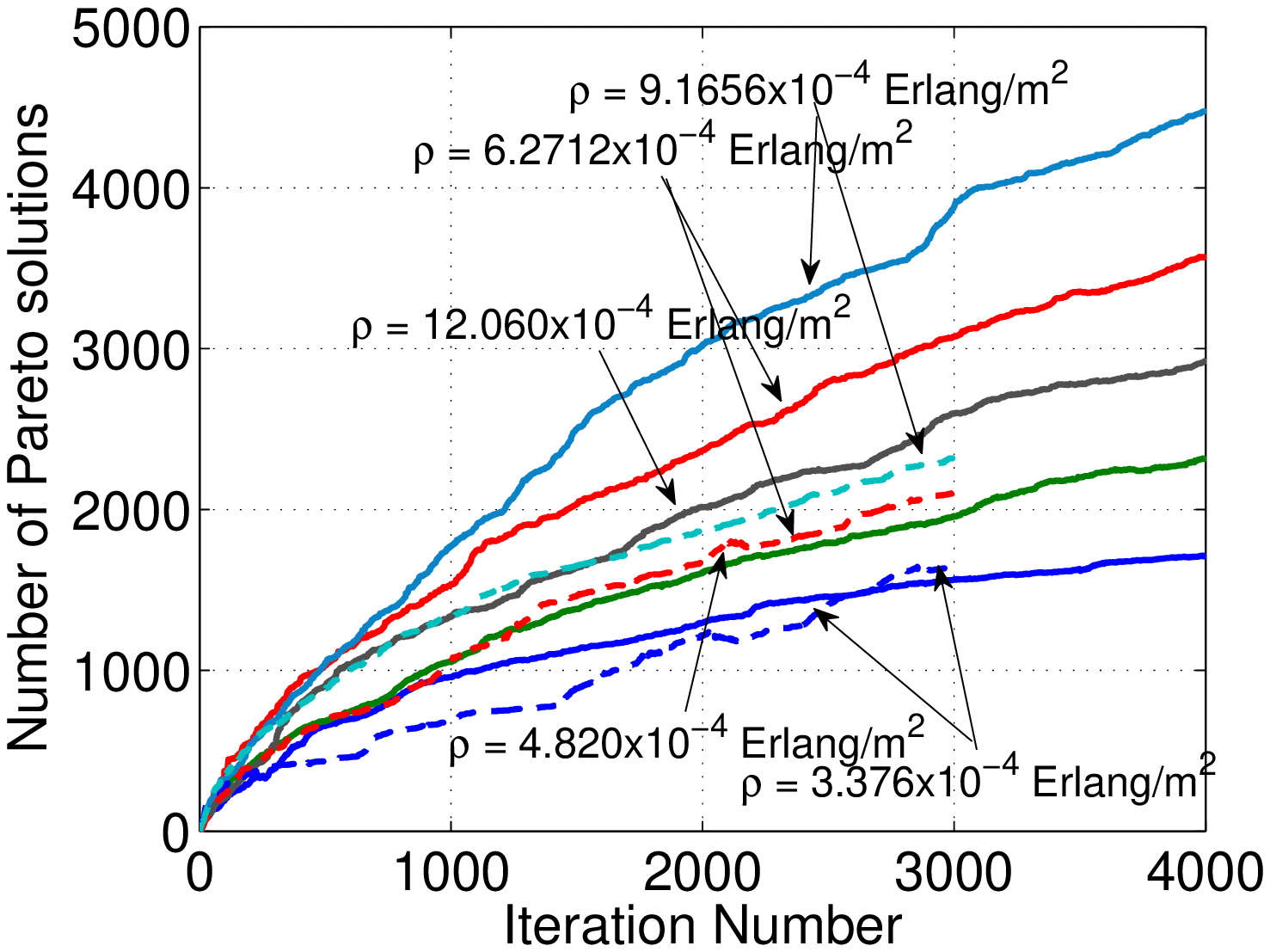}\label{Number_of_Pareto_Solutions_vs_Itretation1}}
\caption{Convergence results (a) Mean and variance of number of active sectors for varying traffic demand density (b) Number of combinations in the search space vs Number of active sectors (c) Number of solutions in the Pareto solution set vs Iteration number}
\label{Iteration_Req_vs_Erlang1}
\end{center}
\end{figure*}

\subsection{Convergence and Complexity Analysis}
The two important performance metrics used for evaluating the performance of MOEAs are \textit{Convergence} and \textit{diversity} \cite{deb2001multi}. Due to space limitations, in this section we present only the convergence results. We evaluate the convergence in terms of the number of solutions in the Pareto solution set i.e. $|\mathcal{Q}_{\mathrm{Pareto}}|$. The convergence of the Pareto solution set depends on the crossover probability $p_{co}$ and the mutation probability $p_{mu}$. Since the considered problem is NP-hard, it is infeasible to say whether the final Pareto set converges to the true Pareto front or not. However, it is verified with very low traffic demand density (i.e. $\rho$) that the algorithm converges to the best solution set which is same as the one obtained using exhaustive search when $p_{co}=.7$ and $p_{mu}=.01$ for $N_{\mathrm{Init}}= 100$. Therefore, the results presented in this section 
are based on the above GA parameters.



Figure \ref{CDF_NBS_Erlang_TWC} shows the mean and variance of number of sectors in the set of active sectors in the Pareto solution set obtained after $4000$-th generation for varying traffic demand density. 
The mean, 
\begin{equation} 
\bar{N}_{\mathcal{B}_{\mathrm{on}}} = \frac{1}{|\mathcal{Q}_{\mathrm{Pareto}}|}\sum_{n=1}^{|\mathcal{Q}_{\mathrm{Pareto}}|} |\mathcal{B}_{\mathrm{on}}^n|\nonumber 
\end{equation} 
and the variance, 
\begin{equation} 
\sigma_{N_{\mathcal{B}_{\mathrm{on}}}}^2 =[ \frac{1}{|\mathcal{Q}_{\mathrm{Pareto}}|}\sum_{n=1}^{|\mathcal{Q}_{\mathrm{Pareto}}|} |\mathcal{B}_{\mathrm{on}}^n|^2 ]-\bar{N}_{\mathcal{B}_{\mathrm{on}}}^2. \nonumber 
\end{equation}

It can be seen that as the traffic demand density increases the number of sectors in the Pareto solution set also increases as expected. It is clear from Figures \ref{APC_ASE_COV} and \ref{CDF_NBS_Erlang_TWC} that for a given traffic demand density, the Pareto optimal solutions always converge towards the region of the search space where the number of sectors in the active sector 
set is equal to the number of sectors required to serve that traffic demand. Hence, for a given traffic load, the solutions can be obtained by searching only within a sub-population such that the number of sectors in the set of active sectors the sub-population is equal to the number of sectors that is required to provide that traffic demand.

The solid and dotted lines in Figure \ref{Number_of_Pareto_Solutions_vs_Itretation1} shows the number of solutions in the Pareto solution set $|Q_{\mathrm{Pareto}}|$ vs iteration number for different traffic demand density when the search is performed in the entire search space $\mathcal{T}_{\mathcal{B}}$ and the reduced search space $\mathcal{T}_{\mathcal{B}_{\hat{\rho}_q}}$, respectively. From the figure it can be observed that the number of solutions in the Pareto solution set increases with iteration number and it seems converges quickly during low and high load conditions compared to medium load conditions. It can be observed that the number of solutions in the Pareto solution set converges very quickly when the search is performed over the reduced search space.

\section{Conclusions}\label{conclusion}
We have proposed a dynamic multi-objective optimization framework for OFDMA based cellular systems to identify the set of active sectors and RAN parameters which are required to serve a given traffic demand with minimum area power consumption while maintaining a suitable  trade-offs between other metrics like network coverage, area spectral efficiency, and overlap. 

A Genetic Algorithm based solution approach is proposed to obtain the Pareto optimal solution for the considered multi-objective optimization. It is seen that the maximum energy saving is achieved when the solution (i.e. set of active sectors and RAN parameters) are selected with the objective of minimizing area power consumption. However, it is associated with a reduced area spectral efficiency and also affects the coverage performance. On the other hand, if the solution is selected with the objective of minimizing overlap then the area spectral efficiency and coverage performance improves significantly even though the energy saving is not as high as the former case.

It is further seen that when the RAN parameters are optimized along with the optimized active sector set, the transmit power optimization provides more ES compared to tilt and height optimization. By selecting the solutions based on different objectives, at $20\%$, $50\%$, and $80\%$ network traffic load, it is possible to achieve ES between $60\%$ to $85\%$, $30\%$ to $60\%$, and $10\%$ to $40\%$, respectively. The results show that the proposed framework can be used in practical networks in a centralized manner to adaptively change the network configuration at a faster convergence rate and minimal computational complexity. 

\appendix

\subsection{Derivation of PDF of cumulative interference (Equation \eqref{intf_dist})}
The MGF of sum of binary weighted $N_{\mathcal{B}_{\mathrm{on}}}-1$ independent RVs can be written as the product of MGF of individual RVs as
\begin{equation}
\mathcal{M}_{P_{I_i}}(s) = \prod_{g=1}^{N_{\mathcal{B}_{\mathrm{on}}}-1} \big[\beta_g \mathcal{M}_{P_{r,ig}}(s) +(1-\beta_g)\big].
\end{equation}

Using the definition of MGF i.e. $\mathcal{M}_X(s) = \int_{0}^{\infty}e^{-sx} p_X(x) \ dx$ the MGF of a log-normal-Gamma RVs $P_{r,ig}$ can be written as 
\begin{align}
& \mathcal{M}_{P_{r,ig}}(s)\nonumber \\
& =\int_0^\infty\! e^{-sx} \!\! \int_0^\infty \!\! \bigg(\frac{m}{y}\bigg)^m \frac{x^{m-1}e^{ -\frac{mx}{y}}}{\Gamma(m)} \! \frac{ e^{-\frac{(\ln y-\mu_{\hat{\chi}_{ig}})^2}{2\sigma_{\hat{\chi}_{ig}}^2}}}{\sqrt{2\pi} y \sigma_{\hat{\chi}_{ig}}} dy \ dx.
\end{align}

By letting $t=\frac{(\ln y-\mu_{\hat{\chi}_{ig}})}{\sqrt{2}\sigma_{\hat{\chi}_{ig}}}$, the above expression can be rewritten as 
\begin{align}
&\mathcal{M}_{P_{r,ig}}(s)\nonumber\\& =\int_0^\infty\! \!\!  \!\! \int_0^{\infty}\! \!  \frac{m^m x^{m-1}}{\sqrt{\pi}\Gamma(m)}\!\! \frac{e^{\big[\!\!-mxe^{\!\!-(\sqrt{2}\sigma_{\hat{\chi}_{ig}} t+\mu_{\hat{\chi}_{ig}})}{\big(\!\!1+\frac{se^{(\sqrt{2}\sigma_{\hat{\chi}_{ig}}t+\mu_{\hat{\chi}_{ig}})}}{m}\!\!\big)}\!\!\big]}e^{-t^2}}{e^{m(\sqrt{2}\sigma_y t+\mu_{\hat{\chi}_{ig}})}} dx dt.
\end{align}

By letting \begin{equation}
a=m\exp{\bigg(-(\sqrt{2}\sigma_{\hat{\chi}_{ig}} t+\mu_{\hat{\chi}_{ig}})\bigg)}{\bigg(1+\frac{s\exp{\big(\sqrt{2}\sigma_{\hat{\chi}_{ig}}t+\mu_{\hat{\chi}_{ig}}\big)}}{m}\bigg)}\nonumber\end{equation} and rearranging the terms we get
\begin{align}
&\mathcal{M}_{P_{r,ig}}(s)\nonumber \\&=\frac{1}{\sqrt{\pi}}\!  \int_0^\infty\! \! \frac{m^m}{a^m\! \Gamma(m)}\! \bigg[a^m\!\int_0^{\infty}\!x^{m-1}e^{-ax}  dx\bigg]\frac{ e^{-t^2}}{e^{m(\sqrt{2}\sigma_{\hat{\chi}_{ig}} t+\mu_{\hat{\chi}_{ig}})}}\ dt.
\end{align}

By applying the property of Gamma function i.e $\Gamma(m) = a^m\!\int_0^{\infty}x^{m-1}e^{-ax} \ dx$, the above integral can be written as
\begin{align}
&\mathcal{M}_{P_{r,ig}}(s)\nonumber \\&=  \frac{1}{\sqrt{\pi}} \int_0^\infty \frac{m^me^{m(\sqrt{2}\sigma_{\hat{\chi}_{ig}} t+\mu_y)}e^{-t^2}}{\bigg[me^{-(\sqrt{2}\sigma_y t+\mu_{\hat{\chi}_{ig}})}{\bigg(1+\frac{se^{(\sqrt{2}\sigma_{\hat{\chi}_{ig}}t+\mu_{\hat{\chi}_{ig}})}}{m}\bigg)}\bigg]^m}\ dt \nonumber \\
 &= \frac{1}{\sqrt{\pi}} \int_0^\infty \bigg(1+\frac{se^{(\sqrt{2}\sigma_{\hat{\chi}_{ig}}t+\mu_{\hat{\chi}_{ig}})}}{m}\bigg)^{-m}e^{-t^2}\ dt. \nonumber
\end{align}

The above integral can be rewritten using Gauss-Hermite series as
\begin{align}
\mathcal{M}_{P_{r,ig}}&(s;\mu_{P_{r,ig}},\sigma_{P_{r,ig}}) \nonumber\\= &\sum_{n=1}^{N_H}\! \frac{w_n}{\sqrt{\pi}}\bigg(\!\!1\!+\!\frac{s\exp{\big(\sqrt{2}\sigma_{\hat{\chi}_{ig}}t_n+\mu_{\hat{\chi}_{ig}}\big)}}{m}\!\!\bigg)^{-m}\!\!+\!R_{N_H},
\end{align}
where $N_H$ is the order of the Hermite integration, $R_N$ is reminder term. The weights, $w_n$ and abscissas, $a_n$ for $N_H$ are given in \cite[Tab. 25.10]{abramowitz1974}.

The mean and variance of $P_{I_i}$ can be obtained by solving the following two equations 

\begin{align}
 &\mathcal{M}_{P_{I_i}}(s_i;\mu_X,\sigma_X)\nonumber \\ &= \prod_{g=1}^{N_{\mathcal{B}_{\mathrm{on}}}-1} \!\bigg[\beta_g \mathcal{M}_{P_{r,ig}}(s_i;\mu_{P_{r,ig}},\sigma_{P_{r,ig}}) +(1-\beta_g)\bigg]\nonumber \\&= \prod_{g=1}^{N_{\mathcal{B}_{\mathrm{on}}}-1}\!  \Bigg[\!\beta_g\sum_{n=1}^{N_H} \frac{w_n}{\sqrt{\pi}}\bigg(1+\frac{s_ie^{(\sqrt{2}\sigma_{\hat{\chi}_{ig}}t_n+\mu_{\hat{\chi}_{ig}})}}{m}\bigg)^{-m}\!+\!(1-\beta_g)\!\Bigg].
\end{align}
The values of $\mu_X$ and $\sigma_X$ can be easily obtained by numerically solving the above non-linear equations at different real and positive values of $s$, namely $s_1$ and $s_2$.

\subsection{Derivation of Equation \eqref{cov_prob_location}}
\begin{align}\label{cov}
Pr&(\gamma_{ij} \ge \Gamma_{\mathrm{min}},P_{r,ij}\ge P_{r,\mathrm{min}})\nonumber \\
=&Pr\bigg(\frac{P_{r,ij}}{\Gamma_{\mathrm{min}}} \ge P_{I_i} ,P_{r,ij}\ge P_{r,\mathrm{min}}\bigg)\nonumber \\=&\!\int_{P_{r,\mathrm{min}}}^{\infty}\int_0^{Pr/\Gamma_{\mathrm{min}}} p_{P_{I_i}}(I) \ dI\  p_{P_{r,ij}}(P_r) \ dP_r\nonumber \\ =&\!\int_{P_{r,\mathrm{min}}}^{\infty} \bigg[1-Q\bigg(\!\frac{\ln(P_{r}/\Gamma_{\mathrm{min}})-\!\mu_{P_{I_i}}}{\sigma_{P_{I_i}}}\!\!\bigg)\bigg]p_{P_{r,ij}}(P_r)\ dP_r \nonumber \end{align}
\begin{align} =&\int_{P_{r,\mathrm{min}}}^{\infty}\! p_{P_{r,ij}}(P_r)\ dP_r\nonumber \\ 
&-\int_{P_{r,\mathrm{min}}}^{\infty}\! Q\bigg(\!\frac{\ln(P_{r}/\Gamma_{\mathrm{min}})-\!\mu_{P_{I_i}}}{\sigma_{P_{I_i}}}\!\!\bigg)p_{P_{r,ij}}(P_r) dP_r  \nonumber \\ =Q&\bigg(\!\frac{\ln(P_{r,\mathrm{min}})-\!\mu_{P_{r,ij}}}{\sigma_{P_{r,ij}}}\!\bigg)
\nonumber \\&-\int_{P_{r,\mathrm{min}}}^{\infty}\!\! Q\bigg(\!\frac{\ln(P_{r}/\Gamma_{\mathrm{min}})-\!\mu_{P_{I_i}}}{\sigma_{P_{I_i}}}\!\!\bigg)p_{P_{r,ij}}(P_r) dP_r, 
\end{align}
where $Q(x) = \frac{1}{\sqrt{2\pi}} \int_{x}^{\infty}e^{-\frac{t^2}{2}} dt$.

By applying Chernoff bound to the $Q$-function  (i.e. $Q(x) = e^{-\frac{x^2}{2}}$), \eqref{cov} can be approximated as 
\begin{align}
\hat{Pr}&(\gamma_{ij} \ge \Gamma_{\mathrm{min}},P_{r,ij}\ge P_{r,\mathrm{min}})\nonumber \\ 
&=\exp{\bigg(-\frac{(\ln(P_{r,\mathrm{min}})-\mu_{P_{r,ij}})^2}{2\sigma_{P_{r,ij}}^2}\bigg)}\nonumber \\ &-\int_{P_{r,\mathrm{min}}}^{\infty} \!\!\!e^{-\frac{(\ln(P_{r}/\Gamma_{\mathrm{min}})-\mu_{P_{I_i}})^2}{2\sigma_{P_{I_i}}^2}}\frac{e^{\frac{-(\ln(Pr)-\mu_{P_{r,ij}})^2}{2\sigma_{P_{r,ij}}^2}}}{ \sqrt{2\pi}Pr \sigma_{P_{r,ij}}} \ dP_r. \nonumber 
\end{align}\normalsize

After some manipulations
\begin{align} 
\hat{Pr}&(\gamma_{ij} \ge \Gamma_{\mathrm{min}},P_{r,ij}\ge P_{r,\mathrm{min}}) \nonumber \\ 
=&\exp{\bigg(-\frac{(\ln(P_{r,\mathrm{min}})-\mu_{P_{r,ij}})^2}{2\sigma_{P_{r,ij}}^2}\bigg)}\nonumber \\ &-\frac{\sigma_{P_{I_i}}}{\sqrt{\sigma_{P_{I_i}}^2+\sigma_{P_{r,ij}}^2}}\!\int_{P_{r,\mathrm{min}}}^{\infty} \frac{\exp{\bigg(E-\frac{\big(\ln(\frac{Pr}{\Gamma_{\mathrm{min}}})-\mu_c\big)^2}{2\sigma_c^2}\bigg)}}{\sqrt{2\pi}Pr \sigma_c  }  dP_r, \nonumber
\end{align}
where $E=$
\begin{equation}
\frac{-(\mu_{P_{I_i}}-\mu_{P_{r,ij}})^2-(\mu_{P_{I_i}}+\ln \Gamma_{\mathrm{min}})^2+\mu_{P_{I_i}}^2(1+\frac{2\mu_{P_{r,ij}}\ln \Gamma_{\mathrm{min}}}{\sigma_{P_{I_i}}^2})}{(\sigma_{P_{I_i}}^2+\sigma_{P_{r,ij}}^2)}, \nonumber 
\end{equation}
\begin{equation}
\mu_c=\frac{(\mu_{P_{r,ij}}-\ln \Gamma_{\mathrm{min}})\sigma_{P_{I_i}}^2+\mu_{P_{I_i}}\sigma_{P_{r,ij}}^2}{\sigma_{P_{I_i}}^2+\sigma_{P_{r,ij}}^2}, \nonumber 
\end{equation}
and \begin{equation}
\sigma_c=\frac{\sigma_{P_{I_i}}\sigma_{P_{r,ij}}}{\sqrt{\sigma_{P_{I_i}}^2+\sigma_{P_{r,ij}}^2}}.\nonumber 
\end{equation}

 The integral in the second term is a $Q$-function. By applying Chernoff bound to it the coverage probability can be obtained as 
\begin{align}\label{cov_prob}
\hat{Pr}&(\gamma_{ij} \ge \Gamma_{\mathrm{min}},P_{r,ij}\ge P_{r,\mathrm{min}}) \nonumber \\ 
=&\exp{\bigg(-\frac{(\ln(P_{r,\mathrm{min}})-\mu_{P_{r,ij}})^2}{2\sigma_{P_{r,ij}}^2}\bigg)}\nonumber \\ 
&-\frac{\sigma_{P_{I_i}}\exp{\bigg(E-\big(\frac{\ln(\frac{P_{r,\mathrm{min}}}{\Gamma_{\mathrm{min}}})-\mu_c}{\sqrt{2}{\sigma_c}}\big)^2\bigg)}}{\sqrt{\sigma_{P_{I_i}}^2+\sigma_{P_{r,ij}}^2}}.
\end{align}\normalsize

\subsection{Proof of Theorem 1}

Let $Q_l(t)$ be the random variable represents the number of class-$l$ objects in the queue at time $t$. Let 
\begin{equation}
\boldsymbol{Q}(t):=(Q_1(t),...,Q_l(t),...,Q_L(t))\end{equation}
be the state of the queue at time $t$ and $\{\boldsymbol{Q}(t)\}$ be the corresponding stationary stochastic process. Let $\pi_e(\boldsymbol{n_{uj}})$ be the probability that the queue is in state $\boldsymbol{n_{uj}}$ in equilibrium. In the long run, all the users in the system generates traffic demand belongs to different classes ($l=1,2,...,N_L$). The traffic demand generated by the users belongs to $l$-th class can be written as 
\begin{equation}
\rho_{j}(l) =  w_j(l) \  \rho_{j}.\nonumber
\end{equation}
Then the equilibrium distribution of number of users is given as \cite[Theorem: 2.1]{ross1995} 
\begin{equation}
\pi_e(\boldsymbol{n_{uj}}) =  \frac{\prod\limits_{l=1}^{N_L} \frac{(\rho_{j}(l))^{n_{uj}(l) }}{n_{uj}(l)!}}{{\sum\limits_{\boldsymbol{n_{uj}}\in \boldsymbol{\mathcal{S}}}^{}}\prod\limits_{l=1}^{N_L}\frac{(\rho_{j}(l))^{n_{uj}(l)}}{n_{uj}(l)!}}, \ \boldsymbol{n_{uj}}\in \boldsymbol{\mathcal{S}}. \nonumber
\end{equation}

The blocking probability of calls is obtained using Kaufman-Roberts Algorithm (KRA) \cite{karray2010,kaufman1981}. Let 
\begin{equation}
\boldsymbol{\mathcal{S}}(c):=\{ \boldsymbol{n_{uj}}\in \boldsymbol{\mathcal{S}}: \boldsymbol{n_{uj}}.\boldsymbol{n_{sc}}=c\}\nonumber 
\end{equation} be the set of states for which exactly $c$ number of sub-channels are occupied. Let $g_j(c)$ be the probability that there are \begin{equation}
c=\sum_{l=1}^{N_L}n_{uj}(l) n_{sc}(l)\end{equation} number of sub-channels are occupied i.e. 
\begin{equation}
g_j(c)= \sum_{\boldsymbol{n_{uj}} \in \boldsymbol{\mathcal{S}}(c)} \pi_e(\boldsymbol{n_{uj}}).\nonumber 
\end{equation}
The occupancy probabilities satisfy the following recursive equations \cite{kaufman1981} 

\begin{align}\label{occupancy_probability_1}
g_j(c) = &\frac{1}{c} \sum_{l=1}^{N_L}\rho_j(l) . n_{sc}(l). g_j(c-n_{sc}(l)), \ c=0,...,N_{sc},\nonumber \\ & \sum_{c=0}^{N_{sc}}g_j(c)=1.
\end{align}
The blocking probability of $l$-th class in cell $j$ can be obtained by 
\begin{equation}
P_{bj}(l) = \frac{1}{\sum_{c=0}^{N_{sc}} g_j(c)}\sum_{c=N_{sc}-n_{sc}(l)+1}^{N_{sc}} g_j(c),\ l=1,...,N_L.\nonumber
\end{equation}
 The average blocking probability in $j$-th cell is obtained as the weighted sum of all class-wise blocking probabilities i.e. 
\begin{equation}
P_{bj} = \sum_{l=1}^{N_L}{w_j(l)P_{bj}(l)}.
\end{equation}
Hence Eqn. \eqref{average_bp_cell}.

\subsection{Proof of Lemma 1}

Since there is no closed form expression available for sum of interference components in case of log-normal fading, we prove the lemma considering a single interferer. The fraction of users belong to $l$-th class attached to sector $j$ is obtained as 
	\begin{align}\label{cov}
	w_j(l)=&Pr(\Gamma_{l+1}\ge \gamma_{ij} \ge \Gamma_{l},P_{r,ij}\ge P_{r,\mathrm{min}})\nonumber \\ =& Pr(v_2==0) Pr\bigg(\frac{P_{r,ij}}{\Gamma_{l}} \ge P_{I_i(v_2==0)} \ge \frac{P_{r,ij}}{\Gamma_{l+1}}, P_{r,ij}\ge P_{r,\mathrm{min}}\bigg) \nonumber \\ &+ Pr(v_2==1)  Pr\bigg(\frac{P_{r,ij}}{\Gamma_{l}} \ge P_{I_i(v_2==1)} \ge \frac{P_{r,ij}}{\Gamma_{l+1}}, P_{r,ij}\ge P_{r,\mathrm{min}}\bigg)\nonumber	\end{align}
		\begin{align}
	=&(1-\beta_2) \int_{P_{r,\mathrm{min}}}^{\infty}\int_{Pr/\Gamma_{l+1}}^{Pr/\Gamma_{l}} p_{P_{I_i(v_2==0)}}(I) \ dI\  p_{P_{r,ij}}(P_r) \ dP_r \nonumber \\ &+ \beta_2\int_{P_{r,\mathrm{min}}}^{\infty}\int_{Pr/\Gamma_{l+1}}^{Pr/\Gamma_{l}} p_{P_{I_i(v_2==1)}}(I) \ dI\  p_{P_{r,ij}}(P_r) \ dP_r\nonumber \\= &(1-\beta_2)\int_{P_{r,\mathrm{min}}}^{\infty} \bigg[Q\bigg(\frac{\ln(P_{r}/\Gamma_{l+1})-P_N}{0}\bigg)\nonumber \\ &\hspace{3cm}-Q\bigg(\frac{\ln(P_{r}/\Gamma_{l})-P_N}{0}\bigg)\bigg]p_{P_{r,ij}}(P_r)\ dP_r \nonumber \\ &+\beta_2\int_{P_{r,\mathrm{min}}}^{\infty} \bigg[Q\bigg(\frac{\ln(P_{r}/\Gamma_{l+1})-\mu_{P_{I_i}}}{\sigma_{P_{I_i}}}\bigg)\nonumber \\ &\hspace{3cm}-Q\bigg(\frac{\ln(P_{r}/\Gamma_{l})-\mu_{P_{I_i}}}{\sigma_{P_{I_i}}}\bigg)\bigg]p_{P_{r,ij}}(P_r)\ dP_r \nonumber
\end{align}
\begin{align}	=&0+\beta_2\int_{P_{r,\mathrm{min}}}^{\infty} \bigg[\exp{\bigg(\frac{-(\ln(P_{r}/\Gamma_{l+1})-\mu_{P_{I_i}})^2}{2\sigma_{P_{I_i}}^2}\bigg)}\nonumber \\&\hspace{2cm}-\exp{\bigg(\frac{-(\ln(P_{r}/\Gamma_{l})-\mu_{P_{I_i}})^2}{2\sigma_{P_{I_i}}^2}\bigg)}\bigg]p_{P_{r,ij}}(P_r)\ dP_r.
	\end{align}
	
After applying Chernoff bound, 
\begin{align}\label{fraction_l}
w_j(l)=&\frac{\beta_2}{\sqrt{2}}\bigg[\exp{\bigg(E2-\frac{\big(\ln(\frac{P_{r,\mathrm{min}}}{\Gamma_{l+1}})-{\frac{\mu_{P_{r,ij}}+\mu_{P_{I_i}}-\ln \Gamma_{l+1}}{2}}\big)^2}{2\sigma_{P_{r,ij}}^2}\bigg)} \nonumber \\
	&- \exp{\Bigg(E3-\frac{\big(\ln(\frac{P_{r,\mathrm{min}}}{\Gamma_{l}})-{\frac{\mu_{P_{r,ij}}+\mu_{P_{I_i}}-\ln \Gamma_{l}}{2}}\big)^2}{2\sigma_{P_{r,ij}}^2}\Bigg)}\bigg]
	\end{align}
Here 
\begin{equation}
E2=\frac{-(\mu_{P_{I_i}}-\mu_{P_{r,ij}})^2-(\mu_{P_{I_i}}+\ln \Gamma_{l+1})^2+\mu_{P_{I_i}}^2(1+\frac{2\mu_{P_{r,ij}}\ln \Gamma_{l+1}}{\sigma_{P_{I_i}}^2})}{(\sigma_{P_{I_i}}^2+\sigma_{P_{r,ij}}^2)}\nonumber
\end{equation}
and 
\begin{equation}
E3=\frac{-\big((\mu_{P_{I_i}}-\mu_{P_{r,ij}})^2+(\mu_{P_{I_i}}+\ln \Gamma_{l})^2\big)+\mu_{P_{I_i}}^2(1+\frac{2\mu_{P_{r,ij}}\ln \Gamma_{l}}{\sigma_{P_{I_i}}^2})}{(\sigma_{P_{I_i}}^2+\sigma_{P_{r,ij}}^2)}.\nonumber
\end{equation}

By letting 
\begin{equation}
X(\Gamma_{l+1}) = \frac{1}{\sqrt{2}}\exp{\bigg(E2-\frac{\big(\ln(\frac{P_{r,\mathrm{min}}}{\Gamma_{l+1}})-{\frac{\mu_{P_{r,ij}}+\mu_{P_{I_i}}-\ln \Gamma_{l+1}}{2}}\big)^2}{2\sigma_{P_{r,ij}}^2}\bigg)},\nonumber
\end{equation}
and 
\begin{equation}
X(\Gamma_{l})= \frac{1}{\sqrt{2}} \exp{\bigg(E3-\frac{\big(\ln(\frac{P_{r,\mathrm{min}}}{\Gamma_{l}})-{\frac{\mu_{P_{r,ij}}+\mu_{P_{I_i}}-\ln \Gamma_{l}}{2}}\big)^2}{2\sigma_{P_{r,ij}}^2}\bigg)},\nonumber
\end{equation}
Eqn. \eqref{fraction_l} is rewritten as
\begin{equation}
w_j(l) = \beta_2 [X(\Gamma_{l+1})-X(\Gamma_l)].  
\end{equation}
	
Consider a simple scenario where $N_L=2, \ N_{\mathrm{sc}} = 2, \ \boldsymbol{n_{sc}} = (1,1)$. Then from Eqn. \eqref{occupancy_probability}
\begin{align} 
\beta_{j} = &\frac{1}{2}\frac{\sum_{c=0}^{2} g_j(c)}{\sum_{c=0}^{2}\frac{1}{c} g_j(c)} \nonumber \\= &\frac{1}{2}\frac{\sum_{c=0}^{2} \sum_{l=1}^{2} \rho_{j}. w_j(l).  g_j(c-n_{sc}(l))}{\sum_{c=0}^{2} \frac{1}{c}\sum_{l=1}^{2}\rho_{j}. w_j(l). g_j(c-n_{sc}(l))}.
\end{align}

After expanding the terms we get
\begin{equation} 
\beta_{j} =  \frac{1}{2}\frac{[1+(w_j(1)+w_j(2))]}{[1+\frac{1}{2}(w_j(1)+w_j(2))]}
\end{equation}\normalsize
	
After substituting $w_j(1) = \beta_2 [X(\Gamma_{2})-X(\Gamma_1)]$, $w_j(2) = \beta_2 [X(\Gamma_{\mathrm{max}})-X(\Gamma_2)]$ we get
\begin{equation} 
\beta_{j} =  \frac{1+\beta_2[X(\Gamma_{\mathrm{max}})-X(\Gamma_1)]}{2+\beta_2[X(\Gamma_{\mathrm{max}})-X(\Gamma_1)]}
\end{equation}

After taking derivative of $\beta_j$ with respect to $\beta_2$ we get
\begin{equation}
\frac{\partial \beta_j}{\partial \beta_2} = \frac{[X(\Gamma_{\mathrm{max}})-X(\Gamma_1)]}{\big(2+\beta_2[X(\Gamma_{\mathrm{max}})-X(\Gamma_1)]\big)^2}
\end{equation}

Let $\mu_{P_{r,ij}} = 0, \ \mu_{P_{I_i}} = 0, \ \sigma_{P_{r,ij}}^2 = \sigma_{P_{r,ij}}^2 = \sigma^2$. Then, 
\begin{align}
X&(\Gamma_{\mathrm{max}})\nonumber\\&=\frac{1}{\sqrt{2}}\exp{\bigg(\frac{-\frac{5}{4}(\ln \Gamma_{\mathrm{max}})^2-(\ln P_{r,\mathrm{min}})^2+\ln (P_{r,\mathrm{min}})\ln \Gamma_{\mathrm{max}}}{2\sigma^2}\bigg)}\nonumber\end{align}
and
\begin{equation}
X(\Gamma_{1})=\frac{1}{\sqrt{2}}\exp{\bigg(\frac{-\frac{5}{4}(\ln \Gamma_{1})^2-(\ln P_{r,\mathrm{min}})^2+\ln (P_{r,\mathrm{min}})\ln \Gamma_{1}}{2\sigma^2}\bigg)}.\nonumber\end{equation}

It can be seen that $X$ is increasing function of $\Gamma$. Since $\Gamma_{\mathrm{max}}>\Gamma_1$, $X(\Gamma_{\mathrm{max}}) > X(\Gamma_1)$. 	
	
After taking second derivative we get
\begin{equation}
\frac{\partial^2 \beta_j}{\partial \beta_2^2} =\frac{-2[X(\Gamma_{\mathrm{max}})-X(\Gamma_1)]^2}{\big(2+\beta_2[X(\Gamma_{\mathrm{max}})-X(\Gamma_1)]\big)^3},
\end{equation}
where $(X(\Gamma_1)-X(\Gamma_{\mathrm{max}}))\in [0,1], \beta_2 \in [0,1]$. Therefore $\beta_2(X(\Gamma_1)-X(\Gamma_{\mathrm{max}}))<2$ and hence the denominator term is always positive for $\beta_2 \in [0,1]$. Then it follows that $\frac{\partial^2 \beta_j}{\partial \beta_2^2}< 0$. Therefore, it can be concluded that $\beta_j$ is concave downward on $\beta_2$.

\ifCLASSOPTIONcaptionsoff
  \newpage
\fi
\section*{Acknowledgment}
The authors are grateful to Basabdatta Palit (G. S. Sanyal School of Telecommunications, Indian Institute of Technology Kharagpur) for useful discussions, and to anonymous reviewers whose comments helped to improve the manuscript significantly. 


\end{document}